\documentclass[12pt]{article}
\usepackage[top=1in, left=1in, right=1in, bottom=1in]{geometry}	
\usepackage[parfill]{parskip}	
\usepackage[export]{adjustbox}
\usepackage{graphicx, xcolor}
\usepackage{amssymb}
\usepackage{mathrsfs}
\usepackage{enumitem}
\usepackage{epstopdf}
\usepackage{bbm}
\usepackage{amssymb}
\usepackage{amsmath}
\usepackage{amsfonts}
\usepackage{bm, bbm}

\usepackage{lscape}
\usepackage{rotating}
\usepackage{setspace}
\usepackage{threeparttable}
\usepackage{booktabs}
\usepackage{floatrow}
\usepackage{multirow}
\usepackage{caption}
\usepackage{subcaption}
\usepackage[compact]{titlesec}
\usepackage{lmodern}
\usepackage{comment}
\usepackage{floatrow}
\usepackage{bookmark}

\usepackage{tikz}

\usepackage[ruled,vlined]{algorithm2e}
\SetAlCapHSkip{0pt}       
\SetAlgoNlRelativeSize{0} 
\SetNlSty{}{}{}           

\fontfamily{lmtt}\selectfont
\usepackage[T1]{fontenc}
\usepackage{natbib}
\bibpunct{(}{)}{;}{a}{}{,}
\usepackage{hyperref}
\usepackage[capitalize]{cleveref}
\usepackage{titling}
\newcommand{\subtitle}[1]{%
  \posttitle{%
    \par\end{center}
    \begin{center}\large#1\end{center}
    \vskip0.5em}%
}

\usepackage{amsthm}

\newtheorem{theorem}{Theorem}
\newtheorem{lemma}{Lemma}

\newtheorem{proposition}{Proposition}
\newtheorem{assumption}{Assumption}

\def\E{{\mathbb {E}}}

\def\X{{\mathcal{X}}}

\def\N{{\mathcal{N}}}

\def\T{\mathrm{T}}
\def\A{{\mathcal{A}}}

\def\N{{\mathcal{N}}}

\def\sumi{\sum_{i=1}^{m}}

\def\T{\text{T}}

\def\ols{\mathrm{ols}}
\def\heur{\mathrm{heur}}
\def\plugin{\mathrm{plug-in}}
\def\var{\mathrm{var}}
\def\supp{\mathrm{Supp}}
\def\tr{\mathrm{trace}}
\def\all{\mathrm{all}}
\def\id{\mathrm{ID}}
\def\ad{\mathrm{AD}}
\def\inv{{\mathrm{inv}}}
\def\aug{{\mathrm{aug}}}
\def\uri{{\mathrm{uri}}}

\def\aug{\text{aug}}

\def\gf{\text{g}}
\def\lmt{\texttt{lm}} 

\newcommand{\pb}{{*+}}
\newcommand{\pu}{{*-}}

\def\sm{the Supplementary Materials}

\def\dpp{\mathbb{P}}
\def\dt{\mathbb{T}}
\def\dq{\mathbb{Q}}

\def\pp{\mathcal{P}}
\def\pt{\mathcal{T}}
\def\pq{\mathcal{Q}}

\def\sumij{\sum_{i=1}^m\sum_{j=1}^{n_i}}
\def\sumijz{\sum_{i=1}^m\sum_{j: Z_{ij}=z}}
\def\sumq{\sum_{ij \in \pq}}
\def\sump{\sum_{ij \in \pp}}
\def\sumpt{\sum_{ij \in \pt}}
\def\sumz{\sum_{ij \in \pp: Z_{ij} =z}}
\def\sumc{\sum_{ij: Z_{ij} =0}}
\def\sumt{\sum_{ij: Z_{ij} =1}}
\def\uij{u_{ij}}
\def\indz{1\{Z_{ij}=z\}}

\newcommand{\cp}{\overset{p}{\to}}
\newcommand{\cd}{\overset{d}{\to}}

\newcommand{\R}{{{\mathbb R}}} 

\newcommand{\zeroitem}{\setlength{\itemsep}{0pt}
\setlength{\parsep}{0pt}
\setlength{\parskip}{0pt}}

\newcommand{\indep}{\perp \!\!\! \perp}

\DeclareMathOperator*{\argmin}{arg\,min}

\renewcommand{\hat}{\widehat}
\renewcommand{\bar}{\overline}

\usepackage{color}
\begin{document}
\pagestyle{plain}

\newcommand{\blind}{0}

\newcommand{\tit}{\Large On the Use of Weighting for Personalized and Transparent Evidence Synthesis}

\if0\blind

{\title{\Large \tit
\thanks{
We thank David Bruns-Smith and Maya Mathur for helpful comments. 
This work was partly supported by an award from the Patient Centered Outcomes Research Initiative
(PCORI, ME-2024C2-40180).}\vspace*{.3in}}
\author{\normalsize Wenqi Shi\thanks{Department of Statistics, Harvard University, 1 Oxford Street, Cambridge, MA 02138; email: \url{wenqi_shi@g.harvard.edu}.} \and  \normalsize Jos\'{e} R. Zubizarreta\thanks{Departments of Health Care Policy, Biostatistics, and Statistics, Harvard University, 180 Longwood Avenue, Office 215-A, Boston, MA 02115; email: \url{zubizarreta@hcp.med.harvard.edu}.}}
\date{}

\maketitle
\date{}
}\fi

\if1\blind
\title{ \tit}
\date{}
\maketitle
\fi

\vspace{-.5cm}
\begin{abstract}

Over the past few decades, statistical methods for causal inference have made impressive strides, enabling progress across a range of scientific fields. However, much of this methodological development has been confined to individual studies, limiting its ability to draw more generalizable conclusions. Achieving a thorough understanding of cause and effect typically relies on the integration, reconciliation, and synthesis from diverse study designs and multiple data sources. Furthermore, it is crucial to direct this synthesis effort toward understanding the effect of treatments for specific patient populations. To address these challenges, we present a weighting framework for evidence synthesis that handles both individual- and aggregate-level data, encompassing and extending conventional regression-based meta-analysis methods. We use this approach to tailor meta-analyses, targeting the covariate profiles of patients in a target population in a sample-bounded manner, thereby enhancing their personalization and robustness. We propose a technique to detect studies that meaningfully deviate from the target population, suggesting when it might be prudent to exclude them from the analysis. We establish multiple consistency conditions and demonstrate asymptotic normality for the proposed estimator. We demonstrate the effectiveness of the method through a simulation study and two real-world case studies.

\end{abstract}

\begin{center}
\end{center}
\clearpage
\doublespacing

\singlespacing
\pagebreak
\tableofcontents
\pagebreak
\doublespacing

\section{Introduction}
\label{sec:intro}

In recent years, the field of causal inference has seen remarkable progress, yet much of this innovation has been confined to isolated studies, each conducted under varying conditions. 
Differences in study design, enrolled populations, and outcome measurements can obscure treatment effects, yielding conflicting results. 
Consequently, study integration and meta-analytic techniques have emerged as key tools to synthesize evidence from multiple data sources, often yielding more accurate estimates than those of individual studies \citep{colnet2024causal,rosenbaum2021replication,schmid2020handbook}.
The increasing sophistication and adoption of these techniques are evident in the growing literature, underscoring the central role that meta-analysis plays in assembling credible and generalizable evidence (Figure \ref{fig:paper_counts}).

\begin{figure}[!htbp]
     \centering
     \includegraphics[width=0.875\linewidth]{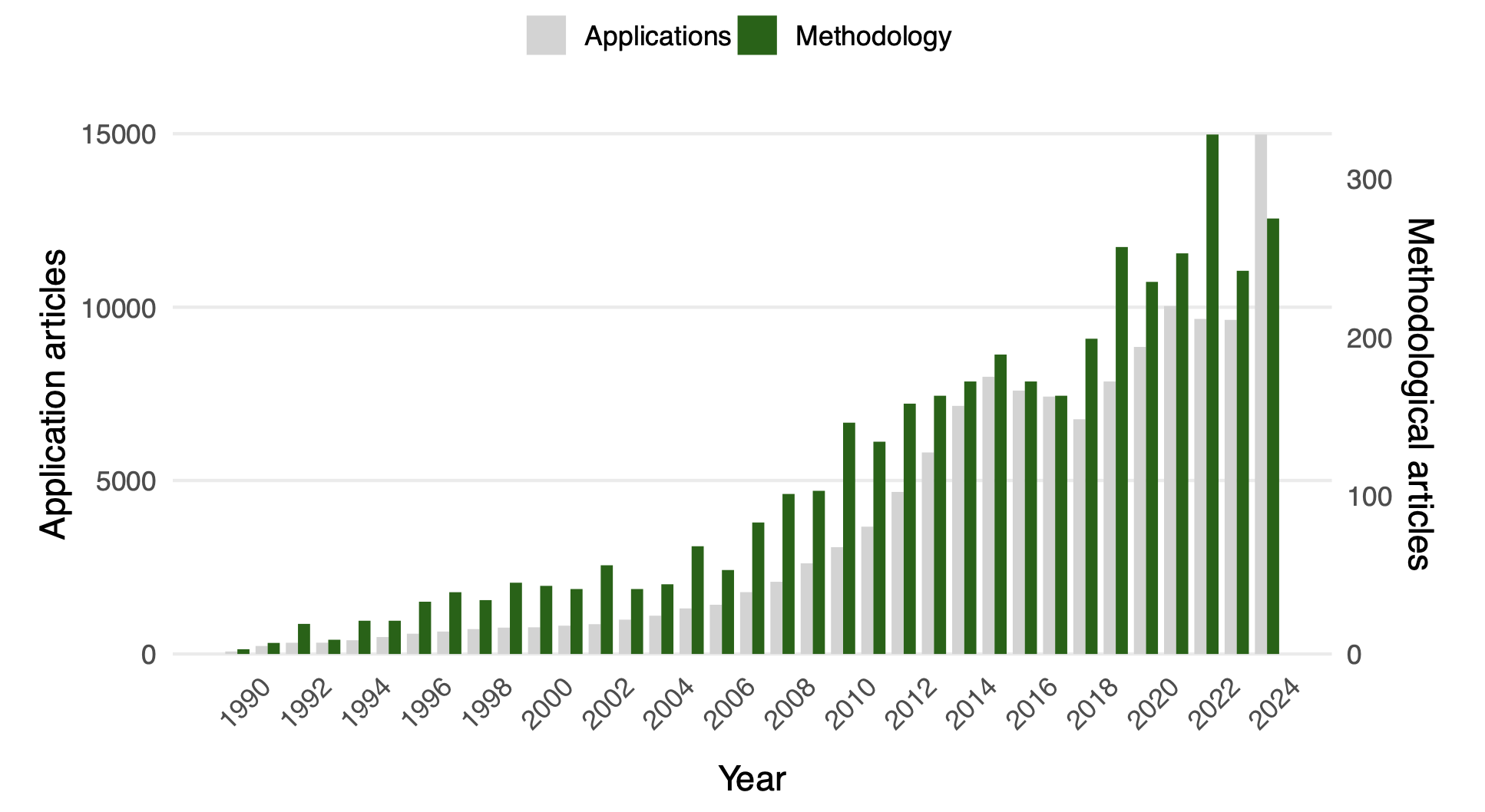}
     \caption{Number of articles with the word ``meta-analysis'' in the title, abstract or keywords from 1990 to 2024 in applications and methodology journals.}
     \label{fig:paper_counts}
\end{figure}

In meta-analysis, there are two primary frameworks, defined by the level of granularity of the available data: individual-level and aggregate-level data \citep{simmonds2005meta}. 
Analyses using individual-level data are conducted with multiple, separately collected datasets and generally offer greater modeling flexibility and improved statistical efficiency.
However, such analyses remain relatively uncommon due to practical and organizational barriers in data sharing and acquisition.
In contrast, meta-analyses based on aggregate-level data are much more common, but they are inherently constrained by the lack of granular information, which limits the flexibility of statistical adjustments and the scope of effect estimates.
In what follows, we use the term ID (individual-level data) in place of the more common IPD (individual participant data) terminology to emphasize that our framework is not restricted to patient settings, and to maintain symmetry with the term AD (aggregate-level data).

ID meta-analyses are commonly conducted using either one-stage or two-stage regression methods.
The one-stage approach simultaneously analyzes all raw individual-level data using mixed-effects multilevel regression models, while the two-stage approach first derives summary estimates for each study and then combines these in a standard meta-analysis. 
AD meta-analyses often rely on study-level summary statistics and are commonly implemented through meta-regression, where aggregated covariates are incorporated into the analysis to explore heterogeneity or potential effect modifiers \citep{tipton2019history}.

Despite their widespread use, traditional meta-analytical approaches face several important limitations. 
First, they often fail to specify a well-defined target population and thus lack a clear causal interpretation beyond the study samples \citep{berenfeld2025causal, dahabreh2020toward, vo2019novel}.
As a result, study populations may not represent the broader target population where treatment decisions are ultimately made, resulting in estimates that do not necessarily apply to the intended group of individuals \citep{degtiar2023review}.
To address this issue, \citet{dahabreh2020toward} proposed a framework for causally interpretable meta-analysis to estimate a treatment effect in a specific target population by modeling both treatment assignment and study selection.
Building on this work, a growing body of methods has been developed for causally interpretable meta-analysis \citep{clark2023causally, rott2024causally, steingrimsson2024systematically}.
While some recent approaches can combine individual- and aggregate-level data \citep{rott2025causally, vo2025integration}, to our knowledge, no existing causally interpretable methods can accommodate the setting where only aggregate-level data are available.
Furthermore, most approaches require individual-level data from the target population, which are often unavailable \citep{dahabreh2023efficient, hong2025estimating, rott2025causally}. 
Finally, it is crucial not only to define a population clearly, but also to deliberately select it so that medical recommendations can be tailored to the specific characteristics of patients, enabling personalized healthcare.

Second, effectively and transparently targeting effect estimates to specific patient populations requires a clear understanding of how individual studies or patient measurements contribute to the overall estimate. 
This requires prioritizing data sources that closely align with the target population and using estimators that give more weight to relevant information sources and avoid extrapolation beyond the support of the actual measurements. 
Traditional methods may inadvertently give undue weight to less similar studies, affecting the accuracy and usefulness of the results. 
Moreover, known issues such as the sign-reversal phenomenon in regression (\citealt{small2017instrumental}; which can lead to counterintuitive negative weights, where increasing an individual's outcome under a treatment can decrease the average outcome for that treatment group) underscore the need for transparent diagnostics. 
Ideally, a meta-analytic framework should incorporate diagnostics that quantify the contribution and influence of each study on the final estimate, detect issues such as extrapolation and sign-reversal, and identify which ``donor'' studies have the greatest weight on the analysis.

Third, existing methods often assume complete overlap between treatment and control groups \citep{brantner2023methods, dahabreh2023efficient}. 
This assumption places strict global constraints on the differences between covariate distributions in treated and control populations, which may be stronger than investigators realize, especially in settings with many covariates \citep{d2021overlap}. 
To address this issue, several methods such as propensity score trimming \citep{crump2009dealing}, classification trees \citep{traskin2011defining}, and overlap weighting \citep{li2018balancing} have been developed to help define a target population with sufficient overlap. 
However, traditional meta-analytic approaches have not fully addressed the challenges posed by populations that do not overlap or proposed estimators and inference procedures that are robust to violations of the overlap assumption. 
In addition, it is desirable to select contributing, or donor, studies in a principled way, for example by enforcing sparsity requirements to prioritize a small number of high quality studies.

To address these challenges, we introduce a weighting framework for evidence synthesis that accommodates both individual- and aggregate-level data, and that both encompasses and extends conventional regression-based meta-analysis methods. 
Our approach enables targeting meta-analyses to the covariate profiles of specific patient populations while ensuring that effect estimates are bounded by the observed data, thereby improving both personalization and robustness. 
To enhance transparency, we also propose a diagnostic technique for detecting measurements that substantially deviate from the defined target profile, providing guidance on when their exclusion may be warranted. 
A key theoretical result of our work, which is of independent interest in regression analysis, demonstrates that, in large samples, our method tends to assign positive weights to observations (or studies) that fall within the support of the target population, while systematically excluding those outside it by assigning them zero weight (Theorem \ref{thm::throw}).
In meta-analysis, this diagnostic provides a practical means to detect and trim deviant studies directly from the original data, for example, without the need to estimate propensity scores.
Finally, we present a general weighting estimator for meta-analysis, establish its consistency under multiple conditions, and demonstrate its asymptotic normality.

Our work is organized as follows.
Section 2 outlines the problem setup, defines a generic target population, and discusses identification with both individual- and aggregate-level data.
Section 3 introduces the personalized and sample bounded meta-analysis weighting method, and shows that conventional one-stage and two-stage ID meta analyses are special cases within this framework.
Section 4 derives formal properties of the method, including the selection mechanism induced by sample boundedness, multiple consistency conditions, and asymptotic normality.
Section 5 details diagnostic procedures and provides theoretical guarantees that the method can identify relevant study donors under correctly specified models.
Section 6 evaluates the method's empirical performance through simulation studies, emphasizing robustness when the target and study populations do not fully overlap.
Finally, Section 7 presents two case studies, one using ID and the other using AD, to illustrate the application of the proposed method.

\section{Target estimand and identification assumptions}
\label{sec:setup}
\vspace{-.25cm}

Consider a population of $n$ individuals nested within $m$ studies, where study $i$ contains $n_i$ individuals so that $\sum_{i=1}^m n_i = n$.
Denote individual $j$ in study $i$ as $ij$. 
For each individual $ij$, we observe the study indicator $G_{ij} = i$, treatment assignment $Z_{ij} \in \{0,1\}$, baseline covariates $X_{ij} = (X_{ij,1}, \ldots, X_{ij,p})^\top \in \mathbb{R}^p$, and a scalar outcome $Y_{ij} \in \mathbb{R}$. 
Let $Y_{ij}(0)$ and $Y_{ij}(1)$ be the potential outcomes under control and treatment, respectively. 
In the observed data, only one of these outcomes is realized for each individual: $Y_{ij} = Z_{ij} Y_{ij}(1) + (1-Z_{ij}) Y_{ij}(0)$ \citep{imbens2015causal}.

In settings where only AD are available, individual-level covariates, treatments, and outcomes are not directly observed. 
Instead, for each study $i$, we observe the estimated causal effect $\hat{\tau}_i$, its variance estimate $\hat{\sigma}_i^2$, and the average covariate vector $\bar{X}_i$.

Suppose the $n$ individuals constitute a simple random sample from the study population $\mathcal{P}$. 
Let $\mathcal{T}$ be the target population of interest. 
Denote $\mathbb{P}$ and $\mathbb{T}$ as the probability measures that characterize $\mathcal{P}$ and $\mathcal{T}$, with respective densities $p(\cdot)$ and $t(\cdot)$.
Define the probability measures for the treated and control subpopulations as $\mathbb{P}_1 = \mathbb{P}_{\cdot \mid Z=1}$ and $\mathbb{P}_0 = \mathbb{P}_{\cdot \mid Z=0}$, with corresponding densities $p_1(\cdot)$ and $p_0(\cdot)$. 
Similarly, put $\mathbb{P}_{(i)} = \mathbb{P}_{\cdot \mid G=i}$ for the measure of the subpopulation of study $i$, with density $p_{(i)}(\cdot)$, for $i = 1, \ldots, m$.
Our estimand of interest is the average treatment effect for the target population $\mathcal{T}$:
\[
\tau := \mathbb{E}_{\mathbb{T}} \{ Y(1) - Y(0) \}.
\]
Further, write $\mathcal{Q} = \mathcal{T} \cup \mathcal{P}$ for the combined target and study population, with corresponding probability measure $\mathbb{Q}$.

Without loss of generality, we model the sampling mechanism as a Bernoulli selection process. Let $S$ denote the indicator for inclusion in the study population ($S = 1$) versus the target population ($S = 0$). We assume that, conditional on covariates, the probability of being included in the study is
\[
\Pr_{\mathbb{Q}}(S=1 \mid X = x) = \frac{p(x)}{\alpha\, t(x) + p(x)},
\]
where $\alpha > 0$ reflects the relative size of the target population to the study population, specifically $\alpha = \Pr_{\mathbb{Q}}(S=0) / \Pr_{\mathbb{Q}}(S=1)$. For example, $\alpha$ will be large when the study population is a small subset of a much larger target population.

\subsection{Identification with ID}
\vspace{-.25cm}
\label{sec:identification_with_id}

We can identify the target average treatment effect, $\tau$, from individual-level data (ID) in both randomized experiments and observational studies, provided the following conditions are satisfied.
\begin{assumption}
\label{cond::1}
Assume the following conditions hold:
\begin{itemize}
    \item[(a)] Overlapping support: $\supp(\dt_X) \subseteq \supp(\dpp_{1,X}) \cap \supp(\dpp_{0, X})$.
    \item[(b)] Conditional exchangeability across treatment groups: 
    $\dpp_{1,\{Y(1), Y(0)\} \mid X} = \dpp_{0,\{Y(1), Y(0)\} \mid X}$.
    \item[(c)] Conditional exchangeability across the study and target populations:
    $\dt_{\{Y(1), Y(0)\} \mid X} = \dpp_{\{Y(1), Y(0)\} \mid X}$.
\end{itemize}
\end{assumption}

Assumption~\ref{cond::1}(a) ensures that the support of the covariate distribution in the target population is contained within the support of both treatment groups in the study population, allowing for partial (rather than complete) overlap between the study and target populations. 
Let $V := \supp(\dt_X)$ denote the support of the target population, and define the propensity score as $e(x) := \Pr_{\dpp}(Z = 1 \mid X = x)$. 
Under Assumption~\ref{cond::1}(a), both $p(x)$ and $t(x)$ are strictly positive and $e(x) \in (0,1)$ for all $x \in V$. 
Furthermore, Assumption~\ref{cond::1}(a) can be viewed as the conjunction of two positivity conditions. 
First, it implies positivity of trial participation. 
Second, it requires a weaker form of treatment positivity, stipulating that the probability of receiving either treatment is strictly between 0 and 1, but only for covariate values in the target support $V$. 
This avoids imposing unnecessary restrictions on covariate regions outside those relevant to the target population.

Assumptions~\ref{cond::1}(b) and~\ref{cond::1}(c) jointly imply that, conditional on covariates, the distributions of the potential outcomes are exchangeable across treatment groups and between the study and target populations. 
Taken together, the conditions in Assumption~\ref{cond::1} ensure that the target average treatment effect is identifiable and can be expressed as
\begin{equation}
\label{eq:est}
\tau = \E_{\dpp_{1}}\left\{ \frac{t(X)}{p(X) e(X)} Y \right\}
      - \E_{\dpp_{0}}\left\{ \frac{t(X)}{p(X)\{1-e(X)\}} Y \right\}.
\end{equation}
The weighting factors in~\eqref{eq:est} correct for potential selection biases arising from differences both between the study and target populations, and between treatment groups. 
They involve the density ratio $t(x)/p(x)$ as well as the propensity score $e(x)$. 
Intuitively, individuals who are underrepresented in the study population or in a particular treatment group relative to the target population will receive a higher weight.

Under Assumption~\ref{cond::1}, the weighting factor is non-negative and satisfies the  balancing condition in~\eqref{eq:balance}, which guarantees that the distribution of any function of the pretreatment covariates is balanced towards its expectation in the target population. 
Specifically, setting $f(x) = 1$ in~\eqref{eq:balance} demonstrates that the expectation of the weighting factor is one. 
{\small
\begin{eqnarray}
\label{eq:balance}
\E_{\dpp_{1}} \left\{ \frac{t(X)}{p(X)e(X)} f(X) \right\} = \E_{\dpp_{0 }} \left\{ \frac{t(X)}{p(X)\{1-e(X)\}} f(X) \right\} 
= \E_{\dt} \left\{f(X) \right\} \text{ for any  function $f$}.
\end{eqnarray}}

These desirable properties of non-negativity and balancing motivate the weighting method proposed in Section~\ref{sec:mthd}, where the weights are estimated by directly solving for the sample analog of these properties in a single step.
Although in theory the true weighting factors can balance any function of the covariates in large samples, in practice these factors are unknown and data is often limited.
To address this, our method prioritizes balance over a set of basis functions that grows with the sample size.
This increasing flexibility allows for greater expressiveness and more accurate approximation as more data become available.

The weighting factor in~\eqref{eq:est} corresponds to a special case of the Riesz representer. 
A common approach for estimating it is to separately model the propensity score and the density ratio, and then multiply the two estimated components. 
However, this strategy can produce weights that are highly variable and fail to achieve adequate covariate balance \citep{chattopadhyay2024one}.  
In contrast, estimating the Riesz representer directly through balancing does not require a known analytic form \citep{chernozhukov2022debiased}, improves finite-sample covariate balance \citep{imai2014covariate}, and often produces more stable estimates \citep{zubizarreta2015stable}.

\subsection{Identification with AD}
\vspace{-.25cm}

When only aggregate-level data (AD) are available, identification of the target causal effect becomes more challenging. 
Here, we formalize the conditions required for identification in this setting. 
A common approach in this context is to synthesize evidence by forming a linear combination of study-level estimates. 
However, the absence of individual-level covariates forces all individuals within a study to share the same weight, limiting the granularity of the adjustments. 
This constraint hinders our ability to construct weights that accurately align the source studies with the target population.

As a result, with this approach it is difficult to identify the target estimand, especially when there are no restrictions on the outcome model. 
To understand this challenge, we consider a general outcome model:
\begin{eqnarray}
\label{eq:general_y}
Y_{ij}(z) = g_z(X_{ij}, \varepsilon_{ij,z}; \beta_{i,z}), \quad z \in \{0,1\},
\end{eqnarray}
where $g_z(\cdot)$ is an unrestricted real-valued function, and $\varepsilon_{ij,z}$ and $\beta_{i,z}$ are unobserved random vectors representing individual-level variability and study-specific heterogeneity, respectively. 
Let $\mathcal{Y}_\all$ denote the class of all such outcome models. 
A classical example within this class is the linear fixed effects model, where $Y_{ij}(z) = \beta_{i,z}^\T X_{ij} + \epsilon_{ij,z}$.

Identification with AD requires several additional assumptions.
Assumption~\ref{cond:ad} (a) states that potential outcomes are conditionally exchangeable across studies given the covariates. 
Part (b) requires the study-level estimators to be unbiased. 
Part (c) further stipulates the existence of a set of weights that align the study-level covariate distributions with that of the target population. 
This last condition is particularly stringent because, in the AD setting, each study $i$ is assigned a single weight $\tilde w_{(i)}$ that applies uniformly to all individuals within that study and cannot adapt to individual-level covariates. 
As a result, a feasible solution for the weights $\tilde w_{(i)}$ may not exist, especially when the number of studies $m$ is small.

\begin{assumption} Assume the following conditions hold:
\label{cond:ad}
\begin{itemize}
    \item[(a)] Conditional exchangeability across studies: $\{ Y(1), Y(0) \} \indep G \mid X$.
    \item[(b)] Unbiased study estimators: $\E\{\hat\tau_i\} = \E_\dpp\{Y(1) - Y(0) \mid G=i\}$ for $i \in \{1, \ldots,m\}$.
    \item[(c)] Existence of covariate balancing weights: $\sumi \tilde w_{(i)} p_{(i)}(x) = t(x)$ for all $x \in \supp(\dpp_X)$.
\end{itemize}
\end{assumption}

Proposition~\ref{prop:ad_identification} shows that, despite its strictness, Assumption~\ref{cond:ad}(c) is not only sufficient but also necessary for identification when the outcome model is unrestricted, provided that parts (a) and (b) of the assumption also hold.

\begin{proposition}
\label{prop:ad_identification}
Suppose the conditions in Assumptions \ref{cond::1} and (a)--(b) hold.
Then Assumption~\ref{cond:ad}(c) is both necessary and sufficient for identification without any restriction on the outcome model; that is, $\E\left\{\sum_{i=1}^m \tilde{w}_{(i)} \hat\tau_i\right\} = \tau$ holds for all $\{Y(1), Y(0)\} \in \mathcal{Y}_\all$.
\end{proposition}

The study-level weights that satisfy the identification condition in Proposition~\ref{prop:ad_identification} must also satisfy the following balancing condition in \eqref{eq:ad_weights_bal}. 
In particular, by setting $f(x) = 1$, we have $\sum_{i=1}^m \tilde w_{(i)} = 1$. 
These considerations motivate the design of our proposed method for the AD setting:
\begin{eqnarray}
\label{eq:ad_weights_bal}
\sum_{i=1}^m \tilde w_{(i)} \E_{\dpp}\{f(X) \mid G=i\} = \E_{\dt}\{f(X)\}, \;
\text{for any function $f$.}
\end{eqnarray}

\section{Methods}
\label{sec:mthd}
\vspace{-.25cm}

In this section, we present a method for conducting personalized and sample-bounded meta-analysis. 
We first introduce the method in settings where ID are available.
We then extend it to a two-stage procedure that accommodates the more general setting in which some, or even all, studies provide only aggregate-level data. 
Finally, we show that classical one- and two-stage ID meta-analysis methods emerge as special cases of our framework.

Our method can be formulated within the following convex optimization framework:
\begin{eqnarray}
\label{eq:generic_opt}
\hat{w} = \arg\min_{w} \left\{ \mathcal{D}(w) : \, w \in \mathcal{B} \cap \mathcal{A} \right\},
\end{eqnarray}
where $\mathcal{D}(w)$ measures the dispersion of the weights, $\mathcal{B}$ encodes covariate balance constraints, and $\mathcal{A}$ imposes additional adjustment constraints such as normalization and sample-boundedness \citep{wang2020minimal}.
Crucially, this formulation allows the weighting factor in Equation \ref{eq:est} to be estimated in a single step, avoiding the multiplicative structure of common multi-stage weighting procedures and yielding more stable estimates \citep{chattopadhyay2024one}.
Specific estimators arise from different specifications of $\mathcal{D}$, $\mathcal{B}$, and $\mathcal{A}$. 
We describe these choices in detail in the following sections. 
A complete summary of the notation is provided in Table~\ref{tab_notation} in the Supplementary Materials.

\subsection{One-stage method for ID}
\vspace{-.25cm}

We begin by considering the setting in which all trials provide ID on covariates, treatments, and outcomes. 
Here, the personalized and sample-bounded estimator is defined as
\begin{eqnarray}
\label{eq:est_theory}
\hat\tau^\pb_\id = \sum_{i=1}^m \sum_{j: Z_{ij}=1} \hat w_{1,ij}^\pb Y_{ij} - \sum_{i=1}^m \sum_{j: Z_{ij}=0} \hat w_{0,ij}^\pb Y_{ij},
\end{eqnarray}
where the weights $\hat w^\pb_{z,ij}$ for $z \in \{0,1\}$ are obtained by solving the following optimization problem:
\begin{eqnarray}
\label{eq:PBM0}
\hat{w}_{z,ij}^\pb = \arg\min_{w_z} \left\{ \mathcal{D}(w_z) : w_z \in \mathcal{B}_\id^* \cap \A^+ \right\}.
\end{eqnarray}

In the proposed formulation, we explicitly minimize a dispersion measure of the weights, defined as $\mathcal{D}(w_z) := \sum_{i,j : Z_{ij}=z} \psi(w_{z,ij})$, where $\psi(\cdot)$ is a convex function. 
For example, one common choice is the squared $l_2$ norm, i.e., $\mathcal{D}^{l_2}(w_z) := \sum_{i,j : Z_{ij}=z} w_{z,ij}^2$.  
Let $B_k (\cdot)$ for $k = 1, \ldots, K$ denote a set of $K$ basis functions of the covariates, say polynomial terms. 
We allow $K$ to increase with the sample size $n$, enabling greater expressiveness approximation accuracy as more data become available.
Denote $\delta_k \geq 0$ as a prespecified imbalance tolerance for each $B_k$.
The balancing constraint $\mathcal{B}_\id^*$ is defined as
\begin{eqnarray*}
\mathcal{B}_\id^*: = \left\{
\begin{array}{l}
\left| \sumz w_{z,ij} B_k(x_{ij}) - B_k^* \right| \leq \delta_k \quad \text{for } k \in A, \\[0.5em]
\left| \sum_{j:Z_{ij}=z} w_{z,ij} \{ B_k(x_{ij}) - B_k^* \}\right| \leq \delta_k \quad \text{for } k \in W, \, i=1,\ldots,m,
\end{array}
\right.
\end{eqnarray*}
where $B_k^*$ is the target value for the $k$-th basis function. 

This balancing constraint operates at two distinct levels. 
First, across-study balancing ensures that the basis functions in set $A$ are aligned with the target covariate profile across the entire dataset. 
Second, within-study balancing enforces alignment of the basis functions in set $W$ within each study, addressing potential between-study heterogeneity. 
Together, these two levels of balancing provide flexibility, allowing the method to effectively address both global and local covariate imbalances. 
The sets $A$ and $W$ partition the full set of basis functions $\{1, \ldots, K\}$.
Personalization is made possible by specifying a target covariate profile as a vector $B^* = (B_1^*, \ldots, B_K^*)^\T \in \mathbb{R}^K$. 
The adjustment constraint is defined as $\A^+ := \{\sum_{i,j:Z_{ij}=z} w_{z,ij} = 1; \, w_{z,ij} \geq 0\}$, which enforces both normalization and sample-boundedness. 
Specifically, normalization guarantees that the weights sum to one, while non-negativity prevents extrapolation, ensuring the estimator remains sample-bounded.
In sections~\ref{sec:theory} and~\ref{sec:diagnostics}, we show that this sample-boundedness constraint translates into a selection mechanism, which helps to detect the data points (e.g., study donors) that are most relevant to the target population.

In the notation $\hat\tau^\pb_\id$, the superscript $*$ indicates that the estimator incorporates the target covariate profile, while $+$ indicates the inclusion of the sample-boundedness (non-negativity) constraint. 
The subscript $\id$ denotes that the estimator is defined in the individual-level data setting. 
Similarly, $\hat\tau^\pu_\id$ denotes the personalized but unbounded estimator, where unbounded refers to the absence of non-negativity constraints. 
In this case, the adjustment constraint is given by $\A^- := \{\sumijz w_{z,ij} = 1\}$, which enforces normalization alone, without restricting the sign of the weights.

The proposed optimization framework~\eqref{eq:PBM0} provides a flexible way to tailor inferences for a specific population $\pt$ by incorporating the target covariate profile vector $B^*$. 
This profile may correspond to a particular population of interest, thereby facilitating the generalization and transportation of causal inferences. 
Crucially, balancing toward this profile does not require access to individual-level data in the target population, which may be unavailable for confidentiality reasons. 
By using summary-level information, such as the sample means of covariates, the proposed framework allows for causally-interpretable inference while maintaining data privacy, making it both flexible and practical for real-world applications where individual-level data may be limited.

\subsection{Two-stage method for AD and mixed data}

\vspace{-.25cm}
Our method also extends to the mixed data setting, where some or even all studies provide only aggregate-level data (AD). 
Algorithm~\ref{alg:two_stage} outlines a two-stage procedure for this setting. 

In the first stage, we perform within-study estimation for each study with available ID using the proposed personalized, sample-bounded approach. 
Specifically, for each study, we solve the same optimization problem~\eqref{eq:PBM0} as in the one-stage method, but restrict the analysis to data from that study alone. 
We then compute the study-specific causal effect estimator $\hat\tau_i$ and its corresponding variance estimator $\hat\sigma^2_i$ using the resulting weights. 
For variance estimation, we consider two approaches: a heuristic and a plug-in method, with further details provided in Section~\ref{sec:inference}. 
The resulting estimates, $\hat\tau_i$ and $\hat\sigma^2_i$, serve as analogues to the aggregate-level statistics available from AD studies. 
A key feature of this stage is its flexibility, as each within-study analysis can use a different set of covariates based on data availability.

In the second stage, we combine all study-level estimates, whether obtained from the first-stage ID analysis or supplied directly as AD from source studies. 
We obtain study-level weights $\hat w^\pb_{(i)}$ by solving optimization problem~\eqref{eq:AD}, which incorporates a target covariate profile and imposes non-negativity constraints for personalization and sample-boundedness.
The objective function minimizes a weighted measure of dispersion among the weights, defined as $\mathcal{D}(w;c):=\sum_i c_i \psi(w_{(i)})$, where each term is scaled by a study-specific factor $c_i>0$. 
Common choices for $c_i$ include the estimated study-level variance $\hat\sigma^2_i$ or the inverse sample size $1/n_i$, to account for study-level precision.
Balancing across studies is enforced through the constraint $\mathcal{B}_\ad^*:=\left\{\left|\sum_{i} w_{(i)}  \bar B_i - B^*\right| \leq \delta \right\}$, which ensures marginal balance only at the aggregate (across-study) level. 
The final variance estimator, $\hat V_\ad^\heur$, is calculated using a heuristic approach, as described in Section~\ref{sec:inference} in \sm.
Both stages maintain the features of personalization and sample-boundedness. 
When only AD are available, we perform the second-stage meta-analysis directly.
\begin{eqnarray}
\label{eq:AD}
\hat{w}^\pb_{(i)} = \arg\min_{w} \left\{ \mathcal{D}(w;c) : \, w \in \mathcal{B}^*_\ad \cap \A^+ \right\}.
\end{eqnarray}

\begin{algorithm}[ht]
\caption{Two-stage personalized and sample-bounded method \label{alg:two_stage}}
\textbf{Stage I: Within-study estimation} 

\KwIn{Individual-level data, target profile $B^*$}
\KwOut{Study-level estimates $\hat\tau_{i}$, $\hat \sigma^2_i$, average basis function values $\bar B_{i}$}
\For{study $i\in\{1, \ldots, m\}$ with available individual-level data}{
    Solve Program~\eqref{eq:PBM0} to obtain weights $\hat w^\pb_{Z_{ij},ij}$

    Compute the study-level effect estimator $\hat\tau_i := \hat\tau^\pb_{\id,i} = \sum_{j=1}^{n_i} Z_{ij} \hat w^\pb_{1,ij} Y_{ij} - \sum_{j=1}^{n_i} (1-Z_{ij}) \hat w^\pb_{0,ij} Y_{ij}$ 

    Compute the study-level variance estimator $\hat\sigma^2_i:= \hat V^\heur_{\id,i}$ in \eqref{eq:v_heur} using $\hat w^\pb_{Z_{ij},ij}$

    Compute the study-level average of basis functions $\bar B_i = 1/n_i \sum_{j=1}^{n_i} B(X_{ij})$
}

\textbf{Stage II: Across-study meta-analysis} 

\KwIn{Study-level estimates $\hat\tau_{i}$, $\hat \sigma^2_i$, scaling factors $c_i$, average basis function values $\bar B_{i}$, target profile $B^*$}
\KwOut{Final estimates $\hat\tau^\pb_\ad$, $\hat V^\heur_\ad$}

Solve Program~\eqref{eq:AD} to obtain weights $\hat w^\pb_{(i)}$

Compute the final effect estimator $\hat\tau^\pb_\ad = \sumi \hat w^\pb_{(i)} \hat\tau_i$

Compute the final variance estimator $\hat V^\heur_\ad$ from~\eqref{eq:v_heur_ad}
\end{algorithm}

\subsection{Connection to classical meta-analysis methods}
\vspace{-.25cm}

In this section, we show that traditional one- and two-stage meta-analyses using ID can be viewed as special cases within our framework.
This connection provides new insights into these classical methods by revealing their implicit weighting schemes.

When ID from all studies are available, a common one-stage approach is to fit a mixed-effects model to the pooled data:
\begin{eqnarray}
\label{eq:one-stage}
    Y_{ij} = \phi_i + \tau_i Z_{ij} + \beta_{i}^\T X_{ij} +  \gamma_{i}^\T X_{ij} Z_{ij} + \epsilon_{ij},
\end{eqnarray}

where $\beta_i= (\beta_{i,1}, \ldots, \beta_{i,p})^\T \in \R^p$ and $\gamma_i = (\gamma_{i,1}, \ldots, \gamma_{i,p})^\T \in \R^p$. 
Here, $\phi_i$ is a study-specific intercept, $\beta_{i}$ controls for the covariates, and $\gamma_{i}$ captures the treatment–covariate interactions. 
The primary parameter of interest is the treatment effect $\tau_i$.
 
These parameters allow for various modeling choices. 
Each parameter can be treated as a fixed effect (i.e., with distinct fixed values for each study), a common effect (e.g., $\tau_i = \tau$ for all $i$), or as a random effect (e.g., $\tau_i \sim \mathcal{N}(\tau, \sigma^2_{\tau})$). 
In particular, if $\tau_i$ is modeled as a fixed effect, the model yields study-specific treatment estimates rather than a single global treatment effect across studies, since each study is assigned its own coefficient.
In our analysis, we assume that both $\tau_i$ and $\phi_i$ are common effects, i.e., $\tau_i = \tau$ and $\phi_i = \phi$ for all studies. 
For the covariates, we allow for either fixed or common effects. 
Let $F$ denote the set of covariates with fixed effects and $C$ the set with common effects, so that $F \cup C = \{1, \ldots, p\}$. 
Specifically, for $l \in C$, we impose $\beta_{i, l} = \beta_{l}$ and $\gamma_{i, l} = \gamma_{l}$, corresponding to common effects across studies.

Two-stage meta-analysis first uses ID to estimate $\tau_i$ within each study by fitting a linear model analogous to \eqref{eq:one-stage}. 
These study-specific estimates are then combined in a second stage using standard inverse-variance meta-analysis to obtain an overall treatment effect estimate.
We show that both one- and two-stage ID meta-analysis can be recovered as special cases of the proposed framework. 
In particular, the one-stage method corresponds to the optimization problem in~\eqref{eq:framework0}, which is a specific instance of the more general formulation in~\eqref{eq:PBM0}.
Here, the objective function $\mathcal{D}^{l_2}(w_z)$, defined as the squared $l_2$ norm of the weights, corresponds to the general dispersion measure with $\psi(w) = w^2$. 
The balancing constraint in~\eqref{eq:PBM0} defined as
\begin{eqnarray*}
\mathcal{B}^\ols_\id:=
\left\{
\begin{array}{l}
\sum_{i=1}^m \sum_{j:Z_{ij}=z} w_{z,ij} X_{ij,l} = 0 \quad \text{for } l \in C \\[0.5em]
\sum_{j:Z_{ij}=z} w_{z,ij} X_{ij,l} = 0 \quad \text{for } l \in F, \, i=1,\ldots,m
\end{array}
\right.
\end{eqnarray*}
is a special case of the general constraint $\mathcal{B}^*_\id$, obtained by choosing the identity basis functions, setting the target profile $B^* = 0$, imposing exact balance ($\delta=0$), and taking the groupings as ${A} = {C}$ and ${W} = {F}$.
Proposition~\ref{prop:one-stage} shows that this special case recovers the classical one-stage ID meta-analysis estimator
\begin{eqnarray}
\label{eq:framework0}
\hat{w}^\ols_{z,ij} = \arg\min_{w_z} \left\{ \mathcal{D}^{l_2}(w_z) : \, w_z \in \mathcal{B}_\id^\ols  \cap \A^- \right\}. 
\end{eqnarray}

\begin{proposition}
\label{prop:one-stage}
Let $\hat{\tau}^\ols$ denote the estimator obtained by fitting the one-stage ID meta-analysis model~\eqref{eq:one-stage}.  
Then, $\hat{\tau}^\ols$ can be recovered by the optimization problem in~\eqref{eq:framework0}, in the sense that
\begin{eqnarray*}
    \hat{\tau}^\ols = \sum_{i=1}^m \sum_{j: Z_{ij}=1} \hat{w}_{1,ij}^\ols Y_{ij} - \sum_{i=1}^m \sum_{j: Z_{ij}=0} \hat{w}_{0, ij}^\ols Y_{ij},
\end{eqnarray*}
where $\hat{w}_{z,ij}^\ols$ are the optimal weights for treatment group $z\in\{0,1\}$ as defined in~\eqref{eq:framework0}.
\end{proposition}

Proposition~\ref{prop:one-stage} builds on \citet{chattopadhyay2023implied} and establishes that the implied weights for the one-stage ID meta-analysis estimator can be obtained by solving a quadratic programming problem that minimizes the squared $l_2$ norm of the weights and enforces a special form covariate balance, where fixed effects produce within-study balance and common effects result in balance across all studies. 
Because the weights are not required to be non-negative, extrapolation is possible.
Analogously, the standard two-stage ID meta-analysis is a special case of our two-stage approach in Algorithm~\ref{alg:two_stage} by setting $B^* = 0$, $\delta = 0$, $\psi(w) = w^2$ and $c_i = \hat\sigma_i^2 + \sigma_\tau^2$, and omitting non-negativity constraints.

The proposed weighting framework encompasses estimators from both one- and two-stage ID meta-analysis, providing new insights into these classical methods. 
However, the formulation of Program \eqref{eq:framework0} does not permit to balance covariates toward a personalized target profile. 
In practice, the study and target populations often differ due to a range of practical reasons. 
To address some of these differences, it is important to integrate the target covariate profile into the optimization framework.
Furthermore, the weights produced by \eqref{eq:framework0} can be negative, which means the final estimator may extrapolate beyond the convex hull of the study samples and is not necessarily sample-bounded. 
In particular, negative weights can produce sign-reversal phenomena, where a higher outcome, say, for a treated unit can result in a lower outcome average for the treated group, making interpretation difficult in practice.

Finally, we note that by selecting an alternative dispersion measure for the weights, such as the $l_1$ norm instead of the squared $l_2$ norm, one can induce sparsity in the optimal weights. 
Specifically, minimizing the sum of the absolute values of the weights (the $l_1$ norm) tends to produce many weights exactly equal to zero. 
Consequently, only a subset of observations or studies receive nonzero weights and are included in the effect estimate. 
While this approach promotes interpretability by selecting a parsimonious set of contributing data points, it may exclude potentially informative measurements and reduce statistical efficiency.

\section{Theoretical results}
\label{sec:theory}
\vspace{-.25cm}

This section presents the theoretical properties of the proposed method for settings with ID. 
We establish several conditions under which the estimator is consistent and prove its asymptotic normality. 
Due to space constraints, when only AD are available, consistency of $\hat\tau^\pb_\ad$ is established separately in Section~\ref{sec:AD_appendix} of the Supplementary Materials.

We focus on the estimator $\hat{\tau}^\pb_\id$ as defined in~\eqref{eq:est_theory}. 
Throughout this section, we consider only across-study balance, where the set of balancing basis functions is $A = \{1, \ldots, K\}$.
To additionally achieve within-study balance, the basis function $B(x)$ can be augmented with interactions between study indicators and basis functions in the set $W$. 
For notational convenience, we include the constant basis function $B_1(x) = 1$, set the target value $B^*_1 = 1$, and the corresponding tolerance $\delta_1 = 0$, so that the first balancing constraint ensures normalization, i.e., $\sum_{i,j: Z_{ij}=z} w_{z,ij} = 1$.

\subsection{Selection mechanism introduced by sample-boundedness}
\vspace{-.25cm}

The optimization framework~\eqref{eq:framework0} for the one-stage ID meta-analysis does not require unit weights to be nonnegative. 
As discussed, allowing negative weights can result in the sign-reversal phenomenon.
Non-negativity constraints on the weights prevent this issue and ensure that the final estimate is sample-bounded. 
This section offers further insights into the role and interpretation of such non-negativity constraints.

\begin{theorem}
\label{thm:bounded}
The solution to Program~\eqref{eq:PBM0} for the personalized and sample bounded estimator $\hat\tau^\pb_\id$ is equivalent to the following two-stage procedure:
\begin{enumerate}
    \item[(I)] Restriction step: Find the set $R \subseteq \{(i, j) : i = 1, \ldots, m;\ j = 1, \ldots, n_i\}$ of units that are assigned non-zero weights. Discard all units not in $R$.
    \item[(II)] Relaxation step: Solve a relaxation of Program~\eqref{eq:PBM0} over the units in $R$ without the non-negativity constraint on weights.
\end{enumerate}
\end{theorem}
 
Theorem~\ref{thm:bounded} states that solving the weighting problem with non-negativity constraints is equivalent to solving the same problem without these constraints but restricted to a selected subset $R$ of the units. 
In the context of the linear mixed-effects model in~\eqref{eq:one-stage}, this means that imposing non-negativity constraints is mathematically equivalent to fitting the regression model on a carefully chosen subset of the data, thereby preventing extrapolation. 
In this sense, the result provides a principled approach to sample-bounded regression.
This result highlights the implicit selection mechanism induced by non-negativity constraints, which prioritize individuals or studies closely aligned with the target population. 
This finding is closely related to Proposition D.2 in \citet{bruns2025augmented} and extends to a broader class of models beyond OLS (e.g., ridge regression).

The subset $R$ is determined in a data-driven manner, and its explicit form is provided in Theorem~\ref{thm:dual} in the Supplementary Materials. 
The geometry of the boundary defining $R$ is governed by the choice of basis functions $B_k(\cdot)$.
Empirically, units included in $R$ tend to cluster near the target profile $B^*$, while those excluded are typically farther from $B^*$. 
Furthermore, Theorem~\ref{thm::throw} in Section~\ref{sec:diagnostics} shows that, under correctly specified models, the proposed method can asymptotically detect relevant data points or study donors.

\subsection{Consistency and asymptotic normality}
\vspace{-.25cm}

This section establishes the consistency and asymptotic normality of $\hat\tau^\pb_\id$. 
Throughout, we require that the target covariate profile $B^*$ in Program \eqref{eq:PBM0} contains relevant summary measures of the target population.
Specifically, $B^*_k = \frac{1}{n^*} \sum_{ij \in \pt} B_k(X_{ij})$ for $k= 1, \ldots, K$. 
Our method makes use only of the sample means of the basis functions and does not need individual-level data from the target population.

Assumption~\ref{cond::1} in Section \ref{sec:identification_with_id} guarantees that the target average treatment effect is identifiable. 
The next two conditions specify the modeling assumptions for the weighting mechanism and the outcome model used in our method. 
Consistency of the personalized and sample-bounded estimator $\hat\tau^\pb_\id$ requires only one of the modeling assumptions to hold.

Assumption~\ref{cond::inverse_prop} which follows specifies that the probability weighting factor in Equation~\eqref{eq:est} is a function of the basis functions $B(x)$, and places conditions on the overlap between the target and study populations. 
The function $\rho(\cdot)$ arises from the dual formulation of Problem~\eqref{eq:PBM0} and is determined by the choice of dispersion measure $\psi(\cdot)$. 
Specifically, $\rho(v) = v(\psi^\prime)^{-1}(v) + \psi\{(\psi^\prime)^{-1}(v)\}$ for $v \in \R$.
Assumption~\ref{cond::outcome_model} states that the conditional mean of the potential outcomes, given the covariates, is a linear function of the basis functions $B(x)$.

\begin{assumption}
\label{cond::inverse_prop}
For all $x \in \supp(\dpp_z)$, there exists some $\tilde\lambda_{1z} \in \mathbb{R}^K$ such that
$
 \tilde w_z(x) = n \rho^\prime \{B(x)^\T \tilde\lambda_{1z}\} \cdot 1\big\{ \rho^\prime \{B(x)^\T \tilde\lambda_{1z}\} \geq 0 \big\},
$
where the probability weight factors $\tilde w_z(x)$ are defined as $\tilde w_1(x) = \frac{t(x)}{p(x)e(x)} \cdot 1\{x \in V\}$ and $\tilde w_0(x) = \frac{t(x)}{p(x)[1-e(x)]} \cdot 1\{x \in V\}$.
\end{assumption}

\begin{assumption}
\label{cond::outcome_model}
The conditional mean of the potential outcome under treatment $Z=z$ satisfies
$
\E_{\dpp} \left[ Y(z) \mid X = x \right] = B(x)^\T \tilde \lambda_{2z}
$
for some $\tilde \lambda_{2z} \in \mathbb{R}^K$, where $\Vert \tilde\lambda_{2z} \Vert_2 \Vert \delta \Vert_2 = o(1)$.
\end{assumption}

Assumptions~\ref{cond::inverse_prop} and~\ref{cond::outcome_model} specify the modeling requirements for our personalized and sample-bounded meta-analysis estimator. 
The probability weighting model depends on the choice of dispersion measure $\psi(\cdot)$, whereas the outcome model is independent of this choice. 
Both models, however, share the same set of basis functions. 
Therefore, the dispersion measure should be chosen based on prior knowledge of the weighting mechanism, and the basis functions should be selected to adequately represent both the outcome model and the probability factors.

\begin{theorem}
\label{thm::consistency}
Suppose Assumption~\ref{cond::1} and the regularity conditions \ref{cond::regularity} in the Supplementary Materials hold.
Then, as $n \to \infty$, $\hat \tau^\pb_\id \cp \tau$ if at least one of the following conditions is satisfied:
\begin{itemize}
\zeroitem
    \item[(a)] The probability weighting model is correctly specified: Assumption~\ref{cond::inverse_prop} holds for $z \in \{0,1\}$.
    \item[(b)] The potential outcome model is correctly specified: Assumption~\ref{cond::outcome_model} holds for $z \in \{0,1\}$.
    \item[(c)] Models are correctly specified for different groups: Assumption~\ref{cond::inverse_prop} holds for $z=0$ and Assumption~\ref{cond::outcome_model} holds for $z=1$, or vice versa.
\end{itemize}
\end{theorem}

Theorem~\ref{thm::consistency} establishes that $\hat{\tau}^\pb_\id$ is consistent under multiple sets of conditions. 
Analogously to doubly robust estimators, consistency is achieved if the model for the probability weighting factor, the model for the potential outcomes, or an appropriate combination of the two is correctly specified. 
Likewise, when only AD are available, the proposed estimator $\hat\tau^\pb_\ad$ exhibits similar multiple consistency conditions, as detailed in Section~\ref{sec:AD_appendix} of the Supplementary Materials.

Finally, Theorem~\ref{thm:clt} states that if both outcome models and probability weighting models are correctly specified, under some additional regularity conditions, the personalized and sample-bounded estimator $\hat\tau^\pb_\id$ is asymptotically normal for the target average treatment effect $\tau$. 

\begin{theorem}
\label{thm:clt}
Suppose Assumptions~\ref{cond::1},~\ref{cond::inverse_prop} and \ref{cond::outcome_model}, and regularity condition~\ref{cond::regularity} in the Supplementary Materials hold. Then, as $n \to \infty$,
$
\sqrt{n} (\hat \tau^\pb_\id - \tau) \xrightarrow{d} \mathcal{N}(0, V),
$
where
$
V = \E_{\dt} \left[ \tilde w_1(X)\, \operatorname{Var}_{\dt} \left\{Y(1) \mid X \right\} +  \tilde w_0(X)\,  \operatorname{Var}_{\dt} \left\{Y(0) \mid X \right\}  \right] + \frac{1}{\alpha} \operatorname{Var}_{\dt} \left[  \E_{\dt} \left\{ Y(1) - Y(0) \mid X\right\} \right].
$
\end{theorem}

We propose two methods for estimating the variance of $\hat\tau^\pb_\id$. 
The first is a heuristic estimator that generalizes the classical variance formula used in linear regression. 
While it does not have formal theoretical guarantees in our context, it is computationally simple, parallels OLS variance estimation, and can be applied broadly to optimization-based estimators. 
The second method is a plug-in estimator that leverages the asymptotic variance characterized in Theorem~\ref{thm:clt}. 
Further details on both variance estimation approaches are provided in Section~\ref{sec:var_appendix} of the Supplementary Materials.

\section{Diagnostics for study influence and target support}
\label{sec:diagnostics}
 \vspace{-.25cm}

In our weighting representation of meta-analysis estimators, each observation is assigned a specific weight, which not only aids interpretability but also serves as a diagnostic tool. 
This feature is particularly valuable in the context of meta-analysis, as it allows investigators to characterize the population actually targeted by covariate adjustments, to calculate the effective sample size after adjustments, and to assess both the magnitude and nature of the contribution of each study to the final estimate. 
Moreover, as shown in Proposition~\ref{prop:one-stage}, the standard one-stage ID meta-analysis can be viewed as a special case of our approach, thus yielding novel diagnostics for regression-based meta-analysis. 
Importantly, these diagnostics can be performed as part of the design stage of the study, since they depend only on covariate and treatment information and do not require outcome data. 
In this sense, the proposed diagnostics can be viewed as design-based tools for evaluating the appropriateness of regression adjustments in meta-analysis.

In particular, determining which observations or studies serve as donors and quantifying their influence on the final estimate is crucial. 
To aid this assessment, we recommend plotting the sample weights for each observation against their indices and covariate values. 
These simple graphical diagnostics offer clear visualizations of donor contributions and help evaluate how well the contributing studies align with the target population.

Theorem~\ref{thm:bounded} demonstrates that imposing non-negativity constraints on the weights implicitly induces a selection mechanism, restricting the analysis to a subset of units. 
Building on this result, Theorem~\ref{thm::throw} establishes an important property of our method. 
Under a correctly specified weighting model, data points within the support of the target population are assigned non-zero weights, while those outside this support receive zero weight and are thus effectively excluded from the analysis.
In other words, the proposed approach actually selects units in the support of the target population, while systematically excluding those outside it.
\begin{theorem}\label{thm::throw}
If Assumptions~\ref{cond::1},~\ref{cond::inverse_prop}, and the regularity condition~\ref{cond::regularity} in the Supplementary Material hold, then
\begin{eqnarray*}
\Pr_{ \dpp} \left\{ \hat w^\pb_\id(X) > 0 \mid X \in V\right\} \rightarrow 1, \quad
\Pr_{ \dpp}\left\{ \hat w^\pb_\id(X) = 0 \mid X \in \supp(\dpp)\setminus V \right\}\rightarrow 1.
\end{eqnarray*}
\end{theorem}

Although Theorem~\ref{thm::throw} is stated for the ID setting, it also extends to the AD setting through a corresponding notion of target support. 
Specifically, we define the target support in the AD setting as $\tilde V = \{i: \tilde w_{(i)} > 0\}$, where $\tilde w_{(i)}$ are the study-level weights that satisfy the density condition in Assumption~\ref{cond:ad}(c). 
With this definition, Theorem~\ref{thm::throw} implies that, asymptotically, our method selects those studies whose covariate distributions align with the target population.

Taken together, Theorems~\ref{thm:bounded} and~\ref{thm::throw} provide a novel perspective on the role of non-negativity constraints in weighting methods, revealing their function as a data-driven detection mechanism of relevant units. 
These constraints not only help detect relevant study donors and ensure that only units within the support of the target population contribute to estimation, but also enhance both interpretability and robustness. 
Theorem~\ref{thm::throw} is of independent interest for regression analysis as it provides an asymptotic characterization of a principled approach to sample-bounded regression for a target population.

\section{Simulation study}
\label{sec:sim}
\vspace{-.25cm}

In this section, we assess the finite sample performance of our proposed method in comparison with other related estimators using simulated data.
We begin by examining a setting where the support of the target population fully overlaps with that of the study population. 
The simulation design is adapted from \citet{dahabreh2023efficient} and includes a total of $n + n^* =10,000$ individuals. 
For individuals in the study population, we observe the source study ID, treatment assignment, three baseline covariates, and a continuous outcome.
Full details of the data generation process are provided in Section~\ref{sec:sm_simulation} of the \sm.

For each simulated dataset, we implemented five estimators:  
(i) the proposed personalized and sample-bounded estimator $\hat\tau^\pb_\id$, obtained from Program~\eqref{eq:PBM0};  
(ii) the personalized and unbounded estimator $\hat\tau^\pu_\id$, computed from the relaxation of~\eqref{eq:PBM0} without the non-negativity constraints;  
(iii) the g-formula estimator $\hat\tau^\gf$;  
(iv) the weighting estimator $\hat\tau^\inv$, which assigns weights based on the inverse of the product of the estimated propensity score and the estimated probability of selection into the study population; and  
(v) the augmented estimator $\hat\tau^\aug$, which combines the g-formula and weighting approaches.
The last three estimators are discussed in \citet{dahabreh2023efficient}.
We note that estimation of $\hat\tau^\inv$ and $\hat\tau^\aug$ requires access to individual-level covariate data from the target population, whereas our estimator $\hat\tau^\pb_\id$ does not.
For the proposed estimators $\hat\tau^\pb_\id$ and $\hat\tau^\pu_\id$, we use the $l_2$ norm of the weights as the dispersion measure and require exact covariate balance (i.e., $\delta = 0$). 
Here, the unbounded estimator $\hat\tau^\pu_\id$ is mathematically equivalent to the g-formula estimator $\hat\tau^\gf$, hence we omit results for $\hat\tau^\gf$.

This initial simulation study design does not fully capture the features of our method, which allows the support of the target population to not fully overlap with the support of the study population.  
To more thoroughly assess the proposed method, we introduce a second simulation design, where the support of the study population only partially overlaps with that of the target, and the outcome model may differ inside and outside the overlap region, $V$. 
{ The boundary of the target support is determined by the three baseline covariates.}
Hereinafter, we refer to the original scenario as the ``fully'' overlapping case and the new scenario as the ``partially'' overlapping case.

\begin{figure}[!htbp]
    \centering
    \includegraphics[width=0.95\linewidth]{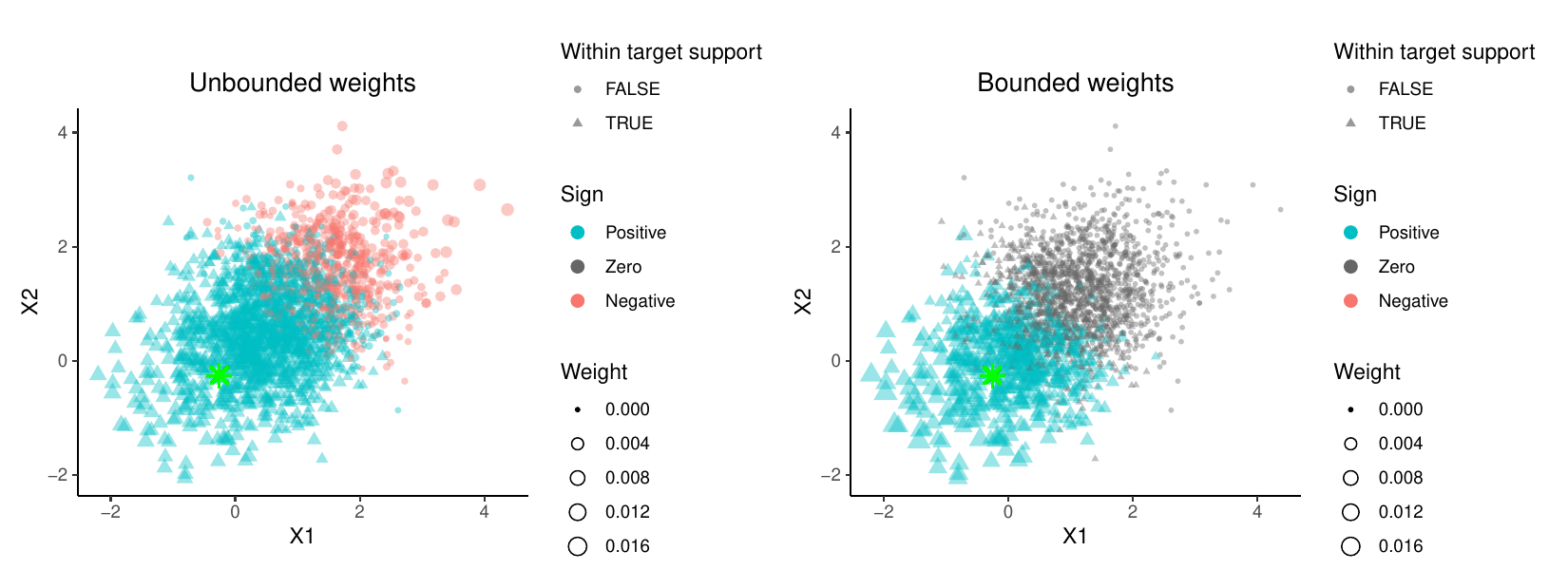}
    \caption{
    Comparison of weights used by the bounded and unbounded estimators in the partially overlapping case. 
The target profile is indicated by a green asterisk. 
Each point represents an individual, with marker size proportional to its assigned weight; colors denote the sign of the weights, and shapes indicate whether individuals are within the target support. 
In the right panel, contours around the blue points represent the detection boundary, illustrating Theorem~\ref{thm::throw} and how the proposed method selectively retains study donors most closely aligned with the target population.
    }
    \label{fig:sim2}
\end{figure}

Figure~\ref{fig:sim2} compares the weights of $\hat\tau^\pb_\id$ and $\hat\tau^\pu_\id$ in the partially overlapping case, using data from a single simulation run. 
Each point represents an observation, with point size proportional to its assigned weight, color indicating the sign of the weight, and shape denoting whether the observation is within the target support. 
The target profile is marked by a green asterisk.  
For the unbounded estimator $\hat\tau^\pu_\id$, a substantial number of observations outside the target support receive negative weights, reflecting the extrapolative nature of this estimator. 
In contrast, the sample-bounded estimator $\hat\tau^\pb_\id$ assigns zero weights to most observations outside the target support, concentrating nearly all nonzero weights on observations within the support region. 
This illustrates how the proposed personalized and sample-bounded approach can be used to detect and retain reliable study donors that align with the target population.
Empirically, the set of observations assigned negative weights by the unbounded method (colored red) substantially overlaps with those assigned zero weights by the bounded method (shown in gray), although the correspondence is not exact. 
This suggests that negative weights in the unbounded estimator may serve as a warning that the corresponding observations are outliers with respect to the target population.
In contrast, as guaranteed by Theorem~\ref{thm::throw}, our sample-bounded estimator automatically discards these outliers by assigning them zero weight.

In the fully overlapping case, all estimators perform well and detailed results are provided in Section \ref{sec:sm_simulation} of \sm.
Table~\ref{tbl:sim2} compares the estimators in the partially overlapping case under various settings. 
In this case, the personalized and sample-bounded estimator $\hat\tau^\pb_\id$ estimator exhibits the smallest bias and RMSE, while the other estimators show non-negligible bias, particularly when sample size is small. This bias arises because these methods allow all individuals, including those outside the target support who follow a different outcome model, to influence the final estimate.
In contrast, our personalized and sample-bounded estimator carefully selects study donors and achieves negligible biases. This shows the robustness of $\hat\tau^\pb_\id$ to overlap violations.

\begin{table}[!htbp]
\centering
\begin{adjustbox}{width=\linewidth}
\begin{tabular}{ccccccccccc}
\hline
$n$    & Balance & $Z$      & \multicolumn{2}{c}{$\hat\tau^\inv$} & \multicolumn{2}{c}{$\hat\tau^\pu_\id$} & \multicolumn{2}{c}{$\hat\tau^\aug$} & \multicolumn{2}{c}{$\hat\tau^\pb_\id$} \\ 
     &         & varies & Bias             & RMSE             & Bias               & RMSE              & Bias             & RMSE             & Bias               & RMSE              \\ \hline
1000 & Yes     & Yes    & 3.5820           & 3.9380           & -0.6916            & 0.7051            & 0.5559           & 0.8503           & -0.0017            & 0.1582            \\
1000 & Yes     & No     & 3.5815           & 3.9304           & -0.6824            & 0.6959            & 0.5488           & 0.8316           & -0.0016            & 0.1561            \\
1000 & No      & Yes    & 3.5858           & 3.9373           & -0.6835            & 0.6969            & 0.5543           & 0.8418           & -0.0026            & 0.1537            \\
1000 & No      & No     & 3.5815           & 3.9314           & -0.6816            & 0.6951            & 0.5508           & 0.8318           & 0.0000             & 0.1546            \\
2000 & Yes     & Yes    & 1.8251           & 2.1097           & -0.4633            & 0.4736            & 0.2268           & 0.3807           & -0.0010            & 0.1100            \\
2000 & Yes     & No     & 1.7787           & 2.0488           & -0.4567            & 0.4667            & 0.2141           & 0.3564           & -0.0036            & 0.1075            \\
2000 & No      & Yes    & 1.7858           & 2.0568           & -0.4560            & 0.4660            & 0.2166           & 0.3630           & -0.0022            & 0.1070            \\
2000 & No      & No     & 1.7830           & 2.0504           & -0.4552            & 0.4652            & 0.2175           & 0.3549           & -0.0017            & 0.1087            \\
5000 & Yes     & Yes    & 0.2905           & 0.4777           & -0.2598            & 0.2697            & 0.0253           & 0.1158           & 0.0002             & 0.0809            \\
5000 & Yes     & No     & 0.2458           & 0.4091           & -0.2552            & 0.2651            & 0.0205           & 0.1072           & -0.0002            & 0.0799            \\
5000 & No      & Yes    & 0.2533           & 0.4242           & -0.2545            & 0.2645            & 0.0215           & 0.1103           & 0.0003             & 0.0806            \\
5000 & No      & No     & 0.2452           & 0.4093           & -0.2548            & 0.2648            & 0.0196           & 0.1064           & 0.0001             & 0.0803            \\ 
\hline
\end{tabular}
\end{adjustbox}
\caption{
Bias and root mean squared error (RMSE) in the partially overlapping case.
}
\vspace{.2cm}
\footnotesize
\flushleft{
Results based on 10,000 simulation runs. Under ``Balance'' ``Yes'' indicates scenarios where trials have, on average, equal sample sizes, while ``No'' refers to scenarios with unequal trial sample sizes. In the column labeled ``$Z$ varies,'' ``Yes'' denotes scenarios where the treatment assignment mechanism differs across trials.\par}
\label{tbl:sim2}
\end{table}

\section{Case studies}
\label{sec:casestudy}
\vspace{-.25cm}

In this section, we apply the proposed method to two case studies. 
The first case study uses individual-level data from the HALT-C trial, allowing a more nuanced application of the proposed estimator. 
The second case study involves a meta-analysis of psychotherapies for depression, illustrating the method's applicability with only aggregate-level data.

\subsection{ID case study: HALT-C trial}
\vspace{-.25cm}
The Hepatitis C Antiviral Long-Term Treatment Against Cirrhosis (HALT-C) trial \citep{di2008prolonged} consists of 1,050 enrolled patients with chronic hepatitis C and advanced fibrosis, all of whom had not responded to prior antiviral therapy. The study was carried out in ten research centers, with patients being randomly assigned to receive peginterferon alfa-2a treatment or no treatment. 
Here we focus on the platelet count measured nine months post-randomization as the outcome.

Following \citet{dahabreh2023efficient}, we restrict our analysis to 948 individuals with complete baseline covariate and outcome data. 
The analysis includes 18 baseline covariates: baseline platelet count, age, sex, previous use of pegylated interferon, race, white blood cell count, history of injected recreational drugs, transfusion history, body mass index, creatine level, smoking status, previous use of combination therapy (interferon and ribavirin), diabetes, serum ferritin, ultrasound evidence of splenomegaly, alcohol use, hemoglobin, and aspartate aminotransferase (AST) levels. 
To showcase a meta-analysis with transportability to a target population, we treat patients from the largest center ($n^* = 199$) as the target population samples and the remaining 9 centers ($n = 749$) as a collection of source studies. 

Our goal is to estimate the average treatment effect in the target population.
We use the randomization data from the largest center as an empirical benchmark.
We implement three estimators: the personalized and sample-bounded estimator $\hat \tau^\pb_\id$, the personalized and unbounded estimator $\hat \tau^\pu_\id$ and the standard estimator $\hat\tau^\ols$ from one-stage ID meta-analysis. 
The two personalized estimators rely on summary-level covariate information from the sample of the target population. 
For these three estimators, we use the $l_2$ norm of the weights as the dispersion measure, take the covariates themselves as basis functions, apply across-study balancing for all covariates, and set the tolerance levels $\delta$ to zero.
For comparison, we include the augmented estimator $\hat{\tau}^\aug$ and the weighting estimator $\hat{\tau}^\inv$ from \citet{dahabreh2023efficient}, which require individual-level covariate information from the sample of the target population. 
In these analyses they specified linear regression models for the conditional expectation of the outcome, and logistic regression models for the probability of trial participation and the probability of treatment among randomized individuals.

Table~\ref{tbl:haltc_res} summarizes the results of our meta-analysis on the HALT-C dataset. 
In this strong-overlap setting --- evidenced by the comparable covariate distributions across centers --- all four estimators produce similar point estimates that closely align with the experimental benchmark. 
These findings show that our sample-bounded estimator achieves efficiency and accuracy comparable to unbounded methods when there is good overlap, while retaining its interpretability and robustness to extrapolation. 
We anticipate the benefits of this robustness would become even more pronounced in settings with poorer overlap.

\begin{table}[!htbp]
\centering
\begin{tabular}{cccc}
\hline
\textbf{Estimator}           & \textbf{Treated mean}                 & \textbf{Control mean}                   & \textbf{Effect estimate}                   \\ \hline
Benchmark      & 121.6                                    & 164.4                                    & \(-42.8\)                                    \\
$\hat\tau^\ols$          & -                     & -                     & \(-20.8\) \((-80.1, 120.7)\)                 \\ 
\( \hat{\tau}^\inv \)           & 123.2 (113.6, 136.0)                     & 165.9 (152.6, 182.9)                     & \(-42.7\) \((-60.1, -26.2)\)                 \\ 
\(\hat{\tau}^\aug \) & 124.5 (116.3, 133.5)                     & 167.1 (157.3, 177.7)                     & \(-42.6\) \(( -51.2, -34.4)\)                 \\
\( \hat{\tau}^\pu_\id \)           & 124.6 (116.8, 132.7)                     & 167.1 (157.5, 177.4)                     & \(-42.5\) \((-50.5, -34.9)  \)               \\
\(\hat{\tau}^\pb_\id \)        & \multicolumn{1}{l}{124.4 (116.6, 132.6)} & \multicolumn{1}{l}{166.8 (157.2, 177.4)} & \multicolumn{1}{l}{\( -42.4\) \((-50.6, -34.8)\)} \\
\hline
\end{tabular}
\caption{Estimated mean outcomes and treatment effects for patients from the largest research center in the HALT-C trial. Values in parentheses indicate 95\% quantile-based confidence intervals computed from 10,000 bootstrap samples. Confidence intervals based on normal approximation and variance estimates yield similar results.
\label{tbl:haltc_res}}
\end{table}

Crucially, our proposed methodology offers a key capability that distinguishes it from conventional approaches. 
Whereas standard meta-analyses are restricted to estimating overall average treatment effects, our method provides the flexibility to target and estimate effects for specific patient populations based solely on summary-level data.
To illustrate this feature, Table~\ref{tbl:haltc_id} presents causal effect estimates for three distinct patient profiles (the specific attributes of these patients are provided in \sm).
The substantial variation in these estimates highlights the personalized nature of our approach and its potential to move beyond population averages to inform individualized treatment decisions.

\begin{table}[!htbp]
\centering
\begin{tabular}{clc}
\hline
\textbf{Patient} & \textbf{Covariate profile}                                                & \textbf{Effect estimate}                           \\ \hline
Patient 1        & Middle-aged male with metabolic syndrome                            & \(-28.4\) \((-42.0, -14.1)\)                     \\
Patient 2        & Young female with mild disease                                      & \(-34.5\) \((-55.6, -14.1)\)                     \\
Patient 3        & \multicolumn{1}{l}{Older male with prior drug use and liver damage} & \multicolumn{1}{l}{\( -20.9\) \((-37.6, -4.04)\)} \\ \hline
\end{tabular}
\caption{Estimated mean outcomes and treatment effects for three hypothetical patient populations with distinct medically relevant covariate profiles. Values in parentheses indicate 95\% quantile-based confidence intervals computed from 10,000 bootstrap samples. Confidence intervals based on normal approximation and variance estimates yield similar results.}
\label{tbl:haltc_id}
\end{table}

\subsection{AD case study: psychotherapy for depression}
\vspace{-.25cm}
For our second case study, we use the database compiled by \citet{cuijpers2019meta}, which provides a comprehensive collection of randomized controlled trials on psychotherapies for depression. This continuously updated dataset covers studies conducted between 1966 and January 1, 2020, and includes a total of 472 studies.
While the database examines a variety of therapeutic approaches, we focus on trials that evaluate cognitive behavioral therapy (CBT), the most commonly studied intervention, using care-as-usual (CAU) as the control condition, which include 97 studies.
We use the depression measure as the outcome and the standardized mean difference (Hedge's g) to calculate effect sizes.

The dataset contains only aggregate-level information. For each study, we include the estimated treatment effect, corresponding standard deviation, and the following two baseline covariates: average age and number of therapy sessions.
To emulate a transportability setting, we randomly select one study to represent the target population and treat the remaining studies as the source studies. Importantly, only summary-level covariate information is available for the target population in this case. The goal is to estimate the average treatment effect in the target population, leveraging the data from the other studies.

We evaluate two estimators: the personalized and sample-bounded estimator $\hat \tau^\pb_\ad$ and the personalized and unbounded estimator $\hat \tau^\pu_\ad$. 
For both estimators, we use the $l_2$ norm of the weights as the dispersion measure, include the covariates and their squared terms as basis functions, and set the tolerance levels $\delta$ to zero.

\begin{figure}[!htbp]
    \centering
    \includegraphics[width=.75\linewidth]{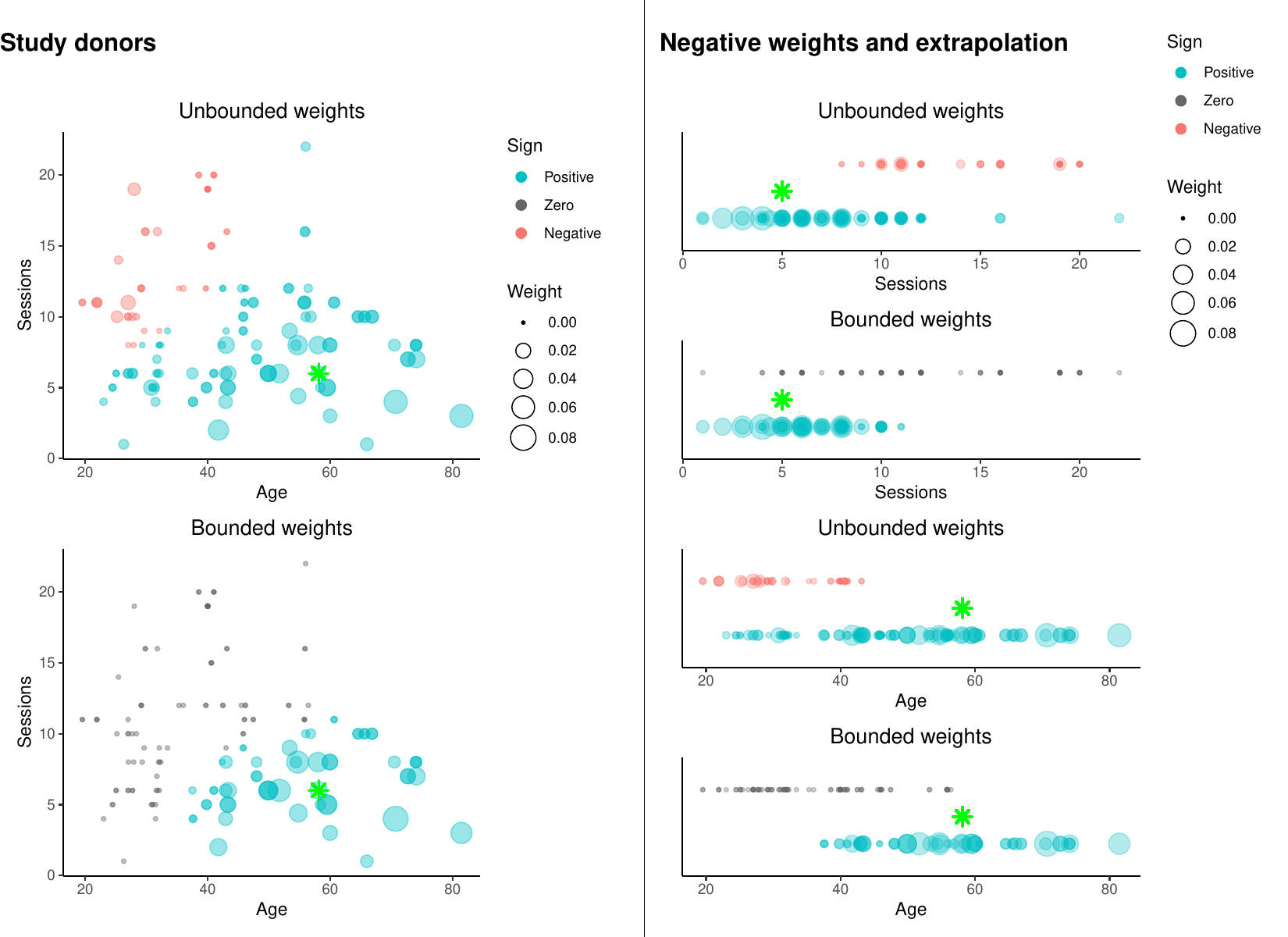}
    \caption{Diagnostics for the unbounded and bounded weights for the depression dataset. 
    Each point represents a study, with marker size proportional to its weight, and color indicating the sign of the weight. The green asterisk marks the target profile. The left panel illustrates how the studies are selected based on age and number of sessions. The right panel explores negative weights and extrapolation.
    }
    \label{fig:dep3}
\end{figure}

Figure~\ref{fig:dep3} illustrates the donor detection and extrapolation behavior of these two estimators. 
In the sub-figures, each point represents a study, with size proportional to its weight and color indicating the sign of the weight. 
The green asterisk denotes the selected target study.
The unbounded estimator assigns negative weights to several studies, particularly those with a higher number of sessions and a lower average age than the target. 
In contrast, the bounded estimator discards these studies by assigning them zero weight, resulting in a more interpretable and sample-bounded selection of donor studies.
The geometry of the detection boundary is influenced by the choice of basis functions. 
In this case, the detection boundary for contributing studies takes the shape of a hyperball as we include the squared terms of the covariates.
Table~\ref{tbl:dep} complements these visual findings by summarizing the quantitative results of the meta-analyses.

\begin{table}[!htbp]
\centering
\begin{tabular}{ccc}
\hline
{Estimator}     & {Estimate (95\% C.I.)} \\ \hline
\( \hat{\tau}^\pu_\ad \) & 0.45 (0.38, 0.57)          \\
\( \hat{\tau}^\pb_\ad  \) & 0.41 (0.35, 0.54)\\   \hline      
\end{tabular}
\caption{Results from meta-analyses using the depression psychotherapy data.\label{tbl:dep}}
\end{table}

\section{Concluding remarks}
\label{sec:conclusion}
\vspace{-.25cm}

Meta-analysis is a powerful tool for synthesizing evidence across study designs and data sources. 
In this paper, we presented a weighting framework for evidence synthesis that encompasses and extends regression-based meta-analysis to accommodate both individual- and aggregate-level data. 
This approach emphasizes interpretability and robustness by producing sample bounded estimators that avoid extrapolation, and by clarifying the contribution of each constituent piece of evidence to the final estimate. 
A key theoretical result, which is of independent interest for regression analysis, establishes that the sample-boundedness constraint provides a principled, data-driven mechanism for detecting observations most relevant to a given target population.

Several avenues for future research remain open. A promising direction is to extend this framework by leveraging machine learning and generative modeling techniques to incorporate unstructured data sources, such as clinical notes and medical images.
Furthermore, in contexts where hidden bias may be present, incorporating sensitivity analysis methods, such as those developed by \citet{mathur2020sensitivity}, would allow researchers to formally assess the robustness of their conclusions to potential unmeasured confounding, further enhancing the practical utility of our approach.



\bibliography{meta}

\begin{thebibliography}{35}
\newcommand{\enquote}[1]{``#1''}
\expandafter\ifx\csname natexlab\endcsname\relax\def\natexlab#1{#1}\fi

\bibitem[{Berenfeld et~al.(2025)Berenfeld, Boughdiri, Colnet, van Amsterdam, Bellet, Khellaf, Scornet, and Josse}]{berenfeld2025causal}
Berenfeld, C., Boughdiri, A., Colnet, B., van Amsterdam, W.~A., Bellet, A., Khellaf, R., Scornet, E., and Josse, J. (2025), \enquote{Causal Meta-Analysis: Rethinking the Foundations of Evidence-Based Medicine,} \textit{arXiv preprint arXiv:2505.20168}.

\bibitem[{Brantner et~al.(2023)Brantner, Chang, Nguyen, Hong, Di~Stefano, and Stuart}]{brantner2023methods}
Brantner, C.~L., Chang, T.-H., Nguyen, T.~Q., Hong, H., Di~Stefano, L., and Stuart, E.~A. (2023), \enquote{Methods for integrating trials and non-experimental data to examine treatment effect heterogeneity,} \textit{Statistical Science}, 38, 640--654.

\bibitem[{Bruns-Smith et~al.(2025)Bruns-Smith, Dukes, Feller, and Ogburn}]{bruns2025augmented}
Bruns-Smith, D., Dukes, O., Feller, A., and Ogburn, E.~L. (2025), \enquote{Augmented balancing weights as linear regression,} \textit{Journal of the Royal Statistical Society Series B: Statistical Methodology}, qkaf019.

\bibitem[{Chattopadhyay et~al.(2024)Chattopadhyay, Cohn, and Zubizarreta}]{chattopadhyay2024one}
Chattopadhyay, A., Cohn, E.~R., and Zubizarreta, J.~R. (2024), \enquote{One-step weighting to generalize and transport treatment effect estimates to a target population,} \textit{The American Statistician}, 78, 280--289.

\bibitem[{Chattopadhyay and Zubizarreta(2023)}]{chattopadhyay2023implied}
Chattopadhyay, A. and Zubizarreta, J.~R. (2023), \enquote{On the implied weights of linear regression for causal inference,} \textit{Biometrika}, 110, 615--629.

\bibitem[{Chernozhukov et~al.(2022)Chernozhukov, Newey, and Singh}]{chernozhukov2022debiased}
Chernozhukov, V., Newey, W.~K., and Singh, R. (2022), \enquote{Debiased machine learning of global and local parameters using regularized Riesz representers,} \textit{The Econometrics Journal}, 25, 576--601.

\bibitem[{Clark et~al.(2023)Clark, Rott, Hodges, and Huling}]{clark2023causally}
Clark, J.~M., Rott, K.~W., Hodges, J.~S., and Huling, J.~D. (2023), \enquote{Causally-interpretable random-effects meta-analysis,} \textit{arXiv preprint arXiv:2302.03544}.

\bibitem[{Colnet et~al.(2024)Colnet, Mayer, Chen, Dieng, Li, Varoquaux, Vert, Josse, and Yang}]{colnet2024causal}
Colnet, B., Mayer, I., Chen, G., Dieng, A., Li, R., Varoquaux, G., Vert, J.-P., Josse, J., and Yang, S. (2024), \enquote{Causal inference methods for combining randomized trials and observational studies: a review,} \textit{Statistical Science}, 39, 165--191.

\bibitem[{Crump et~al.(2009)Crump, Hotz, Imbens, and Mitnik}]{crump2009dealing}
Crump, R.~K., Hotz, V.~J., Imbens, G.~W., and Mitnik, O.~A. (2009), \enquote{Dealing with limited overlap in estimation of average treatment effects,} \textit{Biometrika}, 96, 187--199.

\bibitem[{Cuijpers et~al.(2020)Cuijpers, Karyotaki, Ebert, and Harrer}]{cuijpers2019meta}
Cuijpers, P., Karyotaki, E., Ebert, D., and Harrer, M. (2020), \enquote{Meta-analytic database: Randomised trials on psychotherapies for depression,} \url{https://evidencebasedpsychotherapies.shinyapps.io/metapsy/}.

\bibitem[{Dahabreh et~al.(2020)Dahabreh, Petito, Robertson, Hern{\'a}n, and Steingrimsson}]{dahabreh2020toward}
Dahabreh, I.~J., Petito, L.~C., Robertson, S.~E., Hern{\'a}n, M.~A., and Steingrimsson, J.~A. (2020), \enquote{Toward causally interpretable meta-analysis: transporting inferences from multiple randomized trials to a new target population,} \textit{Epidemiology}, 31, 334--344.

\bibitem[{Dahabreh et~al.(2023)Dahabreh, Robertson, Petito, Hern{\'a}n, and Steingrimsson}]{dahabreh2023efficient}
Dahabreh, I.~J., Robertson, S.~E., Petito, L.~C., Hern{\'a}n, M.~A., and Steingrimsson, J.~A. (2023), \enquote{Efficient and robust methods for causally interpretable meta-analysis: Transporting inferences from multiple randomized trials to a target population,} \textit{Biometrics}, 79, 1057--1072.

\bibitem[{Degtiar and Rose(2023)}]{degtiar2023review}
Degtiar, I. and Rose, S. (2023), \enquote{A review of generalizability and transportability,} \textit{Annual Review of Statistics and Its Application}, 10, 501--524.

\bibitem[{Di~Bisceglie et~al.(2008)Di~Bisceglie, Shiffman, Everson, Lindsay, Everhart, Wright, Lee, Lok, Bonkovsky, Morgan, et~al.}]{di2008prolonged}
Di~Bisceglie, A.~M., Shiffman, M.~L., Everson, G.~T., Lindsay, K.~L., Everhart, J.~E., Wright, E.~C., Lee, W.~M., Lok, A.~S., Bonkovsky, H.~L., Morgan, T.~R., et~al. (2008), \enquote{Prolonged therapy of advanced chronic hepatitis C with low-dose peginterferon,} \textit{New England Journal of Medicine}, 359, 2429--2441.

\bibitem[{D’Amour et~al.(2021)D’Amour, Ding, Feller, Lei, and Sekhon}]{d2021overlap}
D’Amour, A., Ding, P., Feller, A., Lei, L., and Sekhon, J. (2021), \enquote{Overlap in observational studies with high-dimensional covariates,} \textit{Journal of Econometrics}, 221, 644--654.

\bibitem[{Fan et~al.(2016)Fan, Imai, Liu, Ning, Yang, et~al.}]{fan2016improving}
Fan, J., Imai, K., Liu, H., Ning, Y., Yang, X., et~al. (2016), \enquote{Improving covariate balancing propensity score: A doubly robust and efficient approach,} \textit{URL: https://imai. fas. harvard. edu/research/CBPStheory. html}.

\bibitem[{Hong et~al.(2025)Hong, Liu, and Stuart}]{hong2025estimating}
Hong, H., Liu, L., and Stuart, E.~A. (2025), \enquote{Estimating target population treatment effects in meta-analysis with individual participant-level data,} \textit{Statistical Methods in Medical Research}, 09622802241307642.

\bibitem[{Imai and Ratkovic(2014)}]{imai2014covariate}
Imai, K. and Ratkovic, M. (2014), \enquote{Covariate balancing propensity score,} \textit{Journal of the Royal Statistical Society Series B: Statistical Methodology}, 76, 243--263.

\bibitem[{Imbens and Rubin(2015)}]{imbens2015causal}
Imbens, G.~W. and Rubin, D.~B. (2015), \textit{Causal inference in statistics, social, and biomedical sciences}, Cambridge university press.

\bibitem[{Li et~al.(2018)Li, Morgan, and Zaslavsky}]{li2018balancing}
Li, F., Morgan, K.~L., and Zaslavsky, A.~M. (2018), \enquote{Balancing covariates via propensity score weighting,} \textit{Journal of the American Statistical Association}, 113, 390--400.

\bibitem[{Mathur and VanderWeele(2020)}]{mathur2020sensitivity}
Mathur, M.~B. and VanderWeele, T.~J. (2020), \enquote{Sensitivity analysis for unmeasured confounding in meta-analyses,} \textit{Journal of the American Statistical Association}.

\bibitem[{Rosenbaum(2021)}]{rosenbaum2021replication}
Rosenbaum, P. (2021), \textit{Replication and evidence factors in observational studies}, Chapman and Hall/CRC.

\bibitem[{Rott et~al.(2024)Rott, Bronfort, Chu, Huling, Leininger, Murad, Wang, and Hodges}]{rott2024causally}
Rott, K.~W., Bronfort, G., Chu, H., Huling, J.~D., Leininger, B., Murad, M.~H., Wang, Z., and Hodges, J.~S. (2024), \enquote{Causally interpretable meta-analysis: Clearly defined causal effects and two case studies,} \textit{Research Synthesis Methods}, 15, 61--72.

\bibitem[{Rott et~al.(2025)Rott, Clark, Murad, Hodges, and Huling}]{rott2025causally}
Rott, K.~W., Clark, J.~M., Murad, M.~H., Hodges, J.~S., and Huling, J.~D. (2025), \enquote{Causally interpretable meta-analysis combining aggregate and individual participant data,} \textit{American Journal of Epidemiology}, 194, 2060--2068.

\bibitem[{Schmid et~al.(2020)Schmid, Stijnen, and White}]{schmid2020handbook}
Schmid, C.~H., Stijnen, T., and White, I. (2020), \textit{Handbook of Meta-Analysis}, CRC Press.

\bibitem[{Simmonds et~al.(2005)Simmonds, Higginsa, Stewartb, Tierneyb, Clarke, and Thompson}]{simmonds2005meta}
Simmonds, M.~C., Higginsa, J.~P., Stewartb, L.~A., Tierneyb, J.~F., Clarke, M.~J., and Thompson, S.~G. (2005), \enquote{Meta-analysis of individual patient data from randomized trials: a review of methods used in practice,} \textit{Clinical Trials}, 2, 209--217.

\bibitem[{Small et~al.(2017)Small, Tan, Ramsahai, Lorch, and Brookhart}]{small2017instrumental}
Small, D.~S., Tan, Z., Ramsahai, R.~R., Lorch, S.~A., and Brookhart, M.~A. (2017), \enquote{Instrumental variable estimation with a stochastic monotonicity assumption,} \textit{Statistical Science}, 32, 561--579.

\bibitem[{Steingrimsson et~al.(2024)Steingrimsson, Barker, Bie, and Dahabreh}]{steingrimsson2024systematically}
Steingrimsson, J.~A., Barker, D.~H., Bie, R., and Dahabreh, I.~J. (2024), \enquote{Systematically missing data in causally interpretable meta-analysis,} \textit{Biostatistics}, 25, 289--305.

\bibitem[{Tipton et~al.(2019)Tipton, Pustejovsky, and Ahmadi}]{tipton2019history}
Tipton, E., Pustejovsky, J.~E., and Ahmadi, H. (2019), \enquote{A history of meta-regression: Technical, conceptual, and practical developments between 1974 and 2018,} \textit{Research synthesis methods}, 10, 161--179.

\bibitem[{Traskin and Small(2011)}]{traskin2011defining}
Traskin, M. and Small, D.~S. (2011), \enquote{Defining the study population for an observational study to ensure sufficient overlap: a tree approach,} \textit{Statistics in Biosciences}, 3, 94--118.

\bibitem[{Tropp et~al.(2015)}]{tropp2015introduction}
Tropp, J.~A. et~al. (2015), \enquote{An introduction to matrix concentration inequalities,} \textit{Foundations and Trends{\textregistered} in Machine Learning}, 8, 1--230.

\bibitem[{Vo et~al.(2025)Vo, Le, Afach, and Vansteelandt}]{vo2025integration}
Vo, T.-T., Le, T. T.~K., Afach, S., and Vansteelandt, S. (2025), \enquote{Integration of aggregated data in causally interpretable meta-analysis by inverse weighting,} \textit{arXiv preprint arXiv:2503.05634}.

\bibitem[{Vo et~al.(2019)Vo, Porcher, Chaimani, and Vansteelandt}]{vo2019novel}
Vo, T.-T., Porcher, R., Chaimani, A., and Vansteelandt, S. (2019), \enquote{A novel approach for identifying and addressing case-mix heterogeneity in individual participant data meta-analysis,} \textit{Research synthesis methods}, 10, 582--596.

\bibitem[{Wang and Zubizarreta(2020)}]{wang2020minimal}
Wang, Y. and Zubizarreta, J.~R. (2020), \enquote{Minimal dispersion approximately balancing weights: asymptotic properties and practical considerations,} \textit{Biometrika}, 107, 93--105.

\bibitem[{Zubizarreta(2015)}]{zubizarreta2015stable}
Zubizarreta, J.~R. (2015), \enquote{Stable weights that balance covariates for estimation with incomplete outcome data,} \textit{Journal of the American Statistical Association}, 110, 910--922.

\end{thebibliography}
\bibliographystyle{asa}

\appendix
\setcounter{equation}{0}
\renewcommand{\theequation}{S.\arabic{equation}}
\setcounter{table}{0}
\renewcommand{\thetable}{S.\arabic{table}}
\setcounter{figure}{0}
\renewcommand{\thefigure}{S.\arabic{figure}}
\setcounter{lemma}{0}
\renewcommand{\thelemma}{S.\arabic{lemma}}
\setcounter{proposition}{0}
\renewcommand{\theproposition}{S.\arabic{proposition}}
\setcounter{corollary}{0}
\renewcommand{\thecorollary}{S.\arabic{corollary}}
\setcounter{assumption}{0}
\renewcommand{\theassumption}{S.\arabic{assumption}}
\setcounter{theorem}{0}
\renewcommand{\thetheorem}{S.\arabic{theorem}}
\setcounter{remark}{0}
\renewcommand{\theremark}{S.\arabic{remark}}

\setcounter{section}{0}
\renewcommand{\thesection}{\Alph{section}}
\renewcommand{\thesubsection}{\thesection.\arabic{subsection}}

\newtheorem{appassumption}{Assumption}[section]
\renewcommand{\theappassumption}{S.\arabic{appassumption}}

\newpage
\bigskip
\section{Supplementary Materials}
\subsection{Variance estimation}
\label{sec:var_appendix}
\vspace{-.25cm}
\label{sec:inference}
This section provides the details of two methods to estimate the variance of the personalized and sample-bounded estimator $\hat\tau^\pb_\id$. The first one, denoted by $\hat V^\heur_\id$, is a heuristic estimator inspired by the classical variance formula in linear regression, and is broadly applicable to general optimization-based estimators. The second estimator, denoted by $\widehat {V}^\plugin_\id$, is a plug-in estimator based on the asymptotic variance characterized in Theorem~\ref{thm:clt}.

In standard linear regression, estimated coefficients are linear combinations of the outcomes, and their variances are computed based on residual variance and the squared weights used in these linear combinations. This variance formula is justified by the assumption in linear regression model that all randomness arises from the outcome error term. Motivated by this idea, we propose a heuristic variance estimator $\hat V^\heur_\id$ for $\hat\tau^\pb_\id$ that recovers the classical OLS variance formula.

Proposition~\ref{prop:heur_var} formalizes this heuristic variance estimator and shows that it reproduces the standard OLS variance estimator when derived from the corresponding quadratic programming problem.
More generally, this approach is applicable to weighting estimators derived from a broad class of optimization problems, including those with arbitrary dispersion measures and non-negativity constraints.
However, it does not account for 
the design-based uncertainty present in our setting.

\begin{proposition}
\label{prop:heur_var}
The heuristic variance estimator $\hat V^\heur_\id$ is defined below. It is equivalent to the variance estimator from fitting the linear mixed model~\eqref{eq:one-stage}, when the optimization problem~\eqref{eq:est_theory} is replaced with the special case quadratic programming problem~\eqref{eq:framework0}.
\begin{itemize}
    \item [1.] Solve optimization problem~\eqref{eq:PBM0} to obtain weights $\hat w^\pb_{z,ij}$ and the effect estimator $\hat\tau^\pb_\id$.
    \item [2.] Compute the residual variance $s^2$ from fitting the linear mixed model with outcome adjusted by the causal effect estimate
$Y_{ij} - \hat\tau^\pb_\id Z_{ij} = \phi_i + \beta_{i}^\T X_{ij} +  \gamma_{i}^\T X_{ij} Z_{ij} + \epsilon_{ij}$.
    \item [3.] Compute the heuristic variance estimate as 
    \begin{eqnarray}
    \label{eq:v_heur}
    \hat V^\heur_\id = n s^2\sumij \left(\hat w_{Z_{ij},ij}^\pb\right)^2.
    \end{eqnarray}
 \end{itemize}
\end{proposition}

Based on the asymptotic normality result in Theorem~\ref{thm:clt}, we propose another plug-in variance estimator.
Notice that under Assumption~\ref{cond::1} the asymptotic variance $V$ is identifiable, and with Assumption~\ref{cond::outcome_model}, $V$ can be further simplified to
\begin{eqnarray*}
V 
&=&  \E_{\dt} \biggl[   \tilde w_1(X) \var_{\dt} \bigl\{ Y(1) \mid X \bigr\} +  \tilde w_0(X) \var_{\dt} \bigl\{ Y(0) \mid X \bigr\} \biggr] \\
&& + \frac{1}{\alpha} \var_{\dt} \biggl[ \E_{\dt} \bigl\{ Y(1) - Y(0) \mid X \bigr\} \biggr] \\
&=& \E_{\dpp_1} \bigl[ \tilde w_1^2(X) \var_{\dpp_1} \bigl\{ Y(1) \mid X \bigr\} \bigr] + \E_{\dpp_0} \bigl[ \tilde w_0^2(X) \var_{\dpp_0} \bigl\{ Y(0) \mid X \bigr\} \bigr] \\
&& + \frac{1}{\alpha} (\tilde \lambda_{21} - \tilde \lambda_{20})^\T \biggl( \E_{\dpp} \Bigl[ \tfrac{1}{2} \{Z \tilde w_1(X) + (1-Z) \tilde w_0(X)\} B(X) B(X)^\T \Bigr] \\
&& \qquad - \E_{\dt} \bigl\{ B(X) \bigr\}^\T \E_{\dt} \bigl\{ B(X) \bigr\} \biggr) (\tilde \lambda_{21} - \tilde \lambda_{20}).
\end{eqnarray*}
We can then construct a plug-in variance estimator as
\begin{eqnarray}
\label{eq:var_plugin}
\widehat {V}^\plugin_\id
=  
n \sum_{ij} \left(\hat w_{Z_{ij},ij}^\pb \right)^2 s_{Z_{ij}}^2
+ \frac{n}{n^*} (\hat \lambda_{21} - \hat \lambda_{20})^\T S_{B} (\hat \lambda_{21} - \hat \lambda_{20}).
\end{eqnarray}
Here, $\hat \lambda_{2z}$ and $s_z^2$ are the estimated coefficients and residual variances from fitting the linear model $\lmt(Y_{ij} \sim B(X_{ij}))$ within the selected subset $\{ij: Z_{ij}=z, \hat w_{z,ij} >0 \}$ for $z \in \{0,1\}$. The matrix $S_{B}$ is the weighted covariance of the basis functions in the study population samples, defined as 
\begin{eqnarray*}
S_{B} = \sumij  \hat w_{Z_{ij},ij}^\pb B(X_{ij}) B(X_{ij})^\T/2 -  B^*( B^*)^\T.
\end{eqnarray*}
Under Assumption~\ref{cond::1}, $S_{B}$ provides an unbiased estimator of $\mathrm{Cov}_{\dt}(B(X))$, and can be replaced by the sample covariance matrix from target populations samples if that is available.

Notice that the first component of the plug-in variance estimator has a similar structure to the heuristic variance estimator, as it takes the form of the residual variance multiplied by the squared weights.
This term captures the sampling and design-based uncertainty arising from the study population.
Meanwhile, the second component accounts for the sampling randomness from the target population. Its contribution vanishes when the target sample size is much larger, i.e., $n^* \gg n$.

\subsection{Aggregate-level data setting}
\label{sec:AD_appendix}
This section provides the inference and theoretical results of the personalized and sample-bounded estimator $\hat\tau^\pb_\ad$ in the setting where only aggregate-level data are available.

For inference in aggregate-level data setting, we use the heuristic variance estimator inspired by the classical variance formula in ordinary least squares. The variance estimator $\hat V^\heur_\ad$ can be computed through the following three steps.
\begin{itemize}
    \item [1.] Solve optimization problem in~\eqref{eq:AD} to obtain weights $\hat w_{(i)}^\pb$ and the effect estimator $\hat\tau^\pb_\ad$.
    \item [2.] Compute the residual variance $s^2$ from fitting $\lmt(\hat\tau_i - \hat\tau^\pb_\ad \sim 0 + \bar B_i )$ using weighted least squares, with weights $1/c_i$.
    \item [3.] Compute the heuristic variance estimate as 
\begin{eqnarray}
\label{eq:v_heur_ad}
\hat V^\heur_\ad = m s^2 \sum_{i=1}^m  \left\{ \hat w^\pb_{(i)} \right\}^2.
\end{eqnarray}
\end{itemize}


Assumption~\ref{cond:ad_weight}--\ref{cond:ad_outcome} outline the modeling assumptions imposed by our method in the setting where only aggregate-level data are available.
Assumption~\ref{cond:ad_weight} shows that the study-level weighting factor satisfying the covariate density alignment condition in Assumption~\ref{cond:ad}(c) can be expressed as a function of the expected basis functions within each study, i.e., $\E_{\dpp}[B(X)\mid G = i]$.
This functional form takes is analogous to that in Assumption~\ref{cond::inverse_prop} for the individual-level data case. 
However, unlike in the individual-level case, the input no longer includes individual-level covariate values $B(x)$, and instead relies on study-level summaries and scaling factors $c_i$.
\begin{appassumption}
\label{cond:ad_weight}
For $\tilde w_{(i)}$ that satisfies Assumption~\ref{cond:ad}(c), there exists some $\tilde \lambda_{3} \in \R^K$ s.t.
\begin{eqnarray*}
\tilde w_{(i)} = \rho^\prime \left[\E_{\dpp}\{B(X)\mid G = i\}^\T \tilde \lambda_{3} /c_i \right]  1\left\{ \rho^\prime \left[\E_{\dpp}\{B(X)\mid G = i\}^\T \tilde \lambda_{3} /c_i \right] \geq 0 \right\},
\end{eqnarray*}
for all $x \in \supp(\dpp_X)$.
\end{appassumption}

Assumption~\ref{cond:ad_outcome} states that the conditional difference of potential outcome given covariates can be represented as a linear function of the basis functions $B(x)$. This differs from its counterpart in the individual-level data setting (Assumption~\ref{cond::outcome_model}), as it puts a requirement on the outcome difference rather than the outcome itself.
\begin{appassumption}
\label{cond:ad_outcome}
The conditional difference of potential outcomes satisfies
$$\E_{\dpp} [Y(1) - Y(0) \mid X=x] = B(x)^\T \tilde\lambda_4$$
for some $\tilde\lambda_4 \in \R^K$ and $\Vert \tilde\lambda_4 \Vert_2 \Vert \delta \Vert_2 = o(1)$.
\end{appassumption}

Theorem~\ref{thm:ad_consistency} establishes multiple consistency conditions for $\hat\tau^\pb_\ad$. When only aggregate-level data are available, the consistency of the proposed estimator requires either correct specification of the implied model for the weights $\tilde w_{(i)}$ or correct specification of the conditional potential outcome difference. This result is analogous to Theorem~\ref{thm::consistency} in the individual-level data setting.
\begin{theorem} 
\label{thm:ad_consistency}
Under regularity condition~\ref{cond:reg_ad}, Assumption~\ref{cond::1}-~\ref{cond:ad}, we have $\hat\tau^\pb_\ad \cp \tau$ as $n \rightarrow \infty$ if either Assumption~\ref{cond:ad_weight} or Assumption~\ref{cond:ad_outcome} holds.
\end{theorem}

\subsection{Notation table for optimization framework}
\renewcommand{\arraystretch}{1.3} 
\begin{singlespacing}
\begin{table}[htbp]
\caption{Notation for optimization framework}
\label{tab_notation}
\begin{center}
\begin{tabular}{r c l}
\hline
$\mathcal{D}(w_z)$ & $\triangleq$ & dispersion of the weights, $\sumz \psi (w_{z,ij})$\\
$\mathcal{D}^{l_2}(w_z)$ & $\triangleq$ & squared $l_2$ norm of the weights, $\sumz w^2_{z,ij}$\\
$\mathcal{D}(w;c)$ & $\triangleq$ & weighted dispersion of the weights, $\sumi c_i \psi(w_{(i)})$\\
$\A^+$ & $\triangleq$ & $\sumijz w_{z,ij} = 1; \, w_{z,ij} \geq 0$\\
$\A^-$ & $\triangleq$ & $\sumijz w_{z,ij} = 1$\\
$\mathcal{B}_\id^*$ & $\triangleq$ & 
$\left\{
\begin{array}{l}
\left| \sumijz w_{z,ij} B_k(x_{ij}) - B_k^* \right| \leq \delta_k \quad \text{for } k \in A \\[0.5em]
\left| \sum_{j:Z_{ij}=z} w_{z,ij} \{ B_k(x_{ij}) - B_k^* \}\right| \leq \delta_k \quad \text{for } k \in W, \, i=1,\ldots,m
\end{array}
\right.$ \\
$\mathcal{B}^\ols_\id$ & $\triangleq$ & 
$\left\{
\begin{array}{l}
\sum_{i=1}^m \sum_{j:Z_{ij}=z} w_{z,ij} X_{ij,l} = 0 \quad \text{for } l \in C \\[0.5em]
\sum_{j:Z_{ij}=z} w_{z,ij} X_{ij,l} = 0 \quad \text{for } l \in F, \, i=1,\ldots,m
\end{array}
\right. $ \\
$\mathcal{B}_\ad^*$ & $\triangleq$ & $\left|\sum_{i} w_{(i)}  \bar B_i - B^*\right| \leq \delta$\\
$\psi(\cdot)$ & $\triangleq$ & a convex function of the weights\\
$c_i$ & $\triangleq$ & scaling factor for study $i$\\
$B(\cdot)$ & $\triangleq$ & $K$ basis functions of the covariates \\
$B^*$ & $\triangleq$ & target covariate profile\\
$\delta$ & $\triangleq$ & tolerance level to control imbalance\\
$C$ & $\triangleq$ & set of covariates for common effects\\
$F$ & $\triangleq$ & set of covariates for fixed effects\\
$A$ & $\triangleq$ & set of covariates for across-study balance\\
$W$ & $\triangleq$ & set of covariates for within-study balance\\
\hline
\end{tabular}
\end{center}
\end{table}
\end{singlespacing}

\pagebreak

\subsection{Additional details and results from the simulation study}
\label{sec:sm_simulation}
This section provides full details of the simulation design and additional results.

In the fully overlapping case, the data generation process is adapted from \citet{dahabreh2023efficient} and involves the following six steps:
\begin{itemize}
    \item[1. ] Generation of covariates: three covariates $X = (X_1, X_2, X_3)^\T$ are drawn from a mean-zero multivariate normal distribution with all marginal variances equal to 1 and all pairwise correlations equal to 0.5.
    \item[2.] Selection for trial participation: consider three trials ($m=3$), select units into the study population using a logistic-linear model
    \begin{eqnarray*}
        \Pr(S=1 \mid X) = \frac{\exp(\beta_0 + \beta^\T X)}{1+\exp(\beta_0 + \beta^\T X)},
    \end{eqnarray*}
    where $\beta = (\ln(2), \ln(2), \ln(2))^\T$. The intercept $\beta_0$ is solved to achieve the desired size of the study population $n$.
    The rest of the units will be in the target population.
    \item[3.] Allocation of trial participants to specific trials: select units into three trials by multinomial logistic model, $G \mid (X, S=1) \sim \text{Multinomial} (p_1, p_2, p_3; n)$ with
    \begin{eqnarray*}
    p_1 &=& \Pr(G=1 \mid X, S=1) = 1- p_2 - p_3, \\
    p_2 &=& \Pr(G=2 \mid X, S=1) = \frac{\exp(\zeta_0 + \zeta^\T X)}{1 + \exp(\zeta_0 + \zeta^\T X) + \exp(\tilde\zeta_0 + \tilde\zeta^\T X)}, \\
    p_3 &=& \Pr(G=3 \mid X, S=1) = \frac{\exp(\tilde\zeta_0 + \tilde\zeta^\T X)}{1 + \exp(\zeta_0 + \zeta^\T X) + \exp(\tilde\zeta_0 + \tilde\zeta^\T X)},
    \end{eqnarray*}
    with $\zeta = (\ln(1.5), \ln(1.5), \ln(1.5))^\T$ and $\tilde\zeta = (\ln(0.75), \ln(0.75), \ln(0.75))^\T$. The intercepts $\zeta_0$ and $\tilde\zeta_0$ are solved to achieve approximately equal-sized or unequal-sized trials with a 4:2:1 ratio of sample sizes.
    \item[4.] Random treatment assignment: generate treatment for units in three trials using Bernoulli distribution. 
    In one scenario, the treatment assignment mechanism was marginally randomized and constant across trials, $\Pr(Z=1 \mid G=i) = 1/2$ for $i \in \{1,2,3\}$.
    In the second scenario, the treatment assignment mechanism varied, with probabilities $\Pr(Z=1 \mid G=1) = 1/2, \Pr(Z=1 \mid G=2) = 1/3, \Pr(Z=1 \mid G=3) = 2/3$.
    \item[5.] Generation of potential outcomes: generate potential outcomes as $Y(z) = \theta_{z,0} +  \theta_z^\T X + \epsilon_z$ for $z \in \{0,1\}$, with $\epsilon_z$ iid standard normal, $\theta_{0,0} = 1.5, \theta_{1,0} = 0.5$, $\theta_0 = (1,1,1)^\T$ and $\theta_1 = (-1, -1, -1)^\T$. Then observed outcomes can then be generated based on consistency $Y = ZY(1) + (1-Z)Y(0)$.
\end{itemize}
In each simulated dataset, we applied five estimators: (i) the proposed personalized and sample-bounded estimator $\hat\tau^\pb_\id$ obtained from optimization problem~\eqref{eq:PBM0}, 
(ii) the personalized and unbounded estimator $\hat\tau^\pu_\id$ obtained from optimization problem~\eqref{eq:PBM0} without nonnegativity constraints, 
(iii) the g-formula estimator $\hat\tau^\gf$,
(iv) the weighting estimator $\hat\tau^\inv$ weighted by the inverse of the product of the estimated propensity score and the estimated probability of being selected into the study population,
(v) the augmented estimator $\hat\tau^\aug$ based on g-formula and weighting estimator.
The last three estimators are discussed in~\eqref{eq:PBM0}.
Note that the estimation of $\hat\tau^\inv$ and $\hat\tau^\aug$ require individual covariate information from the target population, while our estimator $\hat\tau^\pb_\id$ doesn't.
For the estimators based on optimization framework ($\hat\tau^\pb_\id$ and $\hat\tau^\pu_\id$), the dispersion measure is the $l_2$ norm of the weights, and exact balance is required, i.e., $\delta = 0$.
In this setting, the unbounded estimator $\hat\tau^\pu_\id$ and the g-formula estimator $\hat\tau^\gf$ are mathematically equivalent, so we omit results for $\hat\tau^\gf$.

In the partially overlapping case, we made two specific modifications to the original simulation design.
\begin{itemize}
    \item [$2^*$.] Selection for trial participation: first, set the target support to be
    \begin{eqnarray*}
        V = \left\{X: \frac{\exp(\beta_0 + \beta^\T X)}{1+\exp(\beta_0 + \beta^\T X)} \leq \omega \right\}.
    \end{eqnarray*}
    Select units for participation in any trial using a truncated logistic-linear model
    \begin{eqnarray*}
        \Pr(S=1 \mid X) = \left\{
        \begin{aligned}
        \frac{\exp(\beta_0 + \beta^\T X)}{1+\exp(\beta_0 + \beta^\T X)}, \quad \text{if $X \in V$} \\
        1, \quad \text{if $X \not\in V$}
        \end{aligned}
        \right.
    \end{eqnarray*}
    \item [$5^*$.] Generation of potential outcomes: use different potential outcome models for units with $ X \in V$ and $X \not\in V$. Specifically, let
    \begin{eqnarray*}
    Y(1) &=& \left\{
    \begin{aligned}
        & 0.5 - X_1 - X_2 -  X_3 +  \epsilon, \quad &\text{if $X \in V$} \\
        & 0.5 - X_1 - X_2 -  X_3  + \epsilon + 0.5 \{\beta^\T X + \log(1/\omega -1) \}, \quad &\text{if $X \not\in V$}
        \end{aligned}  \right. \\
    Y(0) &=& \left\{
    \begin{aligned}
        &1.5 + X_1 + X_2 +  X_3 + \epsilon, \quad &\text{if $X \in V$} \\
        &1.5 + X_1 + X_2 +  X_3 + \epsilon - 0.5 \{ \beta^\T X + \log(1/\omega -1) \}, \quad &\text{if $X \not\in V$}
        \end{aligned}  \right. 
    \end{eqnarray*}
    The outcome model is continuous on the boundary of the target support $V$.
\end{itemize}

Table~\ref{tbl:sim1} and~\ref{tbl:sim2.2} compare the performance of all estimators in the fully overlapping and partially overlapping settings, respectively.
Across all scenarios, the standard deviations follow the same ordering with ${\text{sd}}(\hat\tau^\pu_\id) < {\text{sd}}(\hat\tau^\pb_\id) < {\text{sd}}(\hat\tau^\aug) < {\text{sd}}(\hat\tau^\inv)$. 
In particular, both $\hat\tau^\pu_\id$ and $\hat\tau^\pb_\id$ have substantially lower variability compared to the other two estimators, which shows the stability of the one-step balancing weight strategy.
The standard deviation of $\hat\tau^\pu_\id$ is slightly smaller compared to $\hat\tau^\pb_\id$. This result is consistent with Theorem~\ref{thm:bounded} as the bounded method is equivalent to the unbounded method applied to a selected subset of the data.
For inference, we apply the plug-in variance estimator $\hat V^\plugin_\id$ to construct confidence intervals for $\hat\tau^\pb_\id$, and the coverage rates remain close to 95\% across all scenarios.

\begin{table}[!htbp]
\centering
\begin{adjustbox}{width=\linewidth}
\begin{tabular}{ccccccccccccc}
\toprule[2pt]
n    & Balance & $Z$   & \multicolumn{2}{c}{$\hat\tau^\inv$} & \multicolumn{2}{c}{$\hat\tau^\pu_\id$} & \multicolumn{2}{c}{$\hat\tau^\aug$} & \multicolumn{4}{c}{$\hat\tau^\pb_\id$}                                              \\ \cline{4-13}            &          & varies & Bias             & SD              & Bias              & SD               & Bias              & SD               & Bias    & SD     & \multicolumn{1}{l}{CI length} & \multicolumn{1}{l}{coverage} \\ \hline
1000 & Yes      & Yes    & -0.0895          & 0.8001          & 0.0006            & 0.1272           & -0.0004           & 0.2137           & -0.0003 & 0.1573 & 0.6168                         & 0.9488                       \\
1000 & Yes      & No     & -0.1301          & 0.7345          & 0.0006            & 0.1271           & -0.0007           & 0.1975           & 0.0000  & 0.1529 & 0.5992                         & 0.9501                       \\
1000 & No       & Yes    & -0.1235          & 0.7578          & 0.0005            & 0.1274           & 0.0005            & 0.2027           & 0.0004  & 0.1541 & 0.6040                         & 0.9521                       \\
1000 & No       & No     & -0.1310          & 0.7373          & 0.0004            & 0.1267           & -0.0022           & 0.1995           & -0.0012 & 0.1539 & 0.6032                         & 0.9492                       \\
2000 & Yes      & Yes    & -0.0115          & 0.5665          & -0.0004           & 0.0944           & 0.0013            & 0.1504           & 0.0003  & 0.1099 & 0.4310                         & 0.9527                       \\
2000 & Yes      & No     & -0.0574          & 0.5117          & 0.0001            & 0.0936           & -0.0002           & 0.1396           & -0.0002 & 0.1077 & 0.4223                         & 0.9510                       \\
2000 & No       & Yes    & -0.0504          & 0.5278          & 0.0008            & 0.0936           & 0.0011            & 0.1423           & 0.0009  & 0.1087 & 0.4260                         & 0.9485                       \\
2000 & No       & No     & -0.0550          & 0.5135          & 0.0018            & 0.0938           & 0.0028            & 0.1398           & 0.0018  & 0.1085 & 0.4254                         & 0.9472                       \\
5000 & Yes      & Yes    & 0.0255           & 0.3157          & 0.0002            & 0.0746           & -0.0004           & 0.0994           & -0.0001 & 0.0809 & 0.3172                         & 0.9504                       \\
5000 & Yes      & No     & -0.0155          & 0.2745          & 0.0009            & 0.0738           & -0.0005           & 0.0957           & 0.0011  & 0.0797 & 0.3124                         & 0.9497                       \\
5000 & No       & Yes    & -0.0074          & 0.2851          & 0.0011            & 0.0739           & 0.0002            & 0.0959           & 0.0011  & 0.0796 & 0.3121                         & 0.9536                       \\
5000 & No       & No     & -0.0159          & 0.2741          & 0.0002            & 0.0743           & -0.0006           & 0.0949           & -0.0002 & 0.0800 & 0.3137                         & 0.9501               
    \\     \bottomrule[2pt]
\end{tabular}
\end{adjustbox}
\caption{Bias and standard deviation (SD) of all estimators, and confidence interval (CI) length and coverage of $\hat\tau^\pb_\id$, in the \textbf{fully overlapping} case based on 10000 simulation runs. }
\begin{tablenotes}
\footnotesize
\item The total sample size ($n+n^*$) is 10,000.
In the column labeled Balance, Yes denotes scenarios where the trials have equal sample sizes on average; No denotes scenarios with unequal trial sample sizes. In the column labeled $Z$ varies, Yes denotes scenarios where the treatment assignment mechanism varies across trials.
\end{tablenotes}
\label{tbl:sim1}
\end{table}

\begin{table}[!htbp]
\centering
\begin{adjustbox}{width=\linewidth}
\begin{tabular}{ccccccccccccc}
\hline
n    & Balance & $Z$   & \multicolumn{2}{c}{$\hat\tau^\inv$} & \multicolumn{2}{c}{$\hat\tau^\pu_\id$} & \multicolumn{2}{c}{$\hat\tau^\aug$} & \multicolumn{4}{c}{$\hat\tau^\pb_\id$}                                              \\ \cline{4-13} 
     &          & varies & Bias             & SD              & Bias              & SD               & Bias              & SD               & Bias    & SD     & \multicolumn{1}{l}{CI length} & \multicolumn{1}{l}{coverage} \\ \hline
1000 & Yes      & Yes    & 3.5820           & 1.6362          & -0.6916           & 0.1377           & 0.5559            & 0.6435           & -0.0017 & 0.1582 & 0.6966                         & 0.9465                       \\
1000 & Yes      & No     & 3.5815           & 1.6189          & -0.6824           & 0.1365           & 0.5488            & 0.6247           & -0.0016 & 0.1561 & 0.6808                         & 0.9503                       \\
1000 & No       & Yes    & 3.5858           & 1.6263          & -0.6835           & 0.1359           & 0.5543            & 0.6336           & -0.0026 & 0.1537 & 0.6836                         & 0.9502                       \\
1000 & No       & No     & 3.5815           & 1.6215          & -0.6816           & 0.1362           & 0.5508            & 0.6233           & 0.0000  & 0.1546 & 0.6797                         & 0.9471                       \\
2000 & Yes      & Yes    & 1.8251           & 1.0583          & -0.4633           & 0.0979           & 0.2268            & 0.3058           & -0.0010 & 0.1100 & 0.4548                         & 0.9507                       \\
2000 & Yes      & No     & 1.7787           & 1.0168          & -0.4567           & 0.0961           & 0.2141            & 0.2850           & -0.0036 & 0.1074 & 0.4453                         & 0.9508                       \\
2000 & No       & Yes    & 1.7858           & 1.0206          & -0.4560           & 0.0961           & 0.2166            & 0.2913           & -0.0022 & 0.1070 & 0.4445                         & 0.9521                       \\
2000 & No       & No     & 1.7830           & 1.0125          & -0.4552           & 0.0960           & 0.2175            & 0.2805           & -0.0017 & 0.1087 & 0.4451                         & 0.9503                       \\
5000 & Yes      & Yes    & 0.2905           & 0.3793          & -0.2598           & 0.0723           & 0.0253            & 0.1130           & 0.0002  & 0.0809 & 0.3171                         & 0.9486                       \\
5000 & Yes      & No     & 0.2458           & 0.3270          & -0.2552           & 0.0718           & 0.0205            & 0.1053           & -0.0002 & 0.0799 & 0.0802                         & 0.9508                       \\
5000 & No       & Yes    & 0.2533           & 0.3402          & -0.2545           & 0.0722           & 0.0215            & 0.1082           & 0.0003  & 0.0806 & 0.3163                         & 0.9488                       \\
5000 & No       & No     & 0.2452           & 0.3277          & -0.2548           & 0.0720           & 0.0196            & 0.1046           & 0.0001  & 0.0803 & 0.3147                         & 0.9463                      \\     
\hline
\end{tabular}
\end{adjustbox}
\caption{Bias and standard deviation (SD) of all estimators, and confidence interval (CI) length and coverage of $\hat\tau^\pb_\id$, in the \textbf{partially overlapping} case based on 10000 simulation runs. }
\begin{tablenotes}
\footnotesize
\item The total sample size ($n+n^*$) is 10,000.
In the column labeled Balance, Yes denotes scenarios where the trials have equal sample sizes on average; No denotes scenarios with unequal trial sample sizes. In the column labeled $Z$ varies, Yes denotes scenarios where the treatment assignment mechanism varies across trials.
\end{tablenotes}
\label{tbl:sim2.2}
\end{table}


\subsection{Additional details and results from the case studies}

Figure~\ref{fig:haltc_diagnostics} presents two diagnostic plots comparing the personalized and sample-bounded estimator $\hat \tau^\pb_\id$ and the personalized and unbounded estimator $\hat \tau^\pu_\id$. Due to data use restrictions, diagnostics are shown using only summary-level measures.
The left panel displays the distributions of the weights. The unbounded weights exhibit a lower overall dispersion but can take negative values.
In contrast, the bounded weights are strictly nonnegative, as required by the sample-boundedness constraint, and may yield larger extreme values to achieve balance.
The right panel illustrates the representativeness of treatment group and control group in all 18 covariates, measured by absolute standardized mean difference from target profile.
Since the tolerance level $\delta$ is set to be zero, both unbounded and bounded weights exactly recover the target profile.

\begin{figure}[!htbp]
    \centering
    \includegraphics[width=0.75\linewidth]{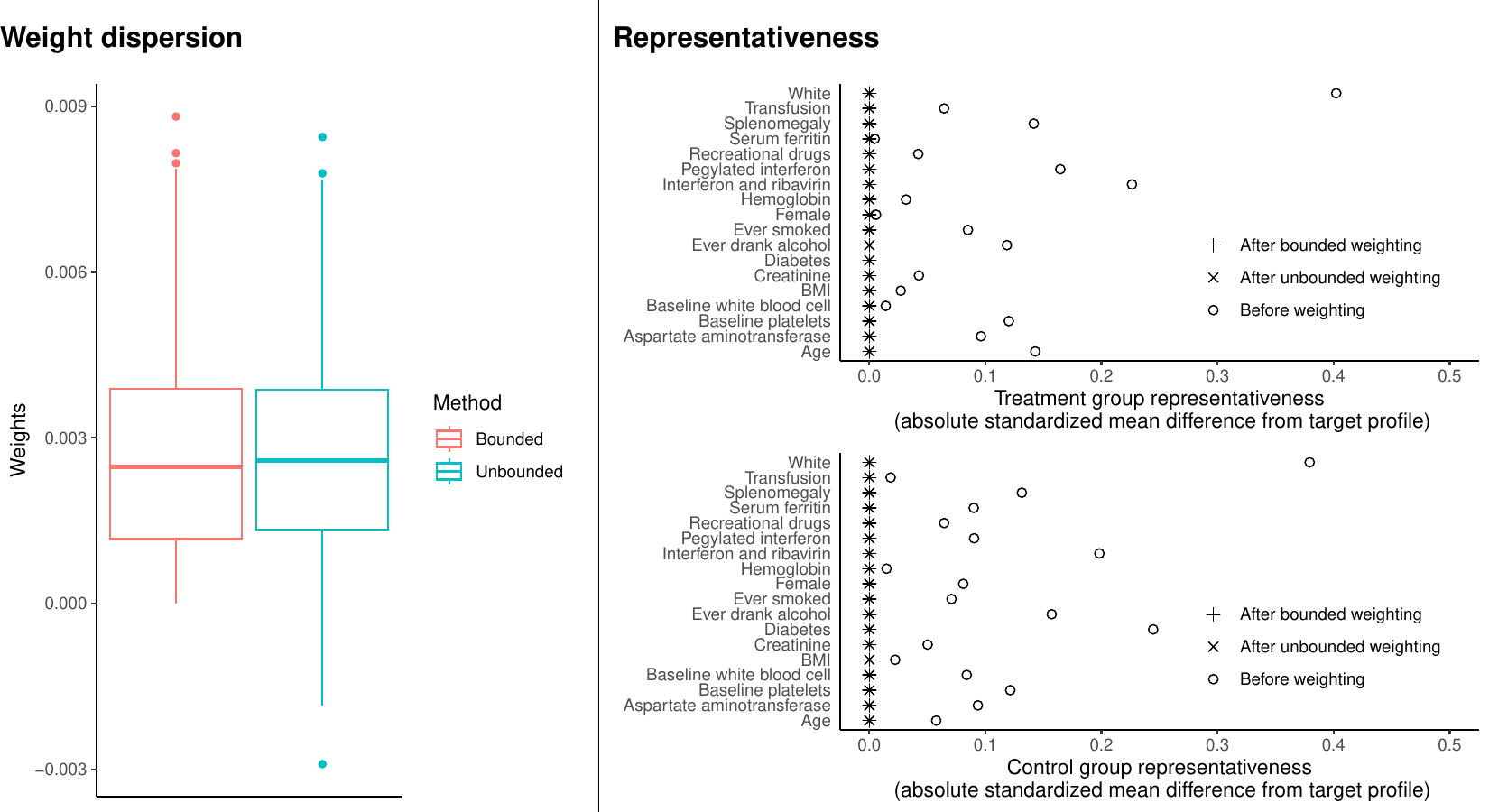}
    \caption{Diagnostics of unbounded and bounded weights for the HALT-C trial, targeting the population from the largest research center.
    The left panel shows the distribution of estimated weights across individuals. The right panel evaluates representativeness by plotting the absolute standardized mean differences between the weighted covariate distributions and the target profile, separately for treatment and control groups.}
    \label{fig:haltc_diagnostics}
\end{figure}

Table~\ref{tbl:haltc_patients} below provides the specific attributes of the three hypothetical HALT-C patients.
\begin{table}[htbp!]
\centering
\caption{Covariate profiles for three hypothetical HALT-C patients}
\begin{tabular}{lccc}
\toprule
\textbf{Covariate} & \textbf{Patient 1} & \textbf{Patient 2} & \textbf{Patient 3} \\
\midrule
Baseline platelet count & 165 & 220 & 110 \\
Age & 55 & 38 & 61 \\
Sex (Female) & 0 & 1 & 0 \\
Previous use of pegylated interferon & 1 & 1 & 0 \\
Race (White) & 1 & 0 & 0 \\
White blood cell count & 6.2 & 5.4 & 5.1 \\
History of injected recreational drug use & 0 & 0 & 1 \\
Transfusion history & 1 & 0 & 0 \\
Body mass index (BMI) & 34.2 & 24.1 & 28.7 \\
Creatinine level & 0.9 & 0.8 & 1.1 \\
Ever smoked & 1 & 0 & 1 \\
Previous use of combination therapy & 1 & 1 & 0 \\
Diabetes & 1 & 0 & 0 \\
Serum ferritin (ng/mL) & 420 & 180 & 600 \\
Splenomegaly (Ultrasound evidence) & 0 & 0 & 1 \\
Ever consumed alcohol & 1 & 0 & 1 \\
Hemoglobin (g/dL) & 13.1 & 14.8 & 12.2 \\
Aspartate aminotransferase (AST, U/L) & 85 & 42 & 132 \\
\bottomrule
\end{tabular}
\label{tbl:haltc_patients}
\end{table}

\subsection{Proofs}
\label{sec:proof}
In this section, we provide proofs of the results presented in the main text.
For clarity, we introduce the following notation.
Let $0_a$ and $0_{a \times b} $ denote the $a \times 1$ vector and $a \times b$ matrix of zeros, and let $1_a$ and $1_{a \times b} $ denote the $a \times 1$ vector and $a \times b$ matrix of ones. Let $I_a$ denote the $a \times a$ identity matrix.
Let $|A|$ denote the cardinality of the set $A$.
For a set of tuples $\{(u_i, v_{i1}, \ldots, v_{iL}): u_i \in \mathbb R,  \ v_{il} \in \mathbb R^{K_l}, \ l = 1,\ldots,L\}$, denote by $\lmt(u_i \sim v_{i1}+\cdots+v_{iL})$ the linear regression of $u_i$ on $(v_{i1}, \ldots, v_{iL})$ following the convention of R software.  We allow each $v_{il}$ to be a scalar or a vector, and use $+$ or $\sum$ to denote concatenation of regressors.
\subsubsection{Proof of Proposition~\ref{prop:ad_identification}}
\begin{proof}
[Proof of Proposition~\ref{prop:ad_identification}]
We first show that $\E[\sumi \tilde w_{(i)} \hat\tau_i] = \tau$  holds if $\sumi \tilde w_{(i)} p_{(i)}(x) = t(x)$ for all $x \in \supp(\dpp_X)$. By Assumption~\ref{cond:ad} and Assumption~\ref{cond::1}(c), we have
\begin{eqnarray*}
\E \left[\sumi \tilde w_{(i)} \hat\tau_i \right]   &=&  
\sumi \tilde w_{(i)} \E_\dpp [Y(1) - Y(0) \mid G=i] \\
&=& \sumi \tilde w_{(i)} \int \E_\dpp [Y(1) - Y(0) \mid x] p_{(i)}(x) dx \\
&=&  \int \E_\dpp [Y(1) - Y(0) \mid x] \left( \sumi \tilde w_{(i)} p_{(i)}(x) \right) dx \\
&=&  \int \E_\dt [Y(1) - Y(0) \mid x] t(x) dx = \E_{\dt} [Y(1) - Y(0) ] = \tau.
\end{eqnarray*}
We then show that if $\E[\sumi \tilde w_{(i)} \hat\tau_i] = \tau$ holds for all $\{ Y(1), Y(0)\} \in \mathcal{Y}_\all$, the covariate density alignment condition $\sumi \tilde w_{(i)} p_{(i)}(x) = t(x)$ for all $x \in \supp(\dpp_X)$ must hold. 
For arbitrary Borel set $A$, set $Y_{ij}(1) = 1\{ X_{ij} \in A \}$ and $Y_{ij}(0) = 0$. By Assumption~\ref{cond:ad} and Assumption~\ref{cond::1}(c), we then have
\begin{eqnarray*}
\E \left[\sumi \tilde w_{(i)} \hat\tau_i \right] 
&=&  \sumi \tilde w_{(i)}    \E_\dpp [Y(1) - Y(0) \mid G=i] \\
&=& \sumi \tilde w_{(i)} \Pr_\dpp(X \in A \mid G=i), \\
\tau &=& \E_{\dt} [Y(1)-Y(0)]  = \Pr_\dt(X \in A).
\end{eqnarray*}
This, together with $\E[\sumi \tilde w_{(i)} \hat\tau_i] = \tau$, ensures that $\sumi \tilde w_{(i)} p_{(i)}(x) = t(x)$ for all $x \in \supp(\dpp_X)$.

\end{proof}

\subsubsection{Proof of Proposition~\ref{prop:one-stage}} 

First, we introduce a lemma from \citet[][Theorem 3]{chattopadhyay2023implied}, which shows that causal effect estimator from uni-regression imputation (URI) can be obtained from solutions to quadratic programming problems.
\begin{lemma} 
\label{lem::url_res}
Consider the linear regression $\lmt( Y_{ij} \sim  Z_{ij} + X_{ij})$.
Let $\hat{\tau}^\uri$ denote the computed coefficient for $Z_{ij}$. 
The URI estimator of the ATE can be expressed as $\hat{\tau}^\uri = \sum_{ij:Z_{ij}=1}  \hat w_{1,ij}^\uri Y_{ij} - \sum_{ij:Z_{ij}=0} \hat w^{\pb \uri }_{0,ij} Y_{ij}$, 
where the weights $\hat w_{z,ij}^\uri$ are solutions to the following optimization problem with $z\in \{0, 1\}$, respectively.
\begin{eqnarray}
\label{eq:url_quadprog}
\text{min}_w & \sum_{i=1}^m \sum_{j: Z_{ij=z}} {(w_{ij}- \bar w_z)^2} & \\
\text{s.t.} &  &\nonumber\\
& \sum_{i=1}^m \sum_{j:Z_{ij}=z} w_{ij}  X_{ij}  =  S_c (S_t + S_c)^{-1} \bar X_t + S_t (S_t + S_c)^{-1} \bar X_c, \quad & \nonumber\\
& \sum_{i=1}^m \sum_{j:Z_{ij}=z} w_{ij}  = 1. \quad &  \nonumber 
\end{eqnarray}
Here $\bar X_t = \frac{1}{n_1} \sum_{ij:Z_{ij}=1} X_{ij}$, $\bar X_c = \frac{1}{n_0} \sum_{ij:Z_{ij}=0} X_{ij}$, $S_t = \sum_{ij:Z_{ij}=1} (X_{ij} - \bar X_t) (X_{ij} - \bar X_t)^\T$, $\bar w_z = \frac{1}{n_z} \sum_{ij:Z_{ij}=z} w_{z,ij}$ and $S_c = \sum_{ij:Z_{ij}=0} (X_{ij} - \bar X_c) (X_{ij} - \bar X_c)^\T$.
\end{lemma}
Since more complex regressions can be treated as special cases of URI by concatenating all the regressors into a new covariate $X$, we can extend the results in Lemma~\ref{lem::url_res} to one-stage ID meta-analysis.

\begin{proof}[Proof of Proposition~\ref{prop:one-stage}]
Recall that Proposition~\ref{prop:one-stage} targets at one-stage ID meta-analysis method that assumes a mixed effect model
as shown below
\begin{eqnarray*}
Y_{ij} = \phi + \tau Z_{ij} + \sum_{l \in F} \left\{ \beta_{i,l} X_{ij,l} +  \gamma_{i,l} X_{ij,l} Z_{ij} \right\} + \sum_{l \in C} \left\{ \beta_{l} X_{ij,l} +  \gamma_{l} X_{ij} Z_{ij} \right\} + \epsilon_{ij}.
\end{eqnarray*}
Denote $G_{ij,i^\prime}$ as an indicator variable that shows whether patient $ij$ belongs to study $i^\prime$, such that $G_{ij,i^\prime} = 1\{ i=i^\prime\}$. 
Using this notation, we can rewrite the model as
\begin{eqnarray}
\label{eq::thm1_one_stage}
Y_{ij} &=& \phi + \tau Z_{ij} + \sum_{l \in F} \left\{ \sum_{i^\prime} \beta_{i^\prime,l} G_{ij,i^\prime} X_{ij,l} +  \sum_{i^\prime} \gamma_{i^\prime,l} G_{ij,i^\prime} X_{ij,l} Z_{ij} \right\} \\
&& + \sum_{l \in C} \left\{ \beta_{l} X_{ij,l} +  \gamma_{l} X_{ij,l} Z_{ij} \right\} + \epsilon_{ij}. \nonumber
\end{eqnarray}
The estimator $\hat\tau^\ols_\id$ for fitting~\eqref{eq::thm1_one_stage} can be obtained by computing the coefficient for $Z_{ij}$ from the linear regression shown below
\begin{eqnarray}
\label{eq:one_stage_lm1}
\lmt \left( Y_{ij} \sim Z_{ij} + \sum_{l \in F}  \sum_{i^\prime} G_{ij,i^\prime} X_{ij,l} +
\sum_{l \in C} X_{ij,l}
+  \left\{\sum_{l \in F}  \sum_{i^\prime} G_{ij,i^\prime} X_{ij,l} +
\sum_{l \in C} X_{ij,l} \right\}Z_{ij}  \right).
\end{eqnarray}
Define $\tilde X_{ij}$ as
\begin{eqnarray*}
\tilde{X}_{ij} = \big(
G_{ij,1} X_{ij,f_1}, \ldots, G_{ij,m} X_{ij,f_1}, 
G_{ij,1} X_{ij,f_2}, \ldots, G_{ij,m} X_{ij,f_{|F|}}, 
X_{ij,c_1}, \ldots, X_{ij,c_{|C|}}
\big).
\end{eqnarray*}
Here $F = \{f_1, f_2, \ldots, f_{|F|}\}$ is the set of covariates that have fixed effects and $C = \{c_1, c_2, \ldots, c_{|C|}\}$ is the set of covariates that have common effects.
Let $\underline{\tilde X}_{ij} = (\tilde{X}_{ij}^\T, Z_{ij} \tilde{X}_{ij}^\T)^\T$.
Then $\underline{\tilde X}_{ij}$ is the concatenation of all the regressors except for treatment $Z_{ij}$ in \eqref{eq:one_stage_lm1}.
The linear regression formula~\eqref{eq:one_stage_lm1} can be re-written as
\begin{eqnarray}
\lmt \left( Y_{ij} \sim Z_{ij} + \underline{\tilde X}_{ij}  \right).
\end{eqnarray}
According to Lemma~\ref{lem::url_res}, the causal effect estimator from regression~\eqref{eq::thm1_one_stage} can be expressed as $\hat{\tau}^\ols_\id = \sum_{ij:Z_{ij}=1} \hat w_{1,ij}^\ols Y_{ij} - \sum_{ij:Z_{ij}=0} \hat w^\ols_{0,ij} Y_{ij}$, where the weights $\hat w_{z, ij}^\ols$ are solutions to the following optimization problem with $z \in \{0,1\}$, respectively.
\begin{eqnarray}
\label{eq:qp_one_stage}
\text{min}_w & \sum_{i=1}^m \sum_{j: Z_{ij=z}} {(w_{ij}- \bar w_z)^2} & \\
\text{s.t.} &  &\nonumber\\
& \sum_{i=1}^m \sum_{j:Z_{ij}=z} w_{ij}  \underline{\tilde {X}}_{ij}  =  \underline{\tilde S}_c (\underline{\tilde S}_t + \underline{\tilde S}_c)^{-1} \underline{\overline{\tilde X}}_t + \underline{\tilde S}_t (\underline{\tilde S}_t + \underline{\tilde S}_c)^{-1}  \underline{\overline{\tilde X}}_c, \quad & \nonumber\\
& \sum_{i=1}^m \sum_{j:Z_{ij}=z} w_{ij}  = 1, \quad &  \nonumber 
\end{eqnarray}
with $\underline{\overline{\tilde X}}_t = \frac{1}{n_1} \sumt \underline{\tilde {X}}_{ij}$, $\underline{\overline{\tilde X}}_c = \frac{1}{n_0} \sumc \underline{\tilde {X}}_{ij}$, $\underline{\tilde S}_t = \sumt (\underline{\tilde {X}}_{ij} - \underline{\overline{\tilde X}}_t) (\underline{\tilde {X}}_{ij} - \underline{\overline{\tilde X}}_t)^\T$ and $\underline{\tilde S}_c = \sumc (\underline{\tilde {X}}_{ij} - \underline{\overline{\tilde X}}_c) (\underline{\tilde {X}}_{ij} - \underline{\overline{\tilde X}}_c)^\T$.
It then suffices to show that the constraints in optimization problem~\eqref{eq:qp_one_stage} are equivalent to the constraints in~\eqref{eq:framework0}.

Let $\overline{\tilde X}_t = \frac{1}{n_1} \sumt {\tilde {X}}_{ij}$, ${\overline{\tilde X}}_c = \frac{1}{n_0} \sumc {\tilde {X}}_{ij}$, ${\tilde S}_t = \sumt ({\tilde {X}}_{ij} - {\overline{\tilde X}}_t) ({\tilde {X}}_{ij} - {\overline{\tilde X}}_t)^\T$ and ${\tilde S}_c = \sumc ({\tilde {X}}_{ij} - {\overline{\tilde X}}_c) ({\tilde {X}}_{ij} - {\overline{\tilde X}}_c)^\T$. 
Denote the dimension of $\underline{\tilde {X}}_{ij}$ as ${\tilde p} = m|F|+|C|$ for simplicity.
We then have
\begin{eqnarray*}
\underline{\overline{\tilde X}}_t = \begin{pmatrix}
\overline{\tilde X}_t \\
\overline{\tilde X}_t
\end{pmatrix}, \quad
\underline{\overline{\tilde X}}_c = \begin{pmatrix}
\overline{\tilde X}_c \\
0_{\tilde p}
\end{pmatrix}, \quad
\underline{\tilde S}_t = \begin{pmatrix}
\tilde S_t & \tilde S_t \\
\tilde S_t & \tilde S_t
\end{pmatrix}, \quad
\underline{\tilde S}_c = \begin{pmatrix}
\tilde S_c & 0_{\tilde p \times \tilde p} \\
0_{\tilde p \times \tilde p} & 0_{\tilde p \times \tilde p}
\end{pmatrix}.
\end{eqnarray*}
Simple algebra gives
\begin{eqnarray*}
\underline{\tilde S}_c (\underline{\tilde S}_t + \underline{\tilde S}_c)^{-1} \underline{\overline{\tilde X}}_t 
&=& 
\begin{pmatrix}
\tilde S_c & 0_{\tilde p \times \tilde p} \\
0_{\tilde p \times \tilde p} & 0_{\tilde p \times \tilde p}
\end{pmatrix} \begin{pmatrix}
\tilde S_t + \tilde S_c & \tilde S_t \\
\tilde S_t & \tilde S_t
\end{pmatrix}^{-1} \begin{pmatrix}
\overline{\tilde X}_t \\
\overline{\tilde X}_t
\end{pmatrix} \\
&=& \begin{pmatrix}
\tilde S_c & 0_{\tilde p \times \tilde p} \\
0_{\tilde p \times \tilde p} & 0_{\tilde p \times \tilde p}
\end{pmatrix} \begin{pmatrix}
\tilde S_c^{-1} & -\tilde S_c^{-1} \\
-\tilde S_c^{-1} & \tilde S_c^{-1}+\tilde S_t^{-1}
\end{pmatrix} \begin{pmatrix}
\overline{\tilde X}_t \\
\overline{\tilde X}_t
\end{pmatrix} \\
&=& \begin{pmatrix}
I & -I \\
0_{\tilde p \times \tilde p} & 0_{\tilde p \times \tilde p}
\end{pmatrix} \begin{pmatrix}
\overline{\tilde X}_t \\
\overline{\tilde X}_t
\end{pmatrix} = \begin{pmatrix}
0_{\tilde p} \\
0_{\tilde p}
\end{pmatrix}.
\end{eqnarray*}
Similarly, we also have $\underline{\tilde S}_t (\underline{\tilde S}_t + \underline{\tilde S}_c)^{-1}  \underline{\overline{\tilde X}}_c = 0_{2\tilde p}$. Recall the definition of $\underline{\tilde {X}}_{ij}$, the first constraint in~\eqref{eq::thm1_one_stage} is then equivalent to the first two constraints in~\eqref{eq:framework0} as
\begin{eqnarray*}
&& \sum_{i=1}^m \sum_{j:Z_{ij}=z} w_{ij}  \underline{\tilde {X}}_{ij}  =  0_{2\tilde p} \\
&\Leftrightarrow&
\sum_{i=1}^m \sum_{j:Z_{ij}=z} w_{ij}  X_{ij,l}  =  0, \quad  \text{for $l \in C$.} \\
&& 
\sum_{j:Z_{ij}=z} w_{ij}  X_{ij,l}  =  0, \quad  \text{for $l \in F$ and $i = 1, \ldots, m$}
\end{eqnarray*}
This finishes the proof.
\end{proof}

\subsubsection{Proof of Theorem~\ref{thm:bounded}}
We first establish a theorem that derives the dual formulation and characterizes the solutions to the general optimization problem~\eqref{eq:fr_theory_c}. 
Recall that for notational simplicity, we include the constant basis function $B_1(x) =1$, set the target value $B^*_1 =1$ and the corresponding tolerance $\delta_1=0$, so that the first balancing constraint enforces normalization, i.e. $\sumijz w_{z,ij} =1$.
Therefore, the optimization problem~\eqref{eq:PBM0} used in Section~\ref{sec:theory} is a simplified case of~\eqref{eq:fr_theory_c}, with the scaling factors set to one, i.e., $c_{ij}=1$ for all $ij$.

\begin{equation}
\label{eq:fr_theory_c}
{\hat{w}^\pb_{z,ij} = \arg\min_{w} \left\{ \sum_{i=1}^m \sum_{j: Z_{ij} = z} c_{ij} \psi(w_{ij}) : \, \left| \sum_{i=1}^m \sum_{j:Z_{ij}=z} \hspace{-.175cm} w_{ij} B_k(X_{ij}) - B_k^* \right| \leq \delta_k, \, \forall k; \, w_{ij} \geq 0 \right\}}.
\end{equation}

The proof follows from \citet{chattopadhyay2024one}{Theorem 8.1}.
\begin{theorem}
\label{thm:dual}
\begin{itemize}
    \item [(a)] The dual problem of~\eqref{eq:fr_theory_c} is equivalent to the empirical loss minimization problem with $L_1$ regularization:
    \begin{eqnarray}
    \label{eq:d1}
        \min_{u \geq 0, \lambda}  \sum_{ij \in \pq} \left[ - \indz S_{ij} c_{ij} \rho \{(B(X_{ij})^\T \lambda - \uij)/c_{ij} \} 
        \right] + (B^*) ^\T \lambda + |\lambda|^\T \delta 
    \end{eqnarray}
    \item [(b)] If $\hat w^\pb_z$ and $\lambda_z^\dag$ are solutions to the primal and dual forms of~\eqref{eq:PBM0}, respectively, then for $ij: Z_{ij}=z$,
    \begin{eqnarray}
    \label{eq:d2}
    \hat w^\pb_{z,ij} = \rho^\prime \{B(X_{ij})^\T \lambda_z^\dag/c_{ij} \} 1\{\rho^\prime \{B(X_{ij})^\T \lambda_z^\dag/c_{ij} \} > 0\}
    \end{eqnarray}
\end{itemize}
\end{theorem}

\begin{proof}[Proof of Theorem~\ref{thm:dual}]
The $k$th balancing constraint in the optimization problem~\eqref{eq:PBM0} can be written as
\begin{eqnarray*}
&& \left | \sumz w_{ij}  B_k(x_{ij})  -  B_k^* \right | \leq \delta_k \\
&\Leftrightarrow & \left | \sum_{ij \in \pq} \indz S_{ij} w_{ij}  B_k(x_{ij})  - (n^*)^{-1} \sum_{ij \in \pq}(1-S_{ij})  B_k(X_{ij}) \right | \leq \delta_k \\
&\Leftrightarrow & \left | \sum_{ij \in \pq} \{ (n^*)^{-1} (1-S_{ij}) - \indz S_{ij} w_{ij}  \} B_k(x_{ij})   \right | \leq \delta_k \\
&\Leftrightarrow & \left | \sum_{ij \in \pq} \xi_{ij} B_k(x_{ij})   \right | \leq \delta_k 
\end{eqnarray*}
with $\xi_{ij} =(n^*)^{-1} (1-S_{ij}) - \indz S_{ij} w_{ij}$. 
Thus, for the units in group $Z_{ij}=z$ of the sample, $w_{ij} = - \xi_{ij}$.
Then the sample boundedness constraint can be written as
\begin{eqnarray*}
 w_{ij} \geq 0 \text{ for all $Z_{ij} = z$} 
\Leftrightarrow  \indz S_{ij} w_{ij} \geq 0 \text{ for all $ij \in \pq$}  
\Leftrightarrow  S_{ij}\xi_{ij} \leq 0 \text{ for all $ij \in \pq$}.
\end{eqnarray*}
The objective function can be written as
\begin{eqnarray*}
\sum_{i=1}^m \sum_{j: Z_{ij} =z} c_{ij}  \psi(w_{ij})
= \sum_{ij \in \pq} \indz S_{ij } c_{ij} h(\xi_{ij})
\end{eqnarray*}
with $h(x) = \psi(-x)$. 

In our nested data structure, each individual is indexed by a pair $ij$. To simplify the representation, we assign a unique global index $q$ to each individual, effectively flattening the nested structure into a single sequence. The global index $q \in \{1,2, \ldots, n\}$ maps to $ij$ via $i = \min\{i: q \leq \sum_{i^\prime=1}^{i} n_{i^\prime} \}$, $j = q - \sum_{i^\prime=1}^{i-1} n_{i^\prime}$. Conversely, the pair index $ij$ maps back to $q$ as $q = j + \sum_{i^\prime=1}^{i-1} n_{i^\prime}$. For clarity, we will use global index $q$ in the following proof.

Let $\underline{A}$ be a $K \times n$ matrix whose $(l,q)$th element is $B_l(X_{q})$; 
$\underline{Q} = \left(\underline{A}^\T, -\underline{A}^\T \right)^\T$;
and $d = (\delta^\T, \delta^\T)^\T$. 
Let $S$ be the vector of $S_{q}$; $\underline{S} = I_{n+n^*} S$.
We can write the primal problem as
\begin{eqnarray}
\label{eq:dual_prob}
\text{min}_w & \sum_{ij \in \pq} \indz S_{ij} c_{ij} h(\xi_{ij}) & \\
\text{s.t.} 
& \underline{Q}\xi \leq d & \nonumber \\
&  \underline{S}\xi \leq 0_{n+n^*} &  \nonumber
\end{eqnarray}
This gives us a convex optimization problem in $\xi$ with linear constraints. 
The Lagrange dual function of the primal problem in Equation~\eqref{eq:dual_prob} is given by
\begin{eqnarray*}
\inf_{\xi} \left\{ \sum_{ij \in \pq} \indz S_{ij} c_{ij} h(\xi_{ij}) + \lambda^\T \underline{Q}\xi + u^\T \underline{S}\xi \right\} - \lambda^\T d.
\end{eqnarray*}
Let $Q_{\cdot,ij}$ be the column of $\underline{Q}$ that corresponds to patient $ij$. The dual objective function can be written as
\begin{eqnarray*}
&& -\sup_\xi \left\{ -\sum_{ij \in \pq} \indz S_{ij} c_{ij} h(\xi_{ij}) - \lambda^\T \underline{Q}\xi - u^\T \underline{S}\xi \right\}- \lambda^\T d \\
&=& -\sum_{ij \in \pq} c_{ij} \sup_{\xi_{ij}} \left\{ - \indz S_{ij}  h(\xi_{ij}) - ( Q_{\cdot,ij}^\T \lambda) \xi_{ij}/c_{ij} - u_{ij}S_{ij}\xi_{ij}/c_{ij} \right\}- \lambda^\T d \\
&=& -\sum_{ij \in \pq} c_{ij} h_{ij}^* (-(Q_{\cdot,ij}^\T \lambda + u_{ij}S_{ij})/c_{ij})  - \lambda^\T d,
\end{eqnarray*}
where $h_{ij}^*(\cdot)$ is the convex conjugate of $\indz S_{ij} h(\cdot)$. Thus, the dual problem is given by
\begin{eqnarray}
\text{max}_{\lambda, u} & -\sum_{ij \in \pq} h_{ij}^* (-(Q_{\cdot,ij}^\T \lambda + u_{ij}S_{ij})/c_{ij})  - \lambda^\T d & \\
\text{s.t.} 
& \lambda \geq 0 & \nonumber \\
&  u \geq 0 &  \nonumber
\end{eqnarray}
Since the last $K$ components of $Q_{\cdot,ij}$ are reflected versions of the first $K$ components, by symmetry
we can write the dual problem as
\begin{eqnarray}
\label{eq:dual_prob2}
\text{max}_{\lambda, u} & -\sum_{ij \in \pq} h_{ij}^* ((Q_{\cdot,ij}^\T \lambda - u_{ij}S_{ij})/c_{ij})  - \lambda^\T d & \\
\text{s.t.} 
& \lambda \geq 0 & \nonumber \\
&  u \geq 0 &  \nonumber
\end{eqnarray}
For the convex conjugate,
\begin{eqnarray*}
h_{ij}^*(v) &=& \sup_{\xi_{ij}} \left\{ - \indz S_{ij} h(\xi_{ij}) + v \xi_{ij}  \right\}\\
&=& \left\{ - \indz S_{ij} h(\hat \xi_{ij}) + v \hat \xi_{ij}  \right\},
\end{eqnarray*}
where $\hat \xi_{ij}$ for $ij \in \pp$ satisfies
\begin{eqnarray*}
\partial/\partial \xi_{ij} \left\{ - \indz S_{ij} h(\hat \xi_{ij}) + v \hat \xi_{ij}  \right\} = 0,
\end{eqnarray*}
and $\hat \xi_{ij} = 1/n^*$ for $ij \in \pt$.
This gives $\hat \xi_{ij}(v) = (h^\prime)^{-1}(v)$ and $\hat w_{z,ij} (v) = -(h^\prime)^{-1}(v)$ for $ij$ with $\indz S_{ij}=1$. Recall we define $\rho(v) = -v(h^\prime)^{-1}(v) + h \{ (h^\prime)^{-1}(v) \}$ and $\rho^\prime(v) = -(h^\prime)^{-1}(v)$. We can further write, for $ij$ with $\indz S_{ij}=1$,
\begin{eqnarray*}
\hat w^\pb_{z,ij} (v) = \rho^\prime(v), \quad
h_{ij}^*(v) = -\rho(v).
\end{eqnarray*}
The objective function then boils down to
\begin{eqnarray*}
&&\sum_{ij \in \pq} - \indz S_{ij} c_{ij} \rho ((Q_{\cdot,ij}^\T \lambda - u_{ij}S_{ij})/c_{ij}) + (1-S_{ij}) \frac{1}{n^*}(Q_{\cdot,ij}^\T \lambda - u_{ij}S_{ij}) \\
&=& \sum_{ij \in \pq} - \indz S_{ij} c_{ij} \rho ((Q_{\cdot,ij}^\T \lambda - u_{ij}S_{ij})/c_{ij})) + \frac{1}{n^*}(1-S_{ij}) Q_{\cdot,ij}^\T \lambda
\end{eqnarray*}
\begin{eqnarray}
\label{eq:dual_prob3}
\text{min}_{\lambda, u} & \sum_{ij \in \pq} \left\{ - \indz S_{ij} c_{ij} \rho ((Q_{\cdot,ij}^\T \lambda - u_{ij}S_{ij})/c_{ij})) + \frac{1}{n^*}(1-S_{ij}) Q_{\cdot,ij}^\T \right\} + \lambda^\T d & \\
\text{s.t.} 
& \lambda \geq 0 & \nonumber \\
&  u \geq 0 &  \nonumber
\end{eqnarray}

This dual problem can be further simplified.
Following the proof of Theorem 1 in \citep{wang2020minimal}, write $\lambda = (\lambda_+^\T, \lambda_{-}^\T)^\T$, where $\lambda_+$ and $\lambda_{-}$ are $K \times 1$ vectors. Denoting $A_{\cdot,ij}$ as the column of $\underline{A}$ that corresponds to patient $ij$, we write the dual objective as
\begin{eqnarray*}
&& g(\lambda,u) \\
&=& \sum_{ij \in \pq} \left\{ - \indz S_{ij} c_{ij} \rho \{ (A_{\cdot,ij}^\T (\lambda_{+} - \lambda_{-}) - u_{ij}S_{ij})/c_{ij} \} + \frac{1}{n^*}(1-S_{ij}) A_{\cdot,ij}^\T (\lambda_{+} - \lambda_{-}) \right\}\\
&& + (\lambda_{+} + \lambda_{-} )^\T d.
\end{eqnarray*}
Denote $\lambda^\dag = (\lambda_+^{\dag\T}, \lambda_{-}^{\dag \T})^\T$ as the dual solution. 
We can show that at least one of $\lambda_{+,l}, \lambda_{-,l}$ equals zero for all $l$ s.t. $\delta_l > 0$.
Assume the $l$th component of $\lambda_+^{\dag\T}$ and $\lambda_{-}^{\dag\T}$ are both strictly positive, for some $l \in \{1,2,\ldots,K\}$. Then define
\begin{eqnarray*}
\lambda^{\dag \dag} = (\lambda_+^{\dag\T} - (0, 0, \ldots, \min(\lambda_{+,l}^{\dag}, \lambda_{-,l}^{\dag}), 0, \ldots, 0 )^\T,
\lambda_-^{\dag\T} - (0, 0, \ldots, \min(\lambda_{+,l}^{\dag}, \lambda_{-,l}^{\dag}), 0, \ldots, 0 )^\T)^\T.
\end{eqnarray*}
Notice that $g(\lambda^{\dag \dag},u) = g(\lambda^{\dag },u) - 2 \delta_l \min(\lambda_{+,l}^{\dag}, \lambda_{-,l}^{\dag}) < g(\lambda^{\dag },u)$ since $\delta_l > 0$. This leads to a contradiction since $\lambda^\dag$ minimizes $g(\lambda,u)$. Thus,
at least one of $\lambda_{+,l}, \lambda_{-,l}$ equals zero. 
Moreover, $\lambda_l \delta_l = 0$ for all $l$ s.t. $\delta_l=0$.
Further notice that $(n^*)^{-1}\sum_{ij \in \pq} (1-S_{ij}) B(X_{ij})^\T \lambda = (B^*) ^\T \lambda$
The dual problem then boils down to
\begin{eqnarray}
\label{eq:dual_prob4}
\text{min}_{\lambda, u} & \sum_{ij \in \pq}  - \indz S_{ij} c_{ij} \rho ((B(X_{ij})^\T \lambda - u_{ij}S_{ij})/c_{ij}) + (B^*) ^\T \lambda + |\lambda|^\T d & \\
\text{s.t.}  
&  u \geq 0 &  \nonumber
\end{eqnarray}
This proves part (a) of Theorem~\ref{thm:dual}.

Suppose that $\lambda_z^\dag$ and $u_z^\dag$ are solutions to the dual forms. 
For $ij: Z_{ij}=z, S_{ij}=1$, we have
\begin{eqnarray*}
\hat w^\pb_{z,ij} =   \rho^\prime\{(B(X_{ij})^\T \lambda_z^\dag - u_{z, ij}^\dag)/c_{ij} \}.
\end{eqnarray*}
The complementary slackness condition gives that $S_{ij} \indz u_{ij}w_{ij} = 0$. Thus,
\begin{eqnarray*}
\hat w^\pb_{z,ij} =   \rho^\prime\{ B(X_{ij})^\T \lambda_z^\dag /c_{ij} \} 1\{ u_{z,ij}^\dag =0 \}.
\end{eqnarray*}
Note that $w_{ij}, u_{ij} \geq 0$ are guaranteed by constraints in the primal and dual problems. Then $\{ u_{z, ij}^\dag = 0  \} \Rightarrow \rho^\prime\{B(X_{ij})^\T \lambda_z^\dag/c_{ij}  \} \geq 0$. Moreover, $\{u_{z,ij}^\dag = 0 \}^c \Rightarrow \{u_{z,ij}^\dag > 0 \} \Rightarrow \{ \rho^\prime\{(B(X_{ij})^\T \lambda_z^\dag - u_{z,ij}^\dag)/c_{ij}  \} = 0\} \Rightarrow \{ \rho^\prime\{B(X_{ij})^\T \lambda_z^\dag /c_{ij}  \} < 0\}$ because 
$$
\rho^{\prime \prime} (v) = -1/h^{\prime \prime} (v) =  -1/\psi^{\prime \prime} (v) < 0
$$
as the objective function $\psi(\cdot)$ is convex.
This proves part (b) of Theorem~\ref{thm:dual} as
\begin{eqnarray*}
    \hat w^\pb_{z,ij} = \rho^\prime \{B(X_{ij})^\T \lambda_z^\dag /c_{ij} \} 1\{\rho^\prime \{B(X_{ij})^\T \lambda_z^\dag /c_{ij} \} > 0\}.
\end{eqnarray*}
This also implies that
\begin{eqnarray*}
u^\dag_{ij} =
\left\{ 
\begin{aligned}
& B(X_{ij})^\T \lambda_z^\dag - c_{ij}(\rho^\prime)^{-1}(0), &  B(X_{ij})^\T \lambda_z^\dag > c_{ij}(\rho^\prime)^{-1}(0)\\
&0, & \text{o.w.}
\end{aligned}
\right. 
\end{eqnarray*}

\end{proof}

\begin{proof}[Proof of Theorem~\ref{thm:bounded}]
Construct the selected set as $R = \{ ij: \hat w^\pb_{Z_{ij},ij} > 0 \}$. 
Since $\hat w^\pb_{Z_{ij},ij} = 0$ for $ij \not\in R$, the solutions to the original problem~\eqref{eq:PBM0} are also solutions to the restricted optimization problem~\eqref{eq:fr_R+} over the selected set $R$.
\begin{eqnarray}
\label{eq:fr_R+}
\text{min}_w & \sum_{ij \in R: Z_{ij} =z}  \psi(w_{ij}) & \\
\text{s.t.} &  &\nonumber\\
& \left | \sum_{ij \in R: Z_{ij} =z} w_{ij}  B_k(x_{ij})  -  B_k^* \right | \leq \delta_k, \quad & \text{for $k = 1, \ldots, K$} \nonumber  \\
& w_{ij} \geq 0  & \text{for $ij \in R$ s.t. $Z_{ij}=z$}. \nonumber
\end{eqnarray}
According to Theorem~\ref{thm:dual}, the dual formulation of~\eqref{eq:fr_R+} is given by
\begin{eqnarray*}
\min_{\lambda, u}  \sum_{ij \in R} \left[ - \indz \rho \{B(X_{ij})^\T \lambda - u_{ij} \}
\right] + (B^*) ^\T \lambda + |\lambda|^\T \delta.
\end{eqnarray*}
Let $\hat w^{*+|R}_{z}$, $\lambda_{z}^{\dag +}$ and $u_{z,ij}^\dag$ denote solutions to the primal and dual forms of~\eqref{eq:fr_R+}. 
From the proof of Theorem~\ref{thm:dual}, it follows that $u_{Z_{ij}, ij}^\dag = 0 $ for all $ij \in R$. Thus, 
\begin{eqnarray}
\label{eq:dual_R+}
\lambda_{z}^{\dag +} = 
\argmin_{\lambda} \sum_{ij \in R} \left[ -\indz \rho \{B(X_{ij})^\T \lambda \}
\right] + (B^*) ^\T \lambda + |\lambda|^\T \delta,
\end{eqnarray}
and the weights are given by $\hat w^{*+|R}_{z, ij} = \rho^\prime \{B(X_{ij})^\T \lambda_{z}^{\dag +} \} 1\{u_{ij}^\dag = 0\} = \rho^\prime \{B(X_{ij})^\T \lambda_{z}^{\dag +} \}$ for all $ij \in R$.

Framework~\eqref{eq:fr_R} below shows the optimization problem without non-negativity constraint and restricted to the selected set $R$. 
\begin{eqnarray}
\label{eq:fr_R}
\text{min}_w & \sum_{ij \in R: Z_{ij} =z}  \psi(w_{ij}) & \\
\text{s.t.} &  &\nonumber\\
& \left | \sum_{ij \in R: Z_{ij} =z} w_{ij}  B_k(x_{ij})  -  B_k^* \right | \leq \delta_k, \quad & \text{for $k = 1, \ldots, K$} \nonumber 
\end{eqnarray}
Similarly to Theorem~\ref{thm:dual}, it can be shown that the dual problem of~\eqref{eq:fr_R} is given by
\begin{eqnarray*}
\label{eq:dual_R}
\min_{\lambda}  \sum_{ij \in R} \left[ - \indz \rho \{B(X_{ij})^\T \lambda \}
\right] + (B^*) ^\T \lambda + |\lambda|^\T \delta.
\end{eqnarray*}
Moreover, the solutions have expression $\hat w^{*-|R}_{z, ij} = \rho^\prime \{B(X_{ij})^\T \lambda_{z}^{\dag-} \}$, where $\hat w^{*-|R}_{z, ij}$ and $\lambda_{z}^{\dag-}$ are solutions to the primal and dual forms of~\eqref{eq:fr_R}. 

It suffices to show that~\eqref{eq:fr_R} and~\eqref{eq:fr_R+} yield the same solutions.
Given the dual of~\eqref{eq:fr_R} and equation~\eqref{eq:dual_R+}, we have $\lambda_{z}^{\dag +} = \lambda_{z}^{\dag -}$. Combining this with the expressions for the solutions proves Theorem~\ref{thm:bounded}.




\end{proof}

\subsubsection{Proof of Theorem~\ref{thm::consistency}}
We first introduce the regularity conditions for asymptotic results.
\begin{appassumption}\label{cond::regularity}
Assume the following conditions hold for $z \in \{0,1\}$
\begin{itemize}
    \item[(a)] 
    There exist constants $C_0, C_1, C_2$ with $C_0 > 0$ and $C_1 < C_2 < 0$, such that $C_1 \leq n \rho^{\prime\prime}(v) \leq C_2$ for all $v$ in a neighborhood of $B(x)^\T \tilde\lambda_{1z}$. Also, $|n \rho^\prime(v)| \leq C_0$ and $\rho^\prime(v) \neq 0$ a.s. for all $v = B(x)^\T \lambda, x \in \X, \lambda$. 
    \item[(b)] $\sup_{x \in \X} \Vert B(x) \Vert_2 \leq CK^{1/2}$ and $\Vert \E_{\dq} \{ B(X)B(X)^\T\} \Vert_F \leq C$ for some $C > 0$, where $\Vert \cdot \Vert_F$ denotes the Frobenius norm.
    \item[(c)] $K = O \{ n ^\alpha \}$ for some $0 < \alpha < 2/3$.
    \item[(d)] The smallest eigenvalue, $\nu_{\min}$, of $\E_{\dq} [B(X)B(X)^\T]$ satisfies $\nu_{\min} > C$ for some constant $C>0$.
    \item[(e)] $\Vert \delta \Vert_2 = O_P[K^{1/4} \{(\log K)/n \}^{1/2}]$ .
    \item[(f)] The ratio $n^*/n$ approaches a finite positive constant as $n \rightarrow \infty$.
    \item[(g)] $\E_{\dt} \{ Y^2(z) \} < \infty$.
\end{itemize}
\end{appassumption}

Under Assumption~\ref{cond::inverse_prop}, for $Z=z \in \{0,1\}$, we have $\tilde w(x) = n \rho^\prime\{ B(x)^\T \tilde\lambda_{1z}\} 1\{ \rho^\prime\{ B(x)^\T \tilde\lambda_{1z}\} \geq 0 \}$.
Theorem~\ref{thm::unif_convg_prob} shows that when probability weight model is correctly specified, the weights uniformly converge to the true probability weighting factor.
\begin{theorem}
\label{thm::unif_convg_prob}
Under Assumption~\ref{cond::1} and~\ref{cond::regularity}, if Assumption~\ref{cond::inverse_prop} hold for $z$, the weights in group $Z=z$ satisfy
\begin{eqnarray*}
\sup_{x \in \X} \left| n \hat w^\pb_z(x) - \tilde w_z(x) \right| 
= O_P(K^{3/4} (\log K)^{1/2}n^{-1/2}) = o_P(1).
\end{eqnarray*}
\end{theorem}
Below we provide a proof of of this result. All probabilities and expectations in this proof are
computed with respect to the probability measure $\dq$.

Let $\lambda_z^\dag$ denote the solution to the dual form of~\eqref{eq:PBM0}. We have 
\begin{eqnarray}
&& \sup_{x \in \X} \left| n \hat w^\pb_z(x) - \tilde w_z(x) \right| \nonumber \\
&=& \sup_{x \in \X} \left| n \rho^\prime\{ B(x)^\T \lambda_z^\dag \} 1\{ \rho^\prime\{ B(x)^\T \lambda_z^\dag \} \geq 0 \} - n \rho^\prime\{ B(x)^\T \tilde\lambda_{1z}\} 1\{ \rho^\prime\{ B(x)^\T \tilde\lambda_{1z}\} \geq 0 \} \right| \nonumber \\
&\leq&  \sup_{x \in \X} \left| n \rho^\prime\{ B(x)^\T \lambda_z^\dag \}  - n \rho^\prime\{ B(x)^\T \tilde\lambda_{1z}\}  \} \right| \nonumber \\
& \leq & C \sup_{x \in \X} \left| B(x)^\T(\lambda_z^\dag -\tilde\lambda_{1z})  \right| \nonumber \\
&\leq & C K^{1/2} \Vert \lambda_z^\dag -\tilde\lambda_{1z}  \Vert_2
\label{eq:wdiff_bound}
\end{eqnarray}
The first equality is due to Assumption~\ref{cond::inverse_prop} and Theorem~\ref{thm:dual}(b).
The first inequality obviously hold when the two terms have the same sign, i.e., $\rho^\prime\{ B(x)^\T \lambda_z^\dag\} \rho^\prime\{ B(x)^\T \tilde\lambda_{1z}\}  \geq 0$.
When the two terms have opposite signs, the inequality follows from the fact that $|a|,|b| \leq |a-b|$ when $ab < 0$. 
The second inequality follows from the mean value theorem and Assumption~\ref{cond::regularity}(a).
The final inequality follows from the Cauchy--Schwarz inequality and Assumption~\ref{cond::regularity}(b).

Therefore, to bound the difference between estimated weights and true probability weight factors, it suffices to bound $\Vert \lambda_z^\dag -\tilde\lambda_{1z}  \Vert_2$. 
Following the proof structure of Lemma 8.3 in \citet{chattopadhyay2023implied}, we establish Lemma~\ref{lem::onestep8.5}--~\ref{lem::onestep8.6} to control the bound.

\begin{lemma}
\label{lem::mat_inequal}
Let \( \underline{W}_1, \underline{W}_2, \ldots, \underline{W}_{n} \) be \( d_1 \times d_2 \) independent random matrices with \( E[\underline{W}_j] = 0 \) and \( \|\underline{W}_j\|_2 \leq R_{n} \) a.s., where \( \|\cdot\|_2 \) denotes the spectral norm. Let
\[
\sigma_{n}^2 := \max \left\{ \left\|\sum_{j=1}^{n} \E[\underline{W}_j \underline{W}_j^\top] \right\|_2, \left\|\sum_{j=1}^{n} \E[\underline{W}_j^\top \underline{W}_j] \right\|_2 \right\}.
\]
Then for all \( t \geq 0 \), we have
\[
\Pr \left( \left\|\sum_{j=1}^{n} \underline{W}_j \right\|_2 \geq t \right) \leq (d_1 + d_2) \exp\left( \frac{-t^2/2}{\sigma_{n}^2 + (R_{n}t/3)} \right).
\]
\end{lemma}
\begin{proof}[Proof of Lemma~\ref{lem::mat_inequal}]
See \citet{tropp2015introduction}.
\end{proof}

\begin{lemma}
\label{lem::onestep8.5}
If Assumption~\ref{cond::1} and~\ref{cond::regularity} hold, we have
$$\Vert  \sumq \{ (n^*)^{-1}(1-S_{ij}) - n^{-1}S_{ij}1\{Z_{ij}=z\} \tilde w(X_{ij}) \} B(X_{ij}) \Vert_2 = O_P( K^{1/4}(\log K)^{1/2} n^{-1/2}).$$
\end{lemma}
\begin{proof}[Proof of Lemma~\ref{lem::onestep8.5}]
We will use Lemma~\ref{lem::mat_inequal} to prove this. Let us denote
\begin{eqnarray*}
\underline{W}_{ij} = \{(n^*)^{-1}(1-S_{ij}) - n^{-1}S_{ij}1\{Z_{ij}=z\} \tilde w(X_{ij})\} B(X_{ij}), \quad \text{for } ij \in \pq.    
\end{eqnarray*}

First, by Assumption~\ref{cond::1} $\E_{\dq}(\underline{W}_{ij}) = 0$. Second, we have
\begin{eqnarray*}
\Vert \underline{W}_{ij} \Vert_2 &=&  | (n^*)^{-1}(1-S_{ij}) - n^{-1} S_{ij}1\{Z_{ij}=z\} \tilde w(X_{ij}) \} | \times \Vert B(X_{ij}) \Vert_2 \\
&\leq& (\min(n,  n^*))^{-1} \left[ 1 + | n \rho^\prime\{B(X_{ij})^\T \tilde \lambda_{1z}  \}|  \right] \sup_{x \in \X} \Vert B(x) \Vert_2 \\
&\leq& (\min(n,  n^*))^{-1}CK^{1/2}.
\end{eqnarray*}
The second inequality follows from Assumption~\ref{cond::regularity}(a)--(b). 

Next, we consider
\begin{eqnarray}
&& \left\Vert \sumq \E \left(  \underline{W}_{ij}^\T \underline{W}_{ij} \right)  \right\Vert_2 \nonumber \\
&=& \sumq \E \left\{  \left\{ (n^*)^{-1}(1-S_{ij}) - n^{-1} S_{ij}1\{Z_{ij}=z\} \tilde w(X_{ij}) \right\}^2 B(X_{ij})^\T B(X_{ij})  \right\} \nonumber \\
&\leq& C(\min(n,  n^*))^{-2} (n+n^*) \tr \left[ \E_{\dq} \left\{  B(X)^\T B(X)  \right\} \right]
\label{eq:sumEl2}
\end{eqnarray}
The inequality follows from Assumption~\ref{cond::inverse_prop} and ~\ref{cond::regularity}(a).
Now, let $ \nu_1, \ldots, \nu_K$  be the eigenvalues of a non-negative definite matrix $A$. By the Cauchy--Schwarz inequality,
$$
\tr(A) \leq K^{1/2} (\nu_1^2 + \ldots + \nu_K^2)^{1/2} = K^{1/2} \Vert A \Vert_F.
$$
Thus, from Equation~\eqref{eq:sumEl2}, we get
\begin{eqnarray*}
\left\Vert \sumq \E \left(  \underline{W}_{ij}^\T \underline{W}_{ij} \right)  \right\Vert_2
\leq C K^{1/2} \frac{n+n^*}{(\min(n,  n^*))^{2}} \left\Vert \E_{\dq} \left\{  B(X)^\T B(X)  \right\} \right\Vert_F   
\leq C^\prime K^{1/2} \frac{n+n^*}{(\min(n,  n^*))^{2}}.
\end{eqnarray*}
Here the second inequality follows from Assumption~\ref{cond::regularity}(b).
Moreover
\begin{eqnarray}
&& \left\Vert \sumq \E \left(  \underline{W}_{ij} \underline{W}_{ij}^\T \right)  \right\Vert_2 \nonumber \\
&=& \sumq \left\Vert \E \left\{ \left\{ (n^*)^{-1} (1-S_{ij}) - n^{-1} S_{ij}1\{Z_{ij}=z\} \tilde w(X_{ij}) \right\}^2 B(X_{ij}) B(X_{ij})^\T  \right\} \right\Vert_2 \nonumber \\
&\leq& C (\min(n,  n^*))^{-2} (n+n^*) \left\Vert \E_{\dq} \left\{  B(X) B(X)^\T  \right\} \right\Vert_2 \nonumber \\
&\leq& C(\min(n,  n^*))^{-2} (n+n^*)\left\Vert \E_{\dq} \left\{  B(X) B(X)^\T  \right\} \right\Vert_F \leq C^{\prime \prime} (\min(n,  n^*))^{-2}(n+n^*).
\label{eq:sumEl2.2}
\end{eqnarray}
Here the first inequality is due to the triangle inequality. The second inequality holds by bounding $\left\{ (1-S_{ij}) - S_{ij}1\{Z_{ij}=z\} \tilde w(X_{ij}) \right\}^2$ as before.
The second inequality holds because spectral norm is dominated by the Frobenius norm.
The final inequality again follows from Assumption~\ref{cond::regularity}(b).

Therefore, combining equations~\eqref{eq:sumEl2} and~\eqref{eq:sumEl2.2} we get
\begin{eqnarray*}
\sigma^2_{n+n^*} := \max \left\{ 
\left\Vert \sumq \E_{\dt} \left(  \underline{W}_{ij}^\T \underline{W}_{ij} \right)  \right\Vert_2,
\left\Vert \sumq \E_{\dt} \left(  \underline{W}_{ij} \underline{W}_{ij}^\T \right)  \right\Vert_2 \right\} \leq C K^{1/2} \frac{n+n^*}{(\min(n,  n^*))^{2}}.
\end{eqnarray*}
Applying Lemma~\ref{lem::mat_inequal}, we have
\begin{eqnarray}
\label{eq:inequal_W}
&& \Pr \left( \left\Vert \sumq \underline{W}_{ij} \right\Vert_2 \geq t \right) \nonumber \\
&\leq& (K+1) \exp \left[ \frac{t^2/2}{CK^{1/2}\frac{n+n^*}{(\min(n,  n^*))^{2}} + C^\prime K^{1/2}t \frac{n+n^*}{3(\min(n,  n^*))^{2}} } \right]
\end{eqnarray}
By Assumption~\ref{cond::regularity}(c), the right hand side of equation~\eqref{eq:inequal_W} goes to 0 if $$t = \bar C K^{1/4}(\log K)^{1/2} \frac{\sqrt{n+n^*}}{\min(n,  n^*)} \quad \text{for some $\bar C >0$.}$$ This, together with Assumption~\ref{cond::regularity}(f), gives $\left\Vert \sumq \underline{W}_{ij} \right\Vert_2 
= O_P( K^{1/4}(\log K)^{1/2} n^{-1/2})$. 
\end{proof}

\begin{lemma}
\label{lem::onestep8.6}
With probability tending to one, $\nu_{\min} \left( \sum_{ij \in \pq: S_{ij}Z_{ij}=1} n^{-1} B(X_{ij}) B(X_{ij})^\T \right) > \tilde C$ for some constant $\tilde C >0$.
\end{lemma}
\begin{proof}[Proof of Lemma~\ref{lem::onestep8.6}]
See~\citet[Lemma 8.6]{chattopadhyay2024one}.
\end{proof}

We then combine Lemma~\ref{lem::onestep8.5} and~\ref{lem::onestep8.6} to bound the term $\Vert \lambda_z^\dag -\tilde\lambda_{1z}  \Vert_2$ in~\eqref{eq:wdiff_bound}.
\begin{lemma}
\label{lem:lambda_diff}
There exists a dual solution $\lambda_z^\dag$ s.t. 
$\Vert \lambda^\dag_z - \tilde\lambda_{1z} \Vert_2 = O_p \left( K^{1/4} (\log K)^{1/2}n^{-1/2} \right)$.
\end{lemma}
\begin{proof}[Proof of Lemma~\ref{lem:lambda_diff}]
Following the proof structure of \citealt{chattopadhyay2024one,fan2016improving,wang2020minimal}. All the subsequent probabilities and expectations are taken with respect to $\dq$.

Set $r = C^* \left\{ K^{1/4} (\log K)^{1/2}n^{-1/2} \right\}$ for a sufficiently large constant $C^* >0$. Let $\Delta = \lambda - \tilde\lambda_{1z}$. Also, set $C = \{ \Delta \in \R^K: \Vert \Delta \Vert_2 \leq r  \}$. To show that there exists a dual solution $\lambda_z^\dag$ s.t. $\Vert \lambda^\dag_z - \tilde\lambda_{1z} \Vert_2 = O_p \left( K^{1/4}\{(\log K)/n \}^{1/2} \right)$, it suffices to show that there exists a $\Delta^\dag \in \R^K$ s.t.$\Pr(\Delta^\dag \in C) \rightarrow 1$ as $n \rightarrow \infty$.

Recall from Theorem~\ref{thm:dual} the dual objective of~\eqref{eq:PBM0} is given by
\begin{eqnarray*}
g(\lambda, u) = \sum_{ij \in \pq} \left[ - \indz S_{ij} \rho \{B(X_{ij})^\T \lambda  - \uij \} \right] + (B^*) ^\T \lambda + |\lambda |^\T \delta.
\end{eqnarray*}
It follows from the proof of Theorem~\ref{thm:dual}
that given $\lambda$,
\begin{eqnarray*}
\argmin_{u_{ij}} g(\lambda, u) =
\left\{ 
\begin{aligned}
& B(X_{ij})^\T \lambda - (\rho^\prime)^{-1}(0), &  \rho^\prime(B(X_{ij})^\T \lambda) > 0\\
&0, & \text{o.w.}
\end{aligned}
\right. 
\end{eqnarray*}
Denote $G(\lambda):= \min_{u \geq 0} g(\lambda, u)$, we then have
\begin{eqnarray*}
G(\lambda) 
&=& \sum_{ij \in \pq} \left[ \indz S_{ij} \tilde G_{ij}(\lambda) \right] + (B^*) ^\T \lambda + |\lambda |^\T \delta, \\
\tilde G_{ij}(\lambda) &=& 
-1\{\rho^\prime \{B(X_{ij})^\T \lambda \} \geq 0 \} \rho \{B(X_{ij})^\T \lambda \} 
-  1\{\rho^\prime \{B(X_{ij})^\T  \lambda \} < 0 \} \rho \{(\rho^\prime)^{-1}(0) \}, 
\end{eqnarray*}
where $\tilde G_{ij}(\lambda)$ is a piecewise function that is continuous and differentiable.

Since $\psi(\cdot)$ is convex, $\rho(\cdot)$ is concave, $G(\tilde \lambda_{1z} + \Delta)$ is also convex in $\Delta$. Moreover, $G(\tilde \lambda_{1z} + \Delta)$ is also continuous in $\Delta$. Therefore, to show $\Pr(\Delta^\dag \in C) \rightarrow 1$, it is enough to show that as $n \rightarrow \infty$
\begin{eqnarray}
\Pr \left( \inf_{\Delta \in \partial \mathcal{C}} G(\tilde \lambda_{1z} + \Delta) - G(\tilde \lambda_{1z}) > 0   \right) \rightarrow 1,
\end{eqnarray}
where $\partial \mathcal{C}$ is the boundary set of $\mathcal{C}$ given by $\partial \mathcal{C} = \{\Delta \in \R^K:  \Vert \Delta \Vert_2 = r \}$.

Now, fix $\Delta \in \partial \mathcal{C}$.
Apply multivariate Taylor's theorem, we have for some intermediate $\tilde {\tilde \lambda}$,
\begin{eqnarray}
&& G(\tilde \lambda_{1z} + \Delta) - G(\tilde \lambda_{1z}) \nonumber \\
&=& \sum_{ij \in \pq} \left[ - \indz S_{ij} 1\{\rho^\prime \{B(X_{ij})^\T \tilde \lambda_{1z}\} \geq 0 \} \rho^\prime \{B(X_{ij})^\T \tilde \lambda_{1z}\}  B(X_{ij})^\T \right]  \Delta \nonumber  \\
&& +  \Delta^\T \sum_{ij \in \pq} \left[ - \indz S_{ij} 1\{\rho^\prime \{B(X_{ij})^\T \tilde {\tilde \lambda}\} \geq 0 \} \rho^{\prime\prime} \{B(X_{ij})^\T \tilde {\tilde \lambda}  \}B(X_{ij}) B(X_{ij})^\T \right] \Delta/2 \nonumber  \\
&& + (B^*) ^\T \Delta + (|\tilde \lambda_{1z} + \Delta| - |\Delta| )^\T \delta \nonumber \\
&\geq& -\Vert \Delta \Vert_2 \left\Vert  
\sum_{ij \in \pq} \left[ - \indz S_{ij} 1\{\rho^\prime \{B(X_{ij})^\T \tilde \lambda_{1z}\} \geq 0 \} \rho^\prime \{B(X_{ij})^\T \tilde \lambda_{1z}  \} B(X_{ij}) \right] + B^* \right\Vert_2 \nonumber \\
&& + (\Delta^\T \underline{M} \Delta)/2  - |\Delta|^\T \delta,
\label{eq:G_diff}
\end{eqnarray}
where $\underline{M} = \sum_{ij \in \pq} \left[ - \indz S_{ij} 1\{\rho^\prime \{B(X_{ij})^\T \tilde {\tilde \lambda} \} \geq 0 \} \rho^{\prime\prime} \{B(X_{ij})^\T \tilde {\tilde \lambda}  \} B(X_{ij}) B(X_{ij})^\T \right]$.
The last inequality is due to the Cauchy--Schwarz inequality (for the first two terms) and the triangle inequality (for the third term).
Since $\Vert \Delta \Vert_2 = r$, by Cauchy--Schwarz we have
\begin{eqnarray*}
&& G(\tilde \lambda_{1z} + \Delta) - G(\tilde \lambda_{1z})\\
& \geq & - r  \left\Vert  
\sum_{ij \in \pq} \left[ - \indz S_{ij} 1\{\rho^\prime \{B(X_{ij})^\T \tilde \lambda_{1z}\} \geq 0 \} \rho^\prime \{B(X_{ij})^\T \tilde \lambda_{1z} \} B(X_{ij}) \right] + B^* \right\Vert_2   \\
&& - r\Vert \delta \Vert_2 +(\Delta^\T \underline{M} \Delta)/2 
\end{eqnarray*}
Moreover,
\begin{eqnarray*}
&& \left\Vert  
\sum_{ij \in \pq} \left[ - \indz S_{ij}  1\{\rho^\prime \{B(X_{ij})^\T \tilde \lambda_{1z}\} \geq 0 \} \rho^\prime \{B(X_{ij})^\T \tilde \lambda_{1z}  \} B(X_{ij}) \right] + B^* \right\Vert_2 \\
&=& \left\Vert  \sumq [(n^*)^{-1} (1-S_{ij}) -  n^{-1} S_{ij}1\{Z_{ij}=z\} \tilde w(X_{ij}) ] B(X_{ij}) \right\Vert_2 \\
&=& O_p \left\{ K^{1/4} (\log K)^{1/2}n^{-1/2} \right\}.
\end{eqnarray*}
Here, the last equality follows from Lemma~\ref{lem::onestep8.5}. This combined with Assumption~\ref{cond::regularity}(e) implies that with probability tending to one,
\begin{eqnarray*}
&& G(\tilde \lambda_{1z} + \Delta) - G(\tilde \lambda_{1z})\\
& \geq & (\Delta^\T \underline{M} \Delta)/2 - r O_p \left\{ K^{1/4} (\log K)^{1/2}n^{-1/2} \right\} \\
&=& \sum_{ij \in \pq} \left[ - \indz S_{ij} 1\{\rho^\prime \{B(X_{ij})^\T \tilde {\tilde \lambda} \} \geq 0 \} \rho^{\prime\prime} \{B(X_{ij})^\T \tilde {\tilde \lambda}  \} (\Delta^\T B(X_{ij}))^2\right]/2 \\
&& - r O_p \left\{ (K^{1/4} (\log K)^{1/2}n^{-1/2} \right\} \\
&\geq& C n^{-1} \sum_{ij \in \pq} \indz S_{ij} (\Delta^\T B(X_{ij}))^2 - r O_p \left\{ K^{1/4} (\log K)^{1/2}n^{-1/2} \right\} \\
& =& C \Delta^\T \left\{ \sum_{ij \in \pp: Z_{ij}=z} n^{-1} ( B(X_{ij})) B(X_{ij}))^\T \right\} \Delta - r O_p \left\{ K^{1/4} (\log K)^{1/2}n^{-1/2} \right\} \\
& \geq & C r^2 \nu_{\min} \left\{ \sum_{ij \in \pp: Z_{ij}=z} n^{-1} ( B(X_{ij})) B(X_{ij}))^\T \right\} - r O_p \left\{ K^{1/4} (\log K)^{1/2}n^{-1/2} \right\} \\
& \geq & C^{\prime \prime} r^2  - r O_p \left\{ K^{1/4} (\log K)^{1/2}n^{-1/2} \right\} \\
& = & C^{\prime \prime} (C^*)^2 \left\{ K^{1/4} (\log K)^{1/2}n^{-1/2} \right\}^2
- C^* O_p \left\{  \left\{ K^{1/4} (\log K)^{1/2}n^{-1/2} \right\}^2 \right\} > 0.
\end{eqnarray*}
Here, the second inequality follows from Assumption~\ref{cond::regularity}(a).
The third inequality holds since for a square matrix $A$, $x^\T A x \geq \nu_{\min}(A) \Vert x \Vert_2^2$. The forth inequality follows from Lemma~\ref{lem::onestep8.6}. 
Finally, the last inequality holds for a choice of $C^*$ large enough. 
This completes the proof.
\end{proof}

\begin{proof}
[Proof of Theorem~\ref{thm::unif_convg_prob}]
Lemma~\ref{lem:lambda_diff} and inequality~\eqref{eq:wdiff_bound} together complete the proof.
\end{proof}

\begin{proof}
[Proof of Theorem~\ref{thm::consistency}]
We first show that when Assumption~\ref{cond::outcome_model} holds, $\sumz \hat{w}^\pb_{z, ij} Y_{ij}(z) \cp \E_{\dt} \{ Y(z) \}$ as $n \rightarrow \infty$ for $z \in \{0, 1\}$. For simplicity, we write $Y_{ij}(z) = m_z(X_{ij}) + \epsilon_{z,ij}$, where $m_z(x) = \E_{\dt}(Y(z) \mid X=x)$ and $\epsilon_{z,ij} = Y_{ij}(z) - m_z(x)$ with $\E_{\dt} (\epsilon_{z, ij} \mid X_{ij}) = 0$. 
Notice that by Assumption~\ref{cond::outcome_model}, $m_z(x) = B(x)^\T \tilde \lambda_{2z} $.
We have
\begin{eqnarray*}
&& \left| \sumz \hat{w}^\pb_{z,ij} Y_{ij}(z) - \E_{\dt} \{ Y(z) \} \right| \\
& \leq &  \left| {\tilde \lambda_{2z}}^\T \left\{  \sum_{ij \in \pq} \hat{w}^\pb_{z,ij} B(X_{ij}) - B^*  \right\}  \right| 
+ \left| {\tilde \lambda_{2z}}^\T B^* - \E_{\dt} \{ Y(z) \}  \right| 
 + \left|  \sumz \hat{w}^\pb_{z,ij} \epsilon_{z,ij}  \right| \\
&\leq& |\tilde \lambda_{2z}|^\T \delta + o_P(1) +  \left|  \sumz \hat{w}^\pb_{z,ij} \epsilon_{z,ij}  \right| \\
&\leq& \Vert \tilde \lambda_{2z} \Vert_2 \Vert \delta \Vert_2 + o_P(1) +  \left|  \sumz \hat{w}^\pb_{z,ij} \epsilon_{z,ij}  \right| \\
&=& o_P(1) + \left|  \sumq \indz S_{ij} \rho^\prime ({\lambda^\dag}^\T B(X_{ij}) ) 1\{ \rho^\prime ({\lambda^\dag}^\T B(X_{ij}) ) \geq 0\} \epsilon_{z,ij}   \right| 
\end{eqnarray*} 
Here, the first inequality follows from the triangle inequality.
The second inequality is due to the imbalance control constraint in the optimization problem.
The third inequality follows from Cauchy--Schwarz inequality.
The last equality is due to Assumption~\ref{cond::outcome_model} and Theorem~\ref{thm:dual}(b).
Notice that Assumption~\ref{cond::regularity}(a) ensures that
\begin{eqnarray*}
- C_0n^{-1} \sumq \indz S_{ij} \epsilon_{z,ij}
\leq
\sumq \indz S_{ij} \rho^\prime ({\lambda^\dag}^\T B(X_{ij})) \epsilon_{z,ij}  \\
\leq
C_0 n^{-1} \sumq \indz S_{ij} \epsilon_{z,ij}.
\end{eqnarray*}
Both the upper and and lower bounds converge to a constant times
\begin{eqnarray*}
\E_{\dq} (\indz S_{ij} \epsilon_{z,ij}) = \E_{\dq} \left\{ \E_{\dq}(S_{ij} \mid X_{ij}) \Pr_{\dpp}( Z_{ij=z} \mid X_{ij})  \E_{\dq}(\epsilon_{z,ij} \mid X_{ij} ) \right\} =0.
\end{eqnarray*}
Therefore, 
\begin{eqnarray*}
\left|  \sumq \indz S_{ij} \rho^\prime ({\lambda^\dag}^\T B(X_{ij}) ) 1\{ \rho^\prime ({\lambda^\dag}^\T B(X_{ij}) ) \geq 0\} \epsilon_{z,ij}  \right| = o_P(1).
\end{eqnarray*}
Now to prove Theorem~\ref{thm::consistency}, it suffices to show that when Assumption~\ref{cond::inverse_prop} holds,  $\sumz \hat{w}^\pb_{z,ij} Y_{ij}(z) \cp \E_{\dt} \{ Y(z) \}$ as $n \rightarrow \infty$ for $z \in \{0, 1\}$.
\begin{eqnarray*}
&& \left| \sumz \hat{w}^\pb_{z,ij} Y_{ij}(z) - \E_{\dt} \{ Y(z) \} \right| \\
& \leq &  \left| \sumz \hat{w}^\pb_{z,ij} Y_{ij}(z) - n^{-1} \sumz \tilde{w}(X_{ij}) Y_{ij}(z) \right| + \left| n^{-1} \sumz \tilde{w}(X_{ij}) Y_{ij}(z) - \E_{\dt} \{ Y(z) \} \right| \\
& \leq& \sup_{x \in \X} \left| n\hat{w}^\pb_{z,ij}  - \tilde{w}(x)\right| \left\{ n^{-1} \sumz |Y_{ij}(z)| \right\}  + o_P(1)\\
&=& o_P(1).
\end{eqnarray*}
The last equality follows from Theorem~\ref{thm::unif_convg_prob}, and this completes the proof.
\end{proof}

\subsubsection{Proof of Theorem~\ref{thm:clt}}
\begin{proof}
[Proof of Theorem~\ref{thm:clt}]
All probabilities and expectations in this proof are computed with respect to the probability measure
$\dq$. We first decompose the estimator $\hat \tau^\pb - \tau$ as
\begin{eqnarray}
\hat\tau^\pb_\id - \tau = R_0 + R_1 + R_2 + \tilde R_1 + \tilde R_2,
\end{eqnarray}
where
\begin{eqnarray*}
R_0 &=& (n^*)^{-1} \sumq (1 - S_{ij}) \left\{ {\tilde \lambda_{21}}^\T B(X_{ij}) - {\tilde \lambda_{20}}^\T B(X_{ij})  \right\} 
+ n^{-1} \sumq S_{ij} Z_{ij} \tilde w_1(X_{ij})  \left\{ Y_{ij} - {\tilde \lambda_{21}}^\T B(X_{ij})  \right\} \\
&& - n^{-1} \sumq S_{ij} (1- Z_{ij}) \tilde w_0(X_{ij})  \left\{ Y_{ij} - {\tilde \lambda_{20}}^\T B(X_{ij})  \right\}
- \tau
\\
R_1 &=& \sumq Z_{ij} S_{ij} \left[ \hat w^\pb_{1, ij} - n^{-1} \tilde w_1(X_{ij}) \right] \{ Y_{ij}(1) - {\tilde \lambda_{21}}^\T B(X_{ij}) \}, \\
R_2 &=& \sumq(Z_{ij} S_{ij} \hat w^\pb_{1,ij} - (n^*)^{-1} (1- S_{ij})) \{ 
{\tilde \lambda_{21}}^\T B(X_{ij}) \} \\
\tilde R_1 &=& \sumq (1-Z_{ij}) S_{ij} \left[ \hat w^\pb_{0,ij} - n^{-1} \tilde w_0(X_{ij}) \right] \{ Y_{ij}(0) - {\tilde \lambda_{20}}^\T B(X_{ij}) \} ,\\
\tilde R_2 &=& \sumq((1 - Z_{ij}) S_{ij} \hat w^\pb_{0,ij} - (n^*)^{-1} (1-S_{ij}) ) \{ {\tilde \lambda_{20}}^\T B(X_{ij}) \},
\end{eqnarray*}
We can re-write
\begin{eqnarray*}
R_0 &=& (n^*)^{-1} \sumpt  \left\{ {\tilde \lambda_{21}}^\T B(X_{ij}) - {\tilde \lambda_{20}}^\T B(X_{ij})  \right\} 
+ n^{-1} \sump Z_{ij} \tilde w_1(X_{ij})  \left\{ Y_{ij} - {\tilde \lambda_{21}}^\T B(X_{ij})  \right\} \\
&& - n^{-1} \sump  (1- Z_{ij}) \tilde w_0(X_{ij})  \left\{ Y_{ij} - {\tilde \lambda_{20}}^\T B(X_{ij})  \right\}
- \tau.
\end{eqnarray*}
Central limit theorem and Slutsky's theorem ensure that as $n \rightarrow \infty$, we have
\begin{eqnarray*}
\sqrt{n} R_0 \cd \N (0, V)
\end{eqnarray*}
with
\begin{eqnarray}
\label{eq::var}
V &=&  \var_{\dt} \left(  {\tilde \lambda_{21}}^\T B(X) - {\tilde \lambda_{20}}^\T B(X)   \right)/{\alpha} \nonumber \\
&& + \var_{\dpp} \left( Z \tilde w_1(X)  \left\{ Y - {\tilde \lambda_{21}}^\T B(X)  \right\} -(1- Z) \tilde w_0(X)  \left\{ Y - {\tilde \lambda_{20}}^\T B(X)  \right\} \right).
\end{eqnarray}
Recall $\alpha  > 0$ is the odds of not being selected into the study compared to the odds of being selected  ($\alpha = \Pr_{\mathbb{Q}}(S=0)/\Pr_{\mathbb{Q}}(S=1)$). Furthermore, we have
\begin{eqnarray*}
&& \var_{\dt} \left(  {\tilde \lambda_{21}}^\T B(X) - {\tilde \lambda_{20}}^\T B(X)   \right) \\
&=& \var_{\dt} \left(  \E_{\dt} \left[ Y(1) - Y(0) \mid X\right]   \right), \\
&& \var_{\dpp} \left( Z \tilde w_1(X)  \left\{ Y - {\tilde \lambda_{21}}^\T B(X)  \right\} -(1- Z) \tilde w_0(X)  \left\{ Y - {\tilde \lambda_{20}}^\T B(X)  \right\} \right)\\
&=& \E_{\dpp} \left( \left( Z \tilde w_1(X)  \left\{ Y - {\tilde \lambda_{21}}^\T B(X)  \right\} -(1- Z) \tilde w_0(X)  \left\{ Y - {\tilde \lambda_{20}}^\T B(X)  \right\} \right)^2 \right)\\
&=& \E_{\dpp} \left( Z^2 \tilde w_1(X)  \tilde w_1(X) \var_{\dpp} \left[Y(1) \mid X \right] + (1- Z)^2 \tilde w_0(X) \tilde w_0(X)  \var_{\dpp} \left[Y(0) \mid X \right]  \right)\\
&=& \E_{\dt} \left(   \tilde w_1(X) \var_{\dt} \left[Y(1) \mid X \right] +  \tilde w_0(X)  \var_{\dt} \left[Y(0) \mid X \right]  \right).
\end{eqnarray*}
Therefore, we can further write the variance term as
\begin{eqnarray*}
V = \var_{\dt} \left(  \E_{\dt} \left[ Y(1) - Y(0) \mid X\right]   \right)/\alpha + \E_{\dt} \left(   \tilde w_1(X) \var_{\dt} \left[Y(1) \mid X \right] +  \tilde w_0(X)  \var_{\dt} \left[Y(0) \mid X \right]  \right).
\end{eqnarray*}

We first show that $\sqrt{n} R_1 \cp 0$ as $n \rightarrow \infty$.
It suffices to show that $\E_{\dq} (n R_1^2) \rightarrow 0$ due to Markov inequality $\Pr_{\dq}(\sqrt{n} |R_1| \geq a) \leq \E_{\dq} (n R_1^2)/a^2$.
We observe that $\hat w_{ij} \indep (Y_{ij}(0), Y_{ij}(1)) \mid (X_{ij})$ because of Assumption~\ref{cond::1} and the fact that $\hat w_{ij}$ is a function of $(X_{ij})$ and $(Z_{ij})$. 
\begin{eqnarray*}
\E_{\dq} (n R_1^2) & = &
\frac{1}{n} \E_{\dq} \left( \left\{ \sumq Z_{ij} S_{ij} \left[n \hat w^\pb_{1,ij} - \tilde w_1(X_{ij}) \right] \epsilon_{1,ij} \right\}^2  \right) \\
&=& \frac{1}{n} \E_{\dq} \left( \E_{\dq} \left( \left\{ \sumq Z_{ij} S_{ij} \left[n \hat w^\pb_{1,ij} - \tilde w_1(X_{ij}) \right] \epsilon_{1,ij} \right\}^2 \mid (X_{ij}) \right)  \right) \\
&=& \frac{1}{n} \E_{\dq} \left( \E_{\dq} \left(\sumq Z_{ij} S_{ij} \left[n \hat w^\pb_{1,ij} - \tilde w_1(X_{ij}) \right]^2 \epsilon_{1,ij}^2  \mid (X_{ij}) \right)  \right)\\
&\leq& \frac{1}{n} \E_{\dq} \left( \E_{\dq} \left(\sump \left[n \hat w^\pb_{1,ij} - \tilde w_1(X_{ij}) \right]^2 \epsilon_{1,ij}^2  \mid (X_{ij}) \right)  \right) \\
&=& \E_{\dq} \left( \E_{\dq} \left(  \left[n \hat w^\pb_{1,11} - \tilde w_1(X_{11}) \right]^2 \epsilon_{1,11}^2  \mid (X_{ij}) \right)  \right)\\
&=& \E_{\dq} \left( \E_{\dq} \left(  \left[n \hat w^\pb_{1,11} - \tilde w_1(X_{11}) \right]^2 \mid (X_{ij}) \right)\E_{\dq} \left(   \epsilon_{1,11}^2  \mid (X_{ij}) \right)  \right)\\
&\leq& C \E_{\dq} \left( \left[n \hat w^\pb_{1,11} - \tilde w_1(X_{11}) \right]^2 \right) \\
&=& o(1)
\end{eqnarray*}
Here, the second equality is due to the tower property of conditional expectation; 
the final equality follows from Theorem~\ref{thm::unif_convg_prob}. 
The third equality holds since by Assumption~\ref{cond::1}, for $ij \neq \tilde i \tilde j$, we have
\begin{eqnarray*}
&& \E_{\dq} \left( Z_{ij} S_{ij} Z_{\tilde i \tilde j} S_{\tilde i \tilde j} \left[n \hat w^\pb_{1,ij} - \tilde w_1(X_{ij}) \right] \left[n \hat w^\pb_{1,\tilde i \tilde j} - \tilde w_1(X_{\tilde i \tilde j}) \right] \epsilon_{1,ij} \epsilon_{1,\tilde i \tilde j}  \mid (X_{ij}) \right) \\
&=& \E_{\dq} \left( Z_{ij} S_{ij} Z_{\tilde i \tilde j} S_{\tilde i \tilde j} \left[n \hat w^\pb_{1,ij} - \tilde w_1(X_{ij}) \right] \left[n \hat w^\pb_{1,\tilde i \tilde j} - \tilde w_1(X_{\tilde i \tilde j}) \right]  \mid (X_{ij}) \right) \\
&& \E_{\dq} \left( \epsilon_{1,\tilde i \tilde j}  \mid (X_{ij}) \right) 
\E_{\dq} \left(\epsilon_{1,\tilde i \tilde j}  \mid (X_{ij}) \right) =0.
\end{eqnarray*}
Therefore, $\sqrt{n} R_1 \cp 0$ as $n \rightarrow \infty$. Similarly we can show that $\sqrt{n} \tilde R_1 \cp 0$ as $n \rightarrow \infty$.

We then show that $\sqrt{n} R_2 \cp 0$ as $n \rightarrow \infty$. Notice that
\begin{eqnarray*}
n^{1/2} |R_2| \leq \Vert \tilde \lambda_{21} \Vert_2 \left\Vert \sumt \hat w^\pb_{1,ij} B(X_{ij}) -  B^* \right\Vert_2 \leq  \Vert \tilde \lambda_{21} \Vert_2 \Vert \delta \Vert_2 = o(1).
\end{eqnarray*}
Here, the first inequality is due to Cauchy-Schwarz; the second inequality follows from the imbalance control constraints; the final equality holds by Assumption~\ref{cond::outcome_model}. Similarly, we can show  $\sqrt{n} \tilde R_2 \cp 0$ as $n \rightarrow \infty$.
This completes the proof.
\end{proof}

\subsubsection{Proof of Proposition~\ref{prop:heur_var}}
\begin{proof}
[Proof of Proposition~\ref{prop:heur_var}]
Following the proof of Proposition~\ref{prop:one-stage}, the linear mixed effect model~\eqref{eq:one-stage} can be re-written as the linear regression $\lmt \left( Y_{ij} \sim Z_{ij} + \underline{\tilde X}_{ij}  \right)$.
Let $\hat\tau^\ols$ denote the resulting OLS estimator from this regression.
Proposition~\ref{prop:one-stage} establishes that 
\begin{eqnarray*}
\hat\tau^\pb_\id &=& \sum_{ij:Z_{ij}=1} \hat w^\pb_{1,ij} Y_{ij} - \sum_{ij:Z_{ij}=0} \hat w^\pb_{0,ij} Y_{ij} \\
&=& \sump (2Z_{ij}-1) \hat w^\pb_{Z_{ij}} Y_{ij}.
\end{eqnarray*}
The corresponding OLS variance estimator is given by
\begin{eqnarray*}
\widehat{\mathrm{s.e.}} (\hat\tau^\ols) = \tilde s^2 \sump (2Z_{ij}-1)^2 \left(\hat w^\pb_{Z_{ij}}\right)^2 = \tilde s^2 \sump  \left(\hat w^\pb_{Z_{ij}}\right)^2,
\end{eqnarray*}
where $\tilde s^2$ denotes the residual variance from fitting the linear regression. Thus, it suffices to prove that $\tilde s^2 = s^2$. 
This equivalence holds if the coefficients on $\underline{\tilde X}_{ij}$ from the regression $\lmt \left( Y_{ij} \sim Z_{ij} + \underline{\tilde X}_{ij}  \right)$ are the same as those from the regression $\lmt \left( Y_{ij} - \hat\tau^\pb_\id Z_{ij} \sim  \underline{\tilde X}_{ij}  \right)$. 
Notice that the joint OLS minimization problem is
\begin{eqnarray*}
(\hat\eta^\ols, \hat\tau^\ols ) = \argmin_{\tau, \eta} \sump ( Y_{ij} - \tau Z_{ij} - \eta^\T \underline{\tilde X}_{ij} )^2.
\end{eqnarray*}
Given that $\hat\tau^\pb = \hat\tau^\ols$, it follows that
\begin{eqnarray*}
\hat\eta^\ols = \argmin_{\tau, \eta} \sump ( Y_{ij} - \tau^\pb Z_{ij} - \eta^\T \underline{\tilde X}_{ij} )^2.
\end{eqnarray*}
This completes the proof.

\end{proof}


\subsubsection{Proof of Theorem~\ref{thm::throw}}
\begin{proof}
[Proof of Theorem~\ref{thm::throw}]
To show that $\Pr_{\dpp}( \hat w^\pb(X) > 0 \mid X \in V) \rightarrow 1$, it suffices to show that $\Pr_{\dpp}( \hat w^\pb(X) = 0, X \in V) \rightarrow 0$.
Theorem~\ref{thm:dual} ensures that 
\begin{eqnarray*}
\{ \hat w^\pb(X) = 0 \} \Leftrightarrow \{ \rho^\prime\{ B(X)^\T \lambda^\dag_z \} \leq 0 \}.
\end{eqnarray*}
On the other hand, Assumption~\ref{cond::inverse_prop} ensures that
\begin{eqnarray*}
\{ X \in V \} \Leftrightarrow \{ \rho^\prime\{ B(X)^\T \tilde\lambda_{1z} \} > 0 \}.
\end{eqnarray*}
Thus,
\begin{eqnarray*}
\Pr_{\dpp}( \hat w^\pb(X) = 0, X \in V)
&\leq&
\Pr_{\dpp}(  \rho^\prime\{ B(X)^\T \lambda^\dag_z\} \rho^\prime\{ B(X)^\T \tilde\lambda_{1z}\} \leq 0) \\
&=& \Pr_{\dpp}( \rho^\prime\{ B(X)^\T  \lambda^\dag_z \} \rho^\prime\{ B(X)^\T \tilde\lambda_{1z}\} < 0),
\end{eqnarray*}
where the last equality is due to Assumption~\ref{cond::regularity}(a). 
Since Assumption~\ref{cond::regularity}(a) also requires $\rho^\prime(\cdot)$ to be strictly decreasing, we get 
\begin{eqnarray*}
&& \left\{  \rho^\prime\{ B(X)^\T  \lambda^\dag_z \} \rho^\prime\{ B(X)^\T \tilde\lambda_{1z}\} < 0 \right\} \\
&\subseteq& \left\{ |B(X)^\T  \lambda^\dag_z -  B(X)^\T \tilde\lambda_{1z}\}| > t  \right\} \cup
\left\{ \left|  B(X)^\T \tilde\lambda_{1z}\ - (\rho^\prime)^{-1}(0)\right| < t  \right\}
\end{eqnarray*}
for all $t>0$.
Notice that
\begin{eqnarray*}
\left|\lambda^\dag_z B(X) -  {\tilde\lambda_{1z} B(X)} \right| 
&\leq& \sup_{x \in \X} \Vert B(x) \Vert_2 \Vert\lambda^\dag_z - \tilde\lambda_{1z}\Vert_2 \\
&\leq& C K^{1/2} \Vert \lambda^\dag_z - \tilde\lambda_{1z}\Vert_2 = o_p(1).
\end{eqnarray*}
The first inequality follows from the Cauchy-Schwarz inequality.
The second inequality follows from Assumption~\ref{cond::regularity}(b). The final equality follows from Lemma~\ref{lem:lambda_diff} and Assumption~\ref{cond::regularity}(c).
Therefore,
\begin{eqnarray*}
&& \Pr( \hat w^\pb(X) = 0, X \in V) \\
&\leq& \Pr(|B(X)^\T  \lambda^\dag_z -  B(X)^\T \tilde\lambda_{1z} | > t) + \Pr( |   B(X)^\T \tilde\lambda_{1z}\ - (\rho^\prime)^{-1}(0)| < t)  \rightarrow 0
\end{eqnarray*}
as $t \rightarrow 0$. This completes the proof for $\Pr_{\dpp}( \hat w^\pb(X) > 0 \mid X \in V) \rightarrow 1$. Similarly, we can prove $\Pr_{\dpp}( \hat w^\pb(X) = 0 \mid X \in \supp(\dpp)\backslash V) \rightarrow 1$.

\end{proof}

\subsubsection{Proof of Theorem~\ref{thm:ad_consistency}}
We first introduce the regularity conditions for consistency of the proposed estimator $\hat\tau^\pb_\ad$ when only aggregate-level data are available. For simplicity we denote $\tau_i =\E_{\dpp}[Y(1)-Y(0) \mid G=i ]$ as the causal effect for source study population $i$.
\begin{appassumption}
\label{cond:reg_ad}
Assume the following conditions hold
\begin{itemize}
    \item [(a)] There exist constants $C_0, C_1, C_2$ with $C_0 >0$ and $C_1 < C_2 < 0$, such that $C_1 \leq m \rho^{\prime\prime}(v) \leq C_2$ for all $v$ in a neighborhood of $B(x)^\T \tilde\lambda_{1z}$. Also, $|m \rho^\prime(v)| \leq C_0$ for all $v = B(x)^\T \lambda, x \in \X, \lambda$.
    \item [(b)] $\sup_{x \in \X} \Vert B(x) \Vert_2 \leq CK^{1/2}$ and $\Vert \E_{\dq} \{ B(X)B(X)^\T\} \Vert_F, \Vert \E_{\dq} \{ \bar B_i \bar B_i^\T\} \Vert_F \leq C$ for some $C > 0$, where $\Vert \cdot \Vert_F$ denotes the Frobenius norm. $\{B(x): x \in \mathcal{X} \}$ is a convex hull.
    \item [(c)] $K = O\{ (n/m)^{\alpha_1} \}$ and $K = O\{ m^{\alpha_2} \}$ for some $0 < \alpha_1 < 1/4$, $0 < \alpha_2 < 1$.
    \item [(d)] The smallest eigenvalue, $\nu_{\min}$, of $\E_{\dq} \{ B(X)B(X)^\T\}$ satisfies $\nu_{\min} > C$ for some constant $C>0$.
    \item [(e)] $\Vert \delta\Vert_2 = O_p\left\{ \sqrt{K^3 m \log m /n} \right\}$
    \item [(f)] The ratios $n^*/n$ and $m(\inf_i n_i)/n$ approaches a finite positive constant as $n \rightarrow \infty$. $m = O(n^\alpha)$ for some $0<\alpha<1$.
    \item[(g)] $\sqrt{n_i}(\hat \tau_i - \tau)$ is sub-Gaussian with variance proxy bounded by a constant $C$.
    \item[(h)] The study-level scaling factors are uniformly bounded with $c_i \in [a_1,a_2]$ with $a_1>0$.
    \item[(i)] $\Vert \tilde\lambda_3\Vert_2, \Vert \tilde\lambda_4\Vert_2 \leq C K^{1/2}$ for some constant $C>0$.
\end{itemize}
\end{appassumption}

The proof of Theorem~\ref{thm:ad_consistency} is analogous to the proof of Thereom~\ref{thm::consistency}. We will first show that when the model of $\tilde w_{(i)}$ is correctly-specified, $\sup_i |m\hat w^\pb_{(i)} - m\tilde w_{(i)}| = o_p(1)$. 
\begin{lemma}
\label{lem:w_cnvg_ad}
Under Assumption~\ref{cond::1}--~\ref{cond:ad}, ~\ref{cond:ad_weight} and~\ref{cond:reg_ad}, the study-level weights satisfy 
$$\sup_i |m\hat w^\pb_{(i)} - m\tilde w_{(i)}| = o_p(1).$$
\end{lemma}

For simplicity denote $\tilde B_i = \E_{\dpp}[B(X) \mid G=i]$.
Let $\lambda^\dag_3$ denote the solution to the dual form of~\eqref{eq:AD}. 
Since the optimization problem~\eqref{eq:AD} can be regarded as a special case of~\eqref{eq:fr_theory_c} with every study only having one observation, from Theorem~\ref{thm:dual} we have
\begin{eqnarray}
\label{eq:ad_w}
\hat w^\pb_{(i)} = \rho^\prime\{\bar B_i \lambda_3^\dag/c_i\} 1\{ \rho^\prime\{ \bar B_i \lambda_3^\dag/c_i\} \geq 0 \}
\end{eqnarray}
This gives
\begin{eqnarray}
&& \sup_i |m\hat w^\pb_{(i)} - m\tilde w_{(i)}|  \nonumber \\
&=& \sup_i \left| m \rho^\prime\{ \bar B_i \lambda_3^\dag/c_i\} 1\{ \rho^\prime\{ \bar B_i \lambda_3^\dag/c_i\} \geq 0 \} -   m \rho^\prime\{ \tilde B_i \tilde\lambda_3 /c_i \} 1\{ \tilde B_i \tilde\lambda_3 /c_i  \geq 0 \} \right| \nonumber \\
& \leq & \sup_i \left| m \rho^\prime\{ \bar B_i \lambda_3^\dag/c_i\}  -   m \rho^\prime\{ \tilde B_i \tilde\lambda_3 /c_i \}  \right|  \nonumber \\
& \leq &  C \sup_i \left| \bar B_i \lambda_3^\dag - \tilde B_i \tilde\lambda_3   \right|  \nonumber \\
& \leq & C \sup_i \left| \bar B_i^\T (\lambda_3^\dag - \tilde\lambda_3) \right| + C \sup_i \left| (\bar B_i - \tilde B_i )^\T \tilde\lambda_3  \right|  \nonumber \\
& \leq & C K^{1/2} \Vert \lambda_3^\dag - \tilde\lambda_3 \Vert_2 + C K^{1/2} \sup_i \Vert| \bar B_i - \tilde B_i \Vert_2 
\label{eq:ad_w_bound}
\end{eqnarray}
The first equality is due to~\eqref{eq:ad_w} and Assumption~\ref{cond:ad_weight}. 
The first inequality holds as in~\eqref{eq:wdiff_bound}.
The second inequality follows from the mean value theorem and Assumption~\ref{cond:reg_ad}(a), (h).
The third inequality follows from the triangle inequality.
The last inequality is due to the Cauchy–Schwarz inequality and Assumption~\ref{cond:reg_ad}(b), (i).
Lemma~\ref{lem:ad_wdiff_term1} and~\ref{lem:ad_wdiff_term2} will bound the two terms in~\eqref{eq:ad_w_bound}, respectively.

\begin{lemma}
\label{lem:ad_wdiff_term2}
Under regularity condition~\ref{cond:reg_ad}, we have $\sup_i \Vert| \bar B_i - \tilde B_i \Vert_2  = O_p\{ \sqrt{K m \log m /n} \}$.
\end{lemma}
\begin{proof}
[Proof of Lemma~\ref{lem:ad_wdiff_term2}]
For each study $i$, Assumption~\ref{cond:reg_ad}(b) ensures that
\begin{eqnarray*}
\Vert  B(X) - \tilde B_i \Vert_2 &\leq& 2CK^{1/2}, \\
\E_{\dpp} \left[ \sum_{j=1}^{n_i} \Vert  B(X) - \tilde B_i \Vert_2^2 \mid G= i  \right] 
&\leq& n_i \E_{\dpp} \left[ \Vert B(X) \Vert_2^2 \mid G= i  \right]  \leq n_i C^2 K.
\end{eqnarray*}
Apply 
the vector Bernstein inequality
to $\sum_{j=1}^{n_i} (B(X_{ij}) - \tilde B_i)$ and we get
\begin{eqnarray*}
\Pr_{\dpp} \left( \left\Vert \sum_{j=1}^{n_i} (B(X_{ij}) - \tilde B_i) \right\Vert_2 \geq t \right) \leq 2 \exp \left( - \frac{t^2/2}{n_i C^2 K + (2CK^{1/2})t/3}\right).
\end{eqnarray*}
Rescaling by $1/n_i$ we have
\begin{eqnarray*}
\Pr_{\dpp} \left( \left\Vert \bar B_i - \tilde B_i \right\Vert_2 \geq t \right) \leq 2 \exp \left( - \frac{n_i t^2/2}{ C^2 K + (2CK^{1/2})t/(3 n_i)}\right).    
\end{eqnarray*}
By union bound we have
\begin{eqnarray*}
\Pr_{\dpp} \left( \sup_i \left\Vert \bar B_i - \tilde B_i \right\Vert_2 \geq t \right) \leq 2m \exp \left( - \frac{\inf_i n_i t^2/2}{ C^2 K + (2CK^{1/2})t/3 }\right).
\end{eqnarray*}
By Assumption~\ref{cond:reg_ad}(f), the right hand side goes to zero if $t = C^\prime \sqrt{K m \log m/n}$ for some $C^\prime>0$. Therefore, $\sup_i \Vert| \bar B_i - \tilde B_i \Vert_2 = O_p\{ \sqrt{K m \log m /n} \}$.

\end{proof}

To bound $\Vert \lambda_3^\dag - \tilde\lambda_3 \Vert_2 $ in~\eqref{eq:ad_w_bound}, we introduce Lemma~\ref{lem:ad_onestep8.5}--\ref{lem:ad_onestep8.6}, which are analogues to Lemma~\ref{lem::onestep8.5}--\ref{lem::onestep8.6}.

\begin{lemma}
\label{lem:ad_onestep8.5}
If Assumption~\ref{cond::1}--~\ref{cond:ad}, ~\ref{cond:ad_weight} and ~\ref{cond:reg_ad} hold, we have
\begin{eqnarray*}
\left\Vert  
\sum_{i=1}^m  - \tilde w_{(i)}  \bar B_i  + B^* \right\Vert_2 = O_p \left\{ K^{1/4} (\log K)^{1/2}n^{-1/2} \right\}.
\end{eqnarray*}
\end{lemma}
\begin{proof}[Proof of Lemma~\ref{lem:ad_onestep8.5}]
Following the proof structure for Lemma~\ref{lem::onestep8.5}, we will apply the inequality in Lemma~\ref{lem::mat_inequal}.
Notice that 
\begin{eqnarray*}
 \sum_{i=1}^m  - \tilde w_{(i)}  \bar B_i  + B^*
= \sumq \left\{(n^*)^{-1} - S_{ij} \frac{1}{n_i} \tilde w_{(i)}   \right\} B(X_{ij}).
\end{eqnarray*}
Denote $\underline{W}_{ij} = \left\{(n^*)^{-1} - S_{ij} \frac{1}{n_i} \tilde w_{(i)}   \right\} B(X_{ij})$, for $ij \in \pq$. By Assumption~\ref{cond:ad}(c), we have $\E_{\dq}(\underline{W}_{ij}) = 0$. Moreover, $\Vert \underline{W}_{ij} \Vert_2$ is bounded as
\begin{eqnarray*}
\Vert \underline{W}_{ij} \Vert_2 &=& \left| (n^*)^{-1} - S_{ij} (m n_i)^{-1} m \tilde w_{(i)} \right| \times \Vert B(X_{ij}) \Vert_2 \\
&=& \left| (n^*)^{-1} - S_{ij} \frac{1}{mn_i} m  \right| \times \Vert B(X_{ij}) \Vert_2 \\
&\leq& \left( \min(n^*, m \inf n_i) \right)^{-1} \left[ 1 + |m \rho^\prime(\tilde B_i^\T \tilde \lambda_3)| \right] \sup_{x \in \X} \Vert B(x) \Vert_2 \\
&\leq&  \left(\min(n^*, m \inf n_i) \right)^{-1} C K^{1/2}.
\end{eqnarray*}
The second inequality follows from Assumption~\ref{cond:reg_ad}(a)--(b).
Next, we have
\begin{eqnarray*}
&& \left\Vert \sumq \E \left(  \underline{W}_{ij}^\T \underline{W}_{ij} \right)  \right\Vert_2 \\
&=& \sumq \E \left\{  \left\{ (n^*)^{-1}(1-S_{ij}) - (mn_i)^{-1} m\tilde w_{(i)} \right\}^2 B(X_{ij})^\T B(X_{ij})  \right\} \\
&\leq& C \left(\min(n^*, m \inf n_i) \right)^{-2} (n+n^*) \tr \left[ \E_{\dq} \left\{  B(X)^\T B(X)  \right\} \right] \\
&\leq& C K^{1/2} \left(\min(n^*, m \inf n_i) \right)^{-2} (n+n^*)  \Vert \E_{\dq} \left\{  B(X)^\T B(X)  \right\} \Vert_F \\
&\leq& C^\prime K^{1/2} \left(\min(n^*, m \inf n_i) \right)^{-2} (n+n^*).
\end{eqnarray*}
The first inequality follows from Assumption~\ref{cond:ad_weight}, ~\ref{cond:reg_ad}(a). The second inequality follows from the Cauchy-Schwartz inequality. The third inequality follows from Assumption~\ref{cond:reg_ad}(b). We also have
\begin{eqnarray*}
&& \left\Vert \sumq \E \left(  \underline{W}_{ij} \underline{W}_{ij}^\T \right)  \right\Vert_2 \\
&\leq& \sumq \left\Vert \E \left\{ \left\{ (n^*)^{-1} (1-S_{ij}) - (mn_i)^{-1} S_{ij} m\tilde w_{(i)} \right\}^2 B(X_{ij}) B(X_{ij})^\T  \right\} \right\Vert_2 \\
&\leq& C \left(\min(n^*, m \inf n_i) \right)^{-2} (n+n^*)\left\Vert \E_{\dq} \left\{  B(X) B(X)^\T  \right\} \right\Vert_2\\
&\leq& C\left(\min(n^*, m \inf n_i) \right)^{-2} (n+n^*) \left\Vert \E _{\dq}\left\{  B(X) B(X)^\T  \right\} \right\Vert_F \\
&\leq& C^{\prime \prime} (\min(n^*,  m \inf n_i))^{-2}(n+n^*).
\end{eqnarray*}
Here the first inequality is due to the triangle inequality. The second inequality holds by bounding $\left\{ (1-S_{ij}) - S_{ij} m\tilde w(X_{ij}) \right\}^2$ as before.
The third inequality holds because spectral norm is dominated by the Frobenius norm.
The final inequality again follows from Assumption~\ref{cond::regularity}(b).

Combining the two results we have 
\begin{eqnarray*}
\sigma^2_{n+n^*} := \max \left\{ 
\left\Vert \sumq \E_{\dt} \left(  \underline{W}_{ij}^\T \underline{W}_{ij} \right)  \right\Vert_2,
\left\Vert \sumq \E_{\dt} \left(  \underline{W}_{ij} \underline{W}_{ij}^\T \right)  \right\Vert_2 \right\} \leq C K^{1/2} \frac{n+n^*}{(\min(n^*,  m \inf n_i))^{2}}.
\end{eqnarray*}
Applying Lemma~\ref{lem::mat_inequal}, we have
\begin{eqnarray}
\label{eq:inequal_W2}
&& \Pr \left( \left\Vert \sumq \underline{W}_{ij} \right\Vert_2 \geq t \right) \nonumber \\
&\leq& (K+1) \exp \left[ \frac{t^2/2}{CK^{1/2}\frac{n+n^*}{(\min(m \inf n_i,  n^*))^{2}} + C^\prime K^{1/2}t \frac{n+n^*}{3(\min(m \inf n_i,  n^*))^{2}} } \right]
\end{eqnarray}
By Assumption~\ref{cond::regularity}(c), the right hand side of equation~\eqref{eq:inequal_W2} goes to 0 if $$t = \bar C K^{1/4}(\log K)^{1/2} \frac{\sqrt{n+n^*}}{\min(m \inf n_i,  n^*)} \quad \text{for some $\bar C >0$.}$$ This, together with Assumption~\ref{cond:reg_ad}(f), gives $\left\Vert \sumq \underline{W}_{ij} \right\Vert_2 
= O_P( K^{1/4}(\log K)^{1/2} n^{-1/2})$.

\end{proof}

\begin{lemma}
\label{lem:ad_onestep8.6}
If Assumption~\ref{cond:reg_ad} holds,
with probability tending to one, $\nu_{\min} \left( \sum_{ij \in \pp: } m^{-1} \sumi \bar B_i \bar B_i^\T \right) > \tilde C$ for some constant $\tilde C >0$.
\end{lemma}
\begin{proof}[Proof of Lemma~\ref{lem:ad_onestep8.6}]
We follow the proof structure of \citet[Lemma 8.6]{chattopadhyay2024one}.
Let $D:= m^{-1} \sumi \bar B_i \bar B_i^\T$ and $D^* = \E_{\dpp}[m^{-1} \sumi \bar B_i \bar B_i^\T]$. 
We will first use Lemma~\ref{lem::mat_inequal} to show that $\Vert D - D^*\Vert_2 = o_p(1)$.
Denote $\underline{W}_{i} = m^{-1} \bar B_i \bar B_i^\T -  m^{-1}\E_{\dpp}[ \bar B_i \bar B_i^\T]$. By construction $\E_{\dpp}[\underline{W}_{i} ] = 0$. Moreover,
\begin{eqnarray*}
\Vert  \bar B_i \bar B_i^\T \Vert_2 \leq \Vert  \bar B_i \bar B_i^\T \Vert_F \leq C^\prime \sup_{x \in \X} \Vert B(x)\Vert_2^2 \leq C^{\prime\prime} K,
\end{eqnarray*}
where the last inequality holds due to Assumption~\ref{cond:reg_ad}(b). 
We also have
\begin{eqnarray*}
\Vert \E_{\dpp}[\bar B_i \bar B_i^\T]  \Vert_2 \leq \Vert \E_{\dpp}[\bar B_i \bar B_i^\T]  \Vert_F \leq C^\prime,
\end{eqnarray*}
where the first inequality is due to the monotonicity of the spectral norm and the last inequality holds due to Assumption~\ref{cond:reg_ad}(b).
This implies $\Vert \underline{W}_{i} \Vert_2 \leq \{C(K+1)\}/m$ for some constant $C>0$. Next, it follows that
\begin{eqnarray*}
\left\|\sumi \E[\underline{W}_i \underline{W}_i^\top] \right\|_2
&\leq& \sumi \left\| \E[\underline{W}_i \underline{W}_i^\top] \right\|_2 \\
&\leq& m^{-2} \sumi \left( CK \Vert \E_{\dpp}[\bar B_i \bar B_i^\T ] \Vert_2 + \Vert \E_{\dpp}[\bar B_i \bar B_i^\T ] \E_{\dpp}[\bar B_i \bar B_i^\T ] \Vert_2 \right) \\
&\leq& m^{-2} \sumi \left( CK \Vert \E_{\dpp}[\bar B_i \bar B_i^\T ] \Vert_2 + \Vert \E_{\dpp}[\bar B_i \bar B_i^\T ]  \Vert_2^2 \right) \\
&\leq& \{C^\prime(K+1)\}/m
\end{eqnarray*}
for some large $C^\prime >0$.
Here the first inequality is due to the triangle inequality; the second inequality is due to Assumption~\ref{cond:reg_ad}(b) and monotonicity of the spectral norm; the third inequality is
due to the submultiplicativity of the spectral norm. Since $\underline{W}_{i}$ is symmetric we also have
\begin{eqnarray*}
\left\|\sumi \E[\underline{W}_i^\T \underline{W}_i] \right\|_2 \leq \{C^\prime(K+1)\}/m.
\end{eqnarray*}
By Lemma~\ref{lem::mat_inequal} we get
\begin{eqnarray*}
\Pr \left( \left\Vert \sumi \underline{W}_{i} \right\Vert_2   \geq t \right) 
&\leq& 2K \exp \left(   \frac{t^2/2}{C^\prime(K+1)m^{-1} + C(K+1)t(2m)^{-1}}\right) \\
&\leq& 2K \exp \left(   m t^2 /\{C^{\prime\prime} k (1+t) \}\right)
\end{eqnarray*}
for a large constant $C^{\prime\prime}>0$. Notice that the RHS goes to zero for $t = \bar C \{(K \log K)/m \}^{1/2}$ with a constant $\bar C  > 0$. This implies
\begin{eqnarray*}
    \Vert D -D^*\Vert_2 = O_p\left[ \{(K \log K)/m \}^{1/2} \right] = o_p(1).
\end{eqnarray*}
The last equality is due to Assumption~\ref{cond:reg_ad}(c). Weyl’s inequality then ensures
\begin{eqnarray*}
\nu_{\min} \left( D \right)  \geq
\nu_{\min} \left( D^* \right) - \Vert D -D^*\Vert_2 \geq C - \Vert D -D^*\Vert_2,
\end{eqnarray*}
where the last inequality follows from Assumption~\ref{cond:reg_ad}(d). Since $\Vert D -D^*\Vert_2 = o_p(1)$, we have for $m$ large enough, $\nu_{\min} \left( D \right) \geq C/2 >0$.

\end{proof}

\begin{lemma}
\label{lem:ad_wdiff_term1}
Under Assumption~\ref{cond::1}--~\ref{cond:ad}, \ref{cond:ad_weight} and~\ref{cond:reg_ad},
there exists a dual solution $\lambda^\dag_3$ to~\eqref{eq:AD} s.t. 
$\Vert \lambda_3^\dag - \tilde\lambda_3 \Vert_2 = O_p\left\{ \sqrt{K^3 m \log m /n} \right\}$.
\end{lemma}
\begin{proof}[Proof of Lemma~\ref{lem:ad_wdiff_term1}]
We follow the proof structure of Lemma~\ref{lem:lambda_diff}. 

Set $r = C^* \left\{ \left\{ \sqrt{K^3 m \log m /n} \right\} \right\}$ for a sufficiently large constant $C^* >0$. Let $\Delta = \lambda - \tilde\lambda_{3}$. Also, set $\mathcal{C} = \{ \Delta \in \R^K: \Vert \Delta \Vert_2 \leq r  \}$. To show that there exists a dual solution $\lambda_3^\dag$ s.t. $\Vert \lambda^\dag_3 - \tilde\lambda_{3} \Vert_2 = O_p\left\{ \sqrt{K^3 m \log m /n} \right\}$, it suffices to show that there exists a $\Delta^\dag \in \R^K$ s.t.$\Pr(\Delta^\dag \in \mathcal{C}) \rightarrow 1$ as $n \rightarrow \infty$.

Recall from Theorem~\ref{thm:dual} the dual objective of~\eqref{eq:PBM0} is given by
\begin{eqnarray*}
g(\lambda, u) = \sum_{ij \in \pq} \left[ - S_{ij} \frac{c_i}{n_i} \rho \left\{{ (\bar B_i^\T \lambda  - u_i)}/{c_i}   \right\} \right] + (B^*) ^\T  \lambda_{3} + |\lambda_{3} |^\T \delta.
\end{eqnarray*}
We further denote $G(\lambda):= \min_{u \geq 0} g(\lambda, u)$. Since $\psi(\cdot)$ is convex, $\rho(\cdot)$ is concave, $G(\tilde \lambda_{3} + \Delta)$ is also convex in $\Delta$. Moreover, $G(\tilde \lambda_{3} + \Delta)$ is also continuous in $\Delta$. Therefore, it suffices to show that as $n \rightarrow \infty$
\begin{eqnarray}
\label{eq:S5.1}
\Pr \left( \inf_{\Delta \in \partial \mathcal{C}} G(\tilde \lambda_{3} + \Delta) - G(\tilde \lambda_{3}) > 0   \right) \rightarrow 1,
\end{eqnarray}
where $\partial \mathcal{C}$ is the boundary set of $\mathcal{C}$ given by $\partial \mathcal{C} = \{\Delta \in \R^K:  \Vert \Delta \Vert_2 = r \}$.

Similarly to~\eqref{eq:G_diff}, by multivariate Taylor's Theorem, Cauchy--Schwarz inequality and the triangle inequality, for some intermediate $\tilde{\tilde \lambda}$ we have
\begin{eqnarray*}
&& G(\tilde \lambda_{3} + \Delta) - G(\tilde \lambda_{3}) \\
&\geq& - r \left\Vert  
\sum_{ij \in \pq} \left[ - S_{ij} \frac{1}{n_i} 1\left \{\rho^\prime \left\{\bar B_i^\T \tilde \lambda_{3} / c_i \right\} \geq 0 \right\} \rho^\prime \left\{\bar B_i^\T \tilde \lambda_{3} / c_i \right\} \bar B_i \right] + B^* \right\Vert_2 \nonumber \\
&& + (\Delta^\T \underline{M} \Delta)/2  - r \Vert\delta\Vert_2,
\end{eqnarray*}
where $\underline{M} = \sum_{ij \in \pq} \left[ -  S_{ij} \frac{1}{n_i c_i} 1\{\rho^\prime \{\bar B_i^\T \tilde {\tilde \lambda}/c_i \} \geq 0 \} \rho^{\prime\prime} \{\bar B_i^\T \tilde {\tilde \lambda} /c_i  \} \bar B_i \bar B_i^\T \right]$.
Recall for simplicity we define $\tilde B_i = \E_{\dpp}[B(X) \mid G=i]$. By Cauchy-Schwartz we have
{ \small
\begin{eqnarray}
&& \left\Vert  
\sum_{ij \in \pq} \left[ - S_{ij} \frac{1}{n_i} 1\left \{\rho^\prime \left\{\bar B_i^\T \tilde \lambda_{3} / c_i \right\} \geq 0 \right\} \rho^\prime \left\{\bar B_i^\T \tilde \lambda_{3} / c_i \right\} \bar B_i \right] + B^* \right\Vert_2 \nonumber \\
&=& \left\Vert  
\sum_{i=1}^m \left[ - 1\left \{\rho^\prime \left\{\bar B_i^\T \tilde \lambda_{3} / c_i \right\} \geq 0 \right\} \rho^\prime \left\{\bar B_i^\T \tilde \lambda_{3} / c_i \right\} \bar B_i \right] + B^* \right\Vert_2 \nonumber \\
&\leq&  \left\Vert  
\sum_{i=1}^m \left[ - 1\left \{\rho^\prime \left\{\tilde B_i^\T \tilde \lambda_{3} / c_i \right\} \geq 0 \right\} \rho^\prime \left\{\tilde B_i^\T \tilde \lambda_{3} / c_i \right\} \bar B_i \right] + B^* \right\Vert_2 \nonumber \\
&& + \left\Vert  
\sum_{i=1}^m  - 1\left \{\rho^\prime \left\{\bar B_i^\T \tilde \lambda_{3} / c_i \right\} \geq 0 \right\} \rho^\prime \left\{\bar B_i^\T \tilde \lambda_{3} / c_i \right\} \bar B_i + 1\left \{\rho^\prime \left\{\tilde B_i^\T \tilde \lambda_{3} / c_i \right\} \geq 0 \right\} \rho^\prime \left\{\tilde B_i^\T \tilde \lambda_{3} / c_i \right\} \bar B_i \right\Vert_2 \nonumber \\
&\leq& \left\Vert  
\sum_{i=1}^m  - \tilde w_{(i)}  \bar B_i  + B^* \right\Vert_2 + \sum_{i=1}^m  \left| 
\rho^\prime \left\{\tilde B_i^\T \tilde \lambda_{3} / c_i \right\}   -  \rho^\prime \left\{\bar B_i^\T \tilde \lambda_{3} / c_i \right\} \right| \Vert \bar B_i \Vert_2 \nonumber \\
&\leq& \left\Vert  
\sum_{i=1}^m  - \tilde w_{(i)}  \bar B_i  + B^* \right\Vert_2 + \sum_{i=1}^m   \left |  
\rho^\prime \left\{\tilde B_i^\T \tilde \lambda_{3} / c_i \right\}   -  \rho^\prime \left\{\bar B_i^\T \tilde \lambda_{3} / c_i \right\} \right| CK^{1/2} 
\label{eq:G_diff_ad_1}
\end{eqnarray}
}
Here the first inequality follows from the triangle inequality. 
The second inequality holds by Assumption~\ref{cond:ad_weight}, the triangle inequality and the Cauchy-Schwartz inequality {\color{red} and ...}. The third inequality is due to Assumption~\ref{cond:reg_ad}(b).

Lemma~\ref{lem:ad_onestep8.5} ensures that the first part in~\eqref{eq:G_diff_ad_1} satisfies
\begin{eqnarray*}
\left\Vert  
\sum_{i=1}^m  - \tilde w_{(i)}  \bar B_i  + B^* \right\Vert_2 = O_p \left\{ K^{1/4} (\log K)^{1/2}n^{-1/2} \right\} = o_p\left\{ \sqrt{K^3 m \log m /n} \right\}.
\end{eqnarray*}
For the second part in~\eqref{eq:G_diff_ad_1}, we have 
\begin{eqnarray*}
&& \sum_{i=1}^m   \left|  
\rho^\prime \left\{\tilde B_i^\T \tilde \lambda_{3} / c_i \right\}   -  \rho^\prime \left\{\bar B_i^\T \tilde \lambda_{3} / c_i \right\} \right| \\
&=& \sum_{i=1}^m \frac{1}{m}   \left|  
m\rho^\prime \left\{\tilde B_i^\T \tilde \lambda_{3} / c_i \right\}   - m \rho^\prime \left\{\bar B_i^\T \tilde \lambda_{3} / c_i \right\} \right| \\
&=& \sum_{i=1}^m \frac{1}{m} \left| \frac{m}{c_i} \rho^{\prime\prime}  \left\{\tilde {\tilde B}_i^\T \tilde{ \lambda}_3 / c_i \right\} (\tilde B_i - \bar B_i)^\T \tilde \lambda_3  \right|  \\
&\leq& C \sup_i \Vert \tilde B_i - \bar B_i \Vert_2 \Vert \tilde \lambda_3 \Vert_2 = O_p\left\{ K\sqrt{ m \log m /n} \right\}.
\end{eqnarray*}
The second equality holds by multivariate Taylor's theorem for some intermediate $\tilde {\tilde B}_i$. The inequality is due to Cauchy-Schwartz inequality and Assumption~\ref{cond:reg_ad}(a) and (h).
The final equation follows from Assumption~\ref{cond:reg_ad}(i) and Lemma~\ref{lem:ad_wdiff_term2}.
Therefore, we get
\begin{eqnarray*}
\left\Vert  
\sum_{ij \in \pq} \left[ - S_{ij} \frac{1}{n_i} 1\left \{\rho^\prime \left\{\bar B_i^\T \tilde \lambda_{3} / c_i \right\} \geq 0 \right\} \rho^\prime \left\{\bar B_i^\T \tilde \lambda_{3} / c_i \right\} \bar B_i \right] + B^* \right\Vert_2 
= O_p\left\{ \sqrt{K^3 m \log m /n} \right\}.    
\end{eqnarray*}
This combined with Assumption~\ref{cond:reg_ad}(e) implies that with probability tending to one,
\begin{eqnarray*}
&& G(\tilde \lambda_{3} + \Delta) - G(\tilde \lambda_{3})\\
& \geq & (\Delta^\T \underline{M} \Delta)/2 - r O_p \left\{  K \sqrt{m \log m/n} \right\} \\
&=& \sum_{ij \in \pq} \left[ - S_{ij} \frac{1}{n_i c_i} 1\{\rho^\prime \{\bar B_i^\T \tilde {\tilde \lambda}/c_i \} \geq 0 \} \rho^{\prime\prime} \{\bar B_i^\T \tilde {\tilde \lambda} /c_i  \} (\Delta^\T \bar B_i)^2\right]/2 \\
&& - r O_p\left\{ \sqrt{K^3 m \log m /n} \right\} \\
&\geq& C m^{-1} \sum_{i=1}^m (\Delta^\T \bar B_i)^2 - r  O_p\left\{ \sqrt{K^3 m \log m /n} \right\} \\
& =& C \Delta^\T \left\{ m^{-1} \sumi \bar B_i \bar B_i^\T \right\} \Delta - r O_p\left\{ \sqrt{K^3 m \log m /n} \right\} \\
& \geq & C r^2 \nu_{\min} \left\{ m^{-1} \sumi \bar B_i \bar B_i^\T \right\} - r O_p O_p\left\{ \sqrt{K^3 m \log m /n} \right\} \\
& \geq & C^{\prime \prime} r^2  - r O_p\left\{ \sqrt{K^3 m \log m /n} \right\} \\
& = & C^{\prime \prime} (C^*)^2 \left\{ K^3 m \log m /n \right\}
- C^* O_p \left\{  K^3 m \log m /n \right\} > 0.
\end{eqnarray*}
Here, the second inequality follows from Assumption~\ref{cond:reg_ad}(a) and (h).
The third inequality holds since for a square matrix $A$, $x^\T A x \geq \nu_{\min}(A) \Vert x \Vert_2^2$. 
The forth inequality follows from Lemma~\ref{lem:ad_onestep8.6}. 
Finally, the last inequality holds for a choice of $C^*$ large enough. 
This completes the proof.
\end{proof}

\begin{proof}[Proof of Lemma~\ref{lem:w_cnvg_ad}]
Recall from~\eqref{eq:ad_w_bound} we have
\begin{eqnarray*}
\sup_i |m\hat w^\pb_{(i)} - m\tilde w_{(i)}| \leq C K^{1/2} \Vert \lambda_3^\dag - \tilde\lambda_3 \Vert_2 + C K^{1/2} \sup_i \Vert| \bar B_i - \tilde B_i \Vert_2 .
\end{eqnarray*}
Lemma~\ref{lem:ad_wdiff_term1} and~\ref{lem:ad_wdiff_term2}, together with Assumption~\ref{cond:reg_ad}(c), ensure that $\sup_i |m\hat w^\pb_{(i)} - m\tilde w_{(i)}| = o_p(1)$
\end{proof}

\begin{proof}[Proof of Theorem~\ref{thm:ad_consistency}]
We first show that when Assumption~\ref{cond:ad_outcome} holds, $\sumi \hat{w}^\pb_{(i)} \hat\tau_i \cp \tau$ as $n \rightarrow \infty$. For simplicity we write $\epsilon_i = \hat\tau_i - \tau_i$ with $\E_{\dq}[\epsilon_i] = 0$ by Assumption~\ref{cond:ad}(b). We get
\begin{eqnarray*}
&& \left| \sumi\hat{w}^\pb_{(i)} \hat\tau_i - \tau \right| \nonumber\\
&\leq& \left|\tilde \lambda_4^\T \left\{ \sumi \hat{w}^\pb_{(i)} \tilde B_i - B^* \right\} \right| + \left| \tilde\lambda_4 B^* - \tau \right| + \left|\sumi \hat{w}^\pb_{(i)} \tilde\lambda_4^\T  \left\{ \bar B_i - \tilde B_i  \right\} \right|   + \left| \sumi\hat{w}^\pb_{(i)} \epsilon_i \right|  \nonumber\\
&\leq& \Vert \tilde \lambda_4 \Vert_2 \Vert \delta\Vert_2 + o_p(1) \left|\sumi \hat{w}^\pb_{(i)} \tilde\lambda_4^\T  \left\{ \bar B_i - \tilde B_i  \right\} \right|   + \left| \sumi\hat{w}^\pb_{(i)} \epsilon_i \right| \nonumber\\ 
&\leq& \Vert \tilde \lambda_4 \Vert_2 \Vert \delta\Vert_2 + o_p(1) + \Vert\tilde\lambda_4\Vert_2 \sup_i \Vert \bar B_i - \tilde B_i \Vert_2     + \sup_i\left| \epsilon_i \right| \nonumber\\
&\leq& o_p(1) + \sup_i\left| \epsilon_i \right|.
\label{eq:ad_con_outcome}
\end{eqnarray*}
The first inequality follows from the triangle inequality and Assumption~\ref{cond:ad_outcome}. The second inequality is due to the imbalance control constraint in the optimization problem. 
The third inequality follows from the Cauchy-Schwarz inequality, and the constraints that $\hat{w}^\pb_{(i)}$ are non-negative and sum to one.
The last inequality holds because of Assumption~\ref{cond:ad_outcome}, Lemma~\ref{lem:ad_wdiff_term2} and Assumption~\ref{cond:reg_ad}(c), (i). We now show that $\sup_i\left| \epsilon_i \right|  = o_p(1)$. Assumption~\ref{cond:reg_ad} and uniform bound ensure that 
\begin{eqnarray*}
\Pr\left( \sup_i \sqrt{n_i}| \hat\tau_i - \tau_i| \geq t  \right) \leq 2m \exp(-t^2/2C),
\end{eqnarray*}
where the RHS goes to 0 if $t = \sqrt{2\bar C\log m}$ for some $\bar C$.
This implies $\sup_i \sqrt{n_i}| \hat\tau_i - \tau_i| = O_p(\sqrt{\log m})$. Thus,
\begin{eqnarray}
\label{eq:sup_epsi}
\sup_i |\epsilon_i| = \sup_i |\hat\tau_i - \tau_i| = O_p\left\{\sqrt{\log m/ \inf_i n_i}\right\} 
= o_p(1),
\end{eqnarray}
where 
the last equality follows from Assumption~\ref{cond:reg_ad}(f). This completes the first half of the proof.

Now to prove Theorem~\ref{thm:ad_consistency}, it suffices to show that when Assumption~\ref{cond:ad_weight} holds,  $\sumi \hat{w}^\pb_{(i)} \hat\tau_i \cp \tau$ as $n \rightarrow \infty$.
\begin{eqnarray*}
 \left| \sumi \hat{w}^\pb_{(i)} \hat\tau_i - \tau \right| 
& \leq &  \left| \sumi \tilde{w}_{(i)} \hat\tau_i - \sumi \tilde{w}_{(i)} \tau_i \right| + \left| \sumi (\tilde{w}_{(i)} -\hat{w}^\pb_{(i)})\hat\tau_i \right|\\
&\leq& \sup_i |\epsilon_i| + \sup_i |m\hat{w}^\pb_{(i)} - m\tilde{w}_{(i)}|\left| \sumi \hat\tau_i/m \right|  
= o_P(1).
\end{eqnarray*}
The first inequality follows from Assumption~\ref{cond:ad} and the triangle inequality.
The final equality is ensured by~\eqref{eq:sup_epsi} and Lemma~\ref{lem:w_cnvg_ad}. This completes the proof.

\end{proof}


\end{document}